\documentclass{article}
\usepackage{kr}

\usepackage{times}
\usepackage{soul}
\usepackage{url}
\usepackage[hidelinks]{hyperref}
\usepackage[utf8]{inputenc}
\usepackage[small]{caption}
\usepackage{graphicx}
\usepackage{amsmath}
\usepackage{amsthm}
\usepackage{amssymb}
\usepackage{booktabs}
\usepackage{algorithm}
\usepackage{paralist}
\usepackage{xspace}
\usepackage{lipsum}
\usepackage{float}
\usepackage[capitalize]{cleveref}

\author{Michael Benedikt$^1$\and Stanislav Kikot\and Johannes Marti \and Piotr Ostropolski-Nalewaja$^4$
\affiliations
$^1$University of Oxford
$^4$University of Wrocław, Poland; TU Dresden, Germany
}

\usepackage{caption}
\title{Monotone Rewritability and the Analysis of   Queries, Views, and Rules}
\usepackage{xcolor}

\definecolor{Gray}{rgb}{0.85, 0.85, 0.85}
\definecolor{DarkGray}{rgb}{0.65, 0.65, 0.65}
\definecolor{Teal}{rgb}{0, 0.5, 0.5}

\usepackage{algorithmicx}
\usepackage[noend]{algpseudocode}
\usepackage{makecell}
\usepackage{multirow}
\usepackage{hhline}

\usepackage{colortbl}

\usepackage{selectp}

\newcommand{\fullversion}{\cite{this-paper-full-version}}

\newif\ifcameraready
\camerareadyfalse 

\ifcameraready 
   \newcommand{\appwrap}[1]{the full version \fullversion} 
   \outputonly{1-11} 
\else 
   \newcommand{\appwrap}[1]{#1} 
\fi

\newcommand{\atw}{\kw{tw}^+}

\newcommand{\certainanswerrewriting}{\kw{C.A.\,Rewriting}}
\newcommand{\unfaithfultests}{\kw{UT}}
\newcommand{\ut}{\kw{ut}}

\newcommand{\pair}[1]{\langle#1\rangle}

\newcommand{\rebut}[1]{{#1}}

\newcommand{\aut}{A}

\newcommand{\codes}[2]{{#1}^\mathrm{code}_{#2}}

\newcommand{\local}{\mathsf{LOCAL}}
\newcommand{\Lang}{\mathcal{L}}
\newcommand{\mtuple}{\tuple{m}}
\newcommand{\ntuple}{\tuple{n}}
\newcommand{\tuple}[1]{\mathbf{#1}}
\renewcommand{\iff}{\quad \mbox{iff} \quad}

\newcommand{\tree}{\mathcal{T}}
\newcommand{\restr}[2]{{#1}\!\!\upharpoonright_{#2}}

\newcommand{\approxInst}{\mathcal{U}}

\newcommand{\view}{V}
\newcommand{\viewback}{\someviewback{\view}}
\newcommand{\someviewback}[1]{\mathrm{Back}_{#1}}
\newcommand{\viewinst}{\mathcal{J}}
\newcommand{\viewinstance}{\viewinst}
\newcommand{\viewschema}{\sigma_V}

\newcommand{\rinst}{\mathcal{R}}

\newcommand{\chase}{\kw{Chase}}

\newcommand{\baseschema}{\kw{B}}

\newcommand{\ysucc}{\kw{YSucc}}
\newcommand{\xsucc}{\kw{XSucc}}
\newcommand{\xproj}{\kw{XProj}}
\newcommand{\yproj}{\kw{YProj}}
\newcommand{\horproj}{\xproj}
\newcommand{\vertproj}{\yproj}

\newcommand{\ha}{\kw{HA}}
\newcommand{\va}{\kw{VA}}
\newcommand{\horadj}{\ha}
\newcommand{\vertadj}{\va}
\newcommand{\rightof}{\kw{RightOf}}
\newcommand{\downto}{\kw{DownTo}}
\newcommand{\qstart}{\kw{start}}
\newcommand{\qverify}{\kw{verify}}

\newcommand{\origin}{\kw{Origin}}

\newcommand{\zerox}{\kw{Xzero}}
\newcommand{\xzero}{\zerox}
\newcommand{\zeroy}{\kw{Yzero}}
\newcommand{\yzero}{\zeroy}
\newcommand{\init}{\kw{Init}}
\newcommand{\bottomedge}{\kw{BottomEdge}}
\newcommand{\axisy}{\kw{AxisY}}

\newcommand{\atorigin}{\kw{AtOrigin}}

\newcommand{\guard}{\kw{G}}
\newcommand{\oguard}{\kw{G}'}

\newtheorem{theorem}{Theorem}
\newtheorem{proposition}{Proposition}
\newtheorem{lemma}{Lemma}

\newtheorem{example}{Example}
\newtheorem{corollary}{Corollary}
\newtheorem{definition}{Definition}
\newtheorem{claim}{Claim}

\newcommand{\adom}{\kw{adom}}

\newcommand{\ruleof}{\kw{RuleOf}}
\newcommand{\goal}{\kw{Goal}}
\newcommand{\reach}{\kw{Reach}}

\newcommand{\code}{\kw{Code}}

\newcommand{\myeat}[1]{}

\newcommand{\tspan}{\kw{tspan}}
\newcommand{\fgdatalog}{\kw{FGDL}}
\newcommand{\fgdl}{\fgdatalog}

\newcommand{\rewriting}{\kw{R}}

\newcommand{\wrt}{w.r.t.}
\newcommand{\poslfp}{\kw{PosLFP}}
\newcommand{\struct}{{\frak M}}

\newcommand{\datalogarrow}{:=}
\newcommand{\datalogwedge}{,}

\newcommand{\canondb}{\kw{Canondb}}
\newcommand{\inst}{{\mathcal I}}
\newcommand{\jnst}{{\mathcal J}}
\newcommand{\jinst}{\jnst}

\newcommand{\vinst}{\jnst}

\newcommand{\TD}{{T\kern-1.1mm{}D}}

\newcommand{\vertices}{{\textsc{vertices}}}

\renewcommand{\fpeval}[2]{\mathsf{FPEval}(#1,#2)}
\newcommand{\doutput}[2]{\mathsf{Output}(#1,#2)}

\newcommand{\logic}[1]{\textup{\small #1}\xspace}
\renewcommand{\vec}[1]{\boldsymbol{#1}}
\newcommand{\aschema}{\mathbf{S}}
\newcommand{\schema}{\aschema}
\newcommand{\views}{\mathbf{V}}
\newcommand{\backview}{\kw{BackV}}
\newcommand{\defeq}{:=}
\newcommand{\mondet}{\kw{MonDet}}

\def\Q{{\mathcal{Q}}}

\newcommand{\cA}{\mathcal{A}}

\newcommand{\A}{\mathbb{A}}

\newcommand{\set}[1]{\left\{{#1}\right\}}

\newcommand{\dom}{{\kw{dom}}}

\newcommand{\fgtgd}{\logic{FGTGD}}

\newcommand{\complexity}[1]{\textsc{#1}}

\newcommand{\conexp}{\complexity{CoNExp}\xspace}
\newcommand{\twoexptime}{\complexity{2Exp}\xspace}
\newcommand{\twoexp}{\twoexptime}
\newcommand{\threeexp}{\complexity{3Exp}\xspace}

\newcommand{\decode}[1]{\mathfrak{D}(#1)}

\renewcommand{\varphi}{\phi}
\newcommand{\kw}[1]{\textsc{#1}}

\newcommand{\lemmagfpbwd}[1]{\hyperref[lemma:backwards-gfp]{#1}\xspace}
\newcommand{\lemmagnfpbwd}[1]{\hyperref[lemma:backwards-gnfp]{#1}\xspace}

\newcommand{\myparagraph}[1]{{\bf #1.}}
\newcommand{\frone}{\kw{FR-1}}

\usepackage{subfigure}

\begin{document}
\maketitle
\begin{abstract}
We study the interaction of views, queries, and background knowledge in the form of existential rules.
The motivating questions concern monotonic determinacy of a query using views w.r.t.  rules,
which refers to the ability 
to  recover the query answer from the views via a monotone function. We study the decidability of monotonic determinacy,
and compare with variations that require  the ``recovery function'' to be in a well-known monotone query language, such as conjunctive queries 
or Datalog.
Surprisingly, we find that even in the presence of basic existential rules, the borderline between well-behaved and badly-behaved answerability
differs radically from the unconstrained case. In order to understand this boundary, we require new results concerning 
entailment problems involving views and rules.
\end{abstract}

\vspace{-2mm}
\section{Introduction} \label{sec:intro}
Views are a mechanism for defining an interface to a set of datasources. In the context of relational databases,
views allow one to associate a complex transformation -- a view definition -- to a new
relation symbol, indicating that the symbol will stand for the output of the transformation. 
Given a set of views $\views$ and an instance $\inst$ of the original schema, an external user will have access to
the \emph{view image of $\inst$}, $\views(\inst)$, in which each view symbol is interpreted by evaluating the corresponding
definition on $\inst$.
 Views have many uses, including   privacy and  query optimization \cite{afratichirkovabook}. 

Our work is motivated by the question of rewriting queries using views. We have a set of views $\views$ defined over some ``base schema'', a query $Q$ over the same
schema, and we want to \emph{rewrite  $Q$ using $\views$}: find a function $R$ on the  view schema such that applying $R$ to the output of the
views definitions will implement $Q$. We will be interested in the setting where $Q$ and the view definitions are \emph{monotone queries} -- defined
in  
languages that do not have negation or difference. Conjunctive queries (CQs), their unions (UCQs), and
Datalog are query languages that define only monotone queries. In this case, it is natural to desire a rewriting that is 
a monotone function of the view images. We write $\views(\inst) \subseteq \views(\inst')$ if for every view definition
in $\views$, its evaluation on $\inst$ is contained in its evaluation on $\inst'$.
Then $Q$ is \emph{monotonically determined over $\views$} \cite{perezthesis,NSV,calvanese2007view} if for every $\inst$, $\inst'$ with  $\views(\inst)\subseteq \views(\inst')$, 
 $\doutput{Q}{\inst} \subseteq \doutput{Q}{\inst'}$.
A query $Q$ may not have any rewriting using the views $\views$ over all datasources, but it may be well-behaved for the data of interest
to an application. 
Here we will consider datasource restrictions given by standard classes of knowledge in the form of \emph{existential rules}:\footnote{Existential rules are also called \emph{tuple-generating dependencies}~\cite{AHV}, \emph{conceptual graph rules}~\cite{SalMug96}, Datalog$^{\pm}$~\cite{datalogpmj}, and \emph{$\forall \exists$-rules}~\cite{Baget-LMS-walking-decidability-line} in the literature.}
rules with existential quantifiers in the head.
Given a set of  views $\views=\view_1 \ldots \view_k$, query $Q$, and rules $\Sigma$, we say that
\emph{$Q$ is monotonically determined over $\views$ w.r.t. $\Sigma$} if for $\inst, \inst'$ satisfying $\Sigma$ we have that
$\views(\inst) \subseteq \views(\inst')$ implies $\doutput{Q}{\inst} \subseteq \doutput{Q}{\inst'}$.

We say a query $R$ over the view schema is a \emph{rewriting} of $Q$ with respect to $\Sigma$ if for every $\inst$ satisfying $\Sigma$ we have that $R(\views(\inst)) = \doutput{Q}{\inst}$ where $\views(\inst)$ is the view image of $\inst$.
\begin{example} \label{ex:constraints}
Suppose our base signature has binary relations $R$ and $S$, and let $Q$ be $\exists x  ~ R(x,x)$.
Consider the view $V(x,y) := R(x,y) \vee S(x,y)$. It is easy to see that $Q$ is  not rewritable over
$V$, monotonically or otherwise: given  a pair in $V$, we cannot tell if it is in $R$ or in $S$.
But suppose we know that our schema satisfies the rule $\forall x  ~ (S(x,x) \rightarrow R(x,x))$.
Then $Q$ has a simple rewriting $\exists  x  ~ V(x,x)$.
\end{example}

There are two natural questions about monotonic determinacy:
\begin{inparaenum}
\item \emph{Decidability}: Fixing languages for specifying queries, views, and rules, can we decide monotonic determinacy?
\item \emph{Required expressiveness of rewritings}: 
can we say that
when  $Q$ is monotonically determined by $\views$ w.r.t. $\Sigma$ there is  necessarily a rewriting in a restricted monotone language 
(e.g. CQ, UCQ, Datalog)?
\end{inparaenum}

The monotonicity requirement is motivated by the fact that when we begin with a monotone query, we expect a monotone
rewriting, and also  by  the need
to make  the two questions above manageable.  If we do not impose that a rewriting is monotone,
it is known that the behavior of even very simple queries and views -- CQs--  is badly-behaved, even in the absence
of rules.  We cannot decide whether a CQ query has a rewriting over CQ views, and in cases where such a rewriting exists we can say basically
nothing about how complex it may need  to be \cite{redspider,rainworm}.
In contrast, it is known that monotonic determinacy behaves very well for CQ queries and views: decidable in NP \cite{NSV,lmss}. And
when rewritings exist, they can be taken to be CQs \cite{NSV}. These results have been shown to extend to many classes of rules
\cite{thebook}.

Monotone determinacy is well understood for first-order queries and views, in the absence of rules. It is decidable in the case of 
UCQ views and queries, and whenever one has monotonic determinacy, one has a rewriting that is a UCQ \cite{thebook}. 
Prior work also
investigated the situation when views and queries are \emph{monotone and recursive}, either as \emph{regular path queries} or 
\emph{Datalog}
\cite{losslessregular,calvanese2007view,determinacyregularpath,uspods20journal}. Recall that Datalog is a standard language for defining recursive monotone queries.
For example, the  Datalog query $Q_R$ consisting of the following two rules computes the pairs $(x,y)$ where $x$ reaches $y$ via edges in a binary relation $R$:
\vspace{-1mm}
\begin{align*}
\reach(x,y) \datalogarrow& R(x,y)\\  \reach(x,y) \datalogarrow& R(x,z), \reach(z,y) 
\end{align*}
\vspace{-5mm}

\noindent
Let us give a rough summary of the situation for Datalog, based
on the most recent work \cite{uspods20journal}. Monotonic determinacy does not behave well for general Datalog queries and views -- it is undecidable, and when
it holds one may need arbitrarily complex rewriting. At a high level,  two well-behaved paradigms for monotonic determinacy  have been identified:
1) \emph{non-recursive queries (i.e. UCQ) and general Datalog views}, and 2) \emph{both views and queries guarded}.
 Specifically, for decidability, we have:
\begin{compactitem}
\item 
\emph{The query being non-recursive suffices for good behaviour}:
The query  $Q$ being an (unguarded) CQ or even a UCQ, suffices for decidability,
even
for  general Datalog views;

\item \emph{Guarded queries and views behave well}:
If the views and queries are both in the standard ``guarded recursive language'', frontier-guarded Datalog, 
monotonic determinacy is decidable, even if there is recursion in both;

\item \emph{One case of a recursive query and non-guarded views behaves well:}
If the query is recursive,  monotonic determinacy 
behaves well
if the query is ``very guarded''-- Monadic Datalog (MDL) --  and the views are CQs.

\item \emph{Recursive queries and unguarded views are a problem}:  If the views are general (unguarded) UCQs and we allow
recursion in the query $Q$ (even MDL), monotonic determinacy is undecidable.
\end{compactitem}

All the cases above that are well-behaved for decidability are also well-behaved for expressiveness of rewritings: whenever monotonic determinacy
holds there is a rewriting in Datalog.
But there are additional ''tame'' cases for rewritability; 
\emph{when a Datalog query is monotonically determined in terms of CQ views or guarded views,  there is a rewriting in Datalog}. That is, 
for views that are CQs or in guarded Datalog, we do not need to restrict the Datalog query $Q$ to be sure that monotonic determinacy is witnessed by a Datalog rewriting.
Figures \ref{fig:nocondecide} and \ref{fig:noconrewrite} provide a detailed look at the case without any knowledge in the form of rules. In the table
Und. means undecidable, and ``n. n. Datalog'' means that rewritings are not necessarily in Datalog.

We  consider how this changes with background knowledge. Clearly, for ``arbitrary knowledge'' 
in first-order logic, nothing interesting can be said about decidability or about the necessary language for rewritings.
We will thus focus on ``tame existential rules'' such as those in the Datalog$^{\pm}$ family
\cite{datalogpm}, frontier-guarded TGDs, which are known to behave well for standard decidability questions. 
In the presence of tame existential rules, a number of new complications arise. First, we need to distinguish two notions of monotonic determinacy:
\emph{finite monotonic determinacy}, in which the rewriting must exist only for finite databases, and
\emph{unrestricted monotonic determinacy}, in which the rewriting holds for all relational structures, finite or infinite.
In the case of Datalog without rules, these two notions coincide \cite{uspods20journal}. But even in the presence of very
simple rules, they differ.
\begin{example} \label{ex:fc}
Let our schema consist of a binary relation $R$, along with  a rule: $\quad\forall x y ~ R(x,y) \rightarrow \exists z ~ R(y,z)$

This is a \emph{Linear TGD}: a universally quantified implication
with a single atom in the hypothesis. It is even a \emph{unary inclusion dependency} (UID): there are no repeated variables, and only one
variable occurs on both sides of the implication. Such rules are extremely well-behaved. 
For example,  the implication problem between them is known
to be decidable in PSPACE \cite{AHV}. Observe that in a \emph{finite} database satisfying the rule, there must be a cycle of $R$ edges.

Consider the Datalog program $Q = Q_R \cup \set{\goal \datalogarrow \reach(x,x)}$, where $Q_R$ is defined above, and
$\goal$ the goal predicate. Then $Q$ is a Boolean query checking for the existence of an $R$ cycle. 
Suppose also that we have only the Boolean view $V() := \exists x y ~ R(x,y)$.  Clearly, we cannot answer $Q$ using the views over all instances.
But over finite databases, $Q$ is equivalent to  $V()$, so it is monotonically determined over finite instances.
\end{example}

In this work, we deal mainly with \emph{unrestricted monotonic determinacy}, admittedly because it is simpler to analyze.
But even for the unrestricted case, our proofs will involve interplay between finite and infinite. And we will sometimes be able
to infer that the finite and unrestricted cases agree.
We show that the case where both views and queries are guarded is still well-behaved 
for decidability in the presence of guarded rules.
But the case of non-guarded views is significantly different. Without rules, non-recursiveness can substitute for guardedness.
But we show that \emph{even when both the query and views are non-recursive,
and even  when the rules are in a well-behaved class, monotonic determinacy is undecidable}.
We show that decidability can be regained by restricting the view and query to be CQs, and also imposing that the 
rules are extremely simple: linear TGDs, a case
analogous to SQL referential constraints.

For expressiveness of rewritings, the most significant differences introduced by rules
come for guarded views and queries. We show that for monotonically determined guarded queries and views,
we obtain rewritings in fixpoint logic, and
thus in polynomial time. In certain cases we can conclude rewritability in Datalog.

\myparagraph{Contributions}
Our first contribution is  a set of  tools  for reasoning with queries, views, and rules:

\begin{compactitem}
\item We extend the ``forward-backward'' method for obtaining rewritings to the presence of rules.

\item We provide new results on \emph{certain-answer rewritability} of Datalog queries with rules, and on \emph{finite controllability of recursive queries with rules}. 

\item We present a new kind of ``compactness property'' for entailments involving Datalog queries, rules, and views. 
It states that when an infinite view instance entails that a query holds on the base instance it is derived from,
then the same holds for some finite subinstance.
This  is crucial  to rewritability and non-rewritability results.

\end{compactitem}
While we apply these tools to get positive results about decidability and rewritings for monotonic determinacy, we   hope they
are of broader interest.

Our second major contribution is a set of
surprising undecidability results -- strongly contrasting with
the case without background knowledge. For example, monotonic determinacy for CQ queries and CQ views is undecidable even
for reasonably simple rules: combinations of frontier-one and linear.

\myparagraph{Organization}
Section \ref{sec:prelims} contains  preliminaries.
Section \ref{sec:tools} presents key tools
for our positive results.
Section \ref{sec:decide} applies the tools to get positive results on decidability and rewritability, while
Section \ref{sec:undecide} presents  the surprising undecidability results.
The paper ends with  a summary in Section \ref{sec:conc}.
Many proofs, along with some auxiliary results,  are deferred to \appwrap{the appendix}.


\begin{figure}
\begin{minipage}{0.49\textwidth}
\centering
\scalebox{0.75}{
\begin{tabular}{|c|c|c|c|c|}
\hline
\rowcolor{Gray}
\cellcolor{DarkGray}Query $\backslash$ Views & CQ            & \makecell{MDL, FGDL}      &UCQ                              & DL\\
\hline
\cellcolor{Gray}CQ,  UCQ                 & NP            & \multirow{3}{*}{$\twoexp$}&$\conexp$                        & $\twoexp$\\
\hhline{|-|-|~|-|-|}
\cellcolor{Gray}MDL                      & in $\threeexp$&                           &\multicolumn{2}{c|}{\multirow{3}{*}{Undecidable}}\\ 
\hhline{|-|-|~|~~|}
\cellcolor{Gray} FGDL                     & Open          &                           &\multicolumn{2}{c|}{}           \\
\hhline{|-|-|-~~|}
\cellcolor{Gray}DL                       &   \multicolumn{4}{c|}{}                                                 \\
\hline
\end{tabular}
}
\caption{Decidability/Complexity without rules} \label{fig:nocondecide} \hfill
\end{minipage} 
\begin{minipage}{0.49\textwidth}
\centering
\scalebox{0.75}{
\begin{tabular}{|c|c|c|c|c|} %
\hline
\rowcolor{Gray}
\cellcolor{DarkGray}
Query $\backslash$ Views & CQ & MDL, FGDL          & UCQ & DL\\

\hline
\cellcolor{Gray}
CQ                       &\multicolumn{4}{c|}{CQ}\\

\hline
\cellcolor{Gray}
UCQ                      &\multicolumn{4}{c|}{UCQ}\\

\hline
\cellcolor{Gray}
MDL                      &\multirow{2}{*}{FGDL}&MDL                 &\multicolumn{2}{c|}{\multirow{3}{*}{n. n. Datalog}}\\

\hhline{|-|~|-|~~|}
\cellcolor{Gray}
FGDL                     & &\multirow{2}{*}{DL} &\multicolumn{2}{c|}{\multirow{3}{*}{}}\\

\hhline{|-|-|~|~~|}
\cellcolor{Gray}
DL                     &\multicolumn{1}{c}{}  &                    &\multicolumn{2}{c|}{\multirow{3}{*}{}}\\
\hline
\end{tabular}
}
\caption{Rewritability without rules} \label{fig:noconrewrite}
\end{minipage} 
\end{figure}

\vspace{-5pt}
\section{Preliminaries} \label{sec:prelims}
We assume the usual notion of relational schema: a finite set of relations,
with each relation associated with a number 
its  \emph{arity}.
For predicate $R$ of arity $n$, an $R$-fact is  of the form $R(c_1 \ldots c_n)$.
A \emph{database instance} (or simply \emph{instance})
is a set of facts.
The  \emph{active domain} of an instance,   $\adom(\inst)$,  is the set of elements that occur
in some fact.
A \emph{query} of arity $n$ over schema $\aschema$
is a function from instances over $\aschema$ to relations of arity $n$.
A \emph{Boolean query} is a query of arity $0$.
The output of a query $Q$ on instance $\inst$ is  denoted as $\doutput{Q}{\inst}$.
We also use the standard notations $\inst \models Q(\vec c)$ or $\inst, \vec c \models Q$
to indicate that $\vec c$ is in the output
of $Q$ on input $\inst$.
A query is \emph{monotone} if  $\inst \subseteq \inst'$ implies that $\doutput{Q}{\inst} \subseteq \doutput{Q}{\inst'}$.
A \emph{homomorphism} from  instance  $\inst$ to instance $\inst'$ is a mapping
$h$ such that $R(c_1 \ldots c_n) \in \inst$ implies $R(h(c_1) \ldots h(c_n)) \in \inst'$.
We assume familiarity with standard \emph{tree automata} over finite trees of fixed branching depth.
A few results  use   automata 
over infinite trees \cite{thomas}. \emph{B\"uchi tree automata} have the same syntactic form as regular tree automata: 
a finite set of states,  subsets that are initial and accepting, and a transition relation. An automaton accepts an 
infinite tree if there is a run that assigns states to each vertex of the tree, obeying the transition relation,
and where every  path hits an accepting state infinitely often.


\noindent
\myparagraph{CQs and Datalog}
We assume familiarity with FO logic, including the notion of free and bound variable.
A \emph{conjunctive query} (CQ) is a
formula of the form $Q(\vec x) = \exists \vec{y}\, \phi(\vec x,  \vec y)$, where
$\phi(\vec x,  \vec y)$ is a conjunction of relational atoms.
Given any CQ $Q$, its \emph{canonical database}, denoted $\canondb(Q)$, is the instance formed 
by treating each atom of $Q$ as a fact.
The  output of a CQ $Q$ on some instance $\inst$ is the set of all tuples $\vec t$ such that there is homomorphism $h: \canondb(Q) \to \inst$ satisfying $\vec t = h(\vec x)$.
A \emph{union of conjunctive queries} UCQ  is a disjunction of CQs with the same free variables.

Datalog is a language for defining  queries over a relational schema $\aschema$.
Datalog rules are of the form:
$P(\vec x)    \datalogarrow    \phi(\vec x, \vec y)$
where $P(\vec x)$ is an atom s.t.  $P \not\in\aschema$,
and $\phi$ is a conjunction of atoms, with $\vec y$ implicitly existentially quantified.
The left side of the rule is the \emph{head}, while the right side is the \emph{body} of the rule.
In a set of rules, the 
relation symbols in rule heads  are called  \emph{intensional
database predicates} (IDBs), while relations in $\aschema$ are called  \emph{extensional
relations} (EDBs).
A \emph{Datalog program} is a finite collection of rules.
For an instance $\inst$ and  Datalog program $\Pi$, we let $\fpeval{\Pi}{\inst}$ denote the 
$\supseteq$-minimal IDB-extension of the input $\inst$ which satisfies $\Pi$ treated as a set of universal FO implications.
A \emph{Datalog query} $Q = (\Pi, \goal)$ is a Datalog program $\Pi$ together
with a distinguished intensional \emph{goal relation} $\goal$ of arity $k \ge 0$. 
The \emph{output} of  Datalog query $Q$ on an instance $\inst$ (denoted as $\doutput{Q}{\inst}$ or simply $Q(\inst)$) 
is the set $\set{\vec c \mid \goal(\vec c) \in \fpeval{\Pi}{\inst}}$.
%
\emph{Monadic Datalog} (MDL) is the fragment of Datalog
where all IDB predicates are unary.
\emph{Frontier-guarded Datalog} ($\fgdatalog$) requires that in each rule $P(\vec x)    \datalogarrow    \phi(\vec x, \vec y)$  there exists an EDB atom in $\phi$ containing $\vec x$. 
Frontier-guarded Datalog does not contain MDL; for example, in an MDL program
we can have a rule $I_1(x) \datalogarrow I_2(x)$, where $I_1$ and $I_2$
are both intensional. However every MDL program can be rewritten
  in $\fgdatalog$ \cite{uspods20journal}, and thus
we declare, as a convention, that any MDL program is FGDL.

When Datalog query $Q$ holds for tuple $\vec t$ in  instance $\inst$, this means  one of an infinite
sequence of CQ queries $Q_n$ holds. 
Here
 $Q_n$ is  obtained by  \emph{unfolding the intensional predicates by their rule bodies some number of times}. 
Such a query is called a  \emph{CQ approximation of $Q$}.
Datalog can also be seen as a subset of Least Fixpoint Logic (LFP), the extension of first-order logic
 with a least fixed point operator construct: if $\phi(x_1 \ldots x_K, X)$ is a formula
with free variables $x_1\ldots x_k$ and an additional $k$-ary second order free variable  $X$, 
$\mu_{X, \vec x} \phi(y_1 \ldots y_k)$ is a formula in which $X$ is not free, but $\vec y$ are free. 
In this work we will also consider an extension of Datalog where we have both the least fixed point operator, the dual greatest
fixpoint operator $\nu$, but no negation or universal quantification. This logic, denoted $\poslfp$, can still only express monotone queries, but
it can express properties beyond Datalog: e.g., the $\poslfp$ formula
$\nu_{P,x}.[A(x) \wedge (\exists  y~R(x,y) \wedge P(y))] (z)$ holds of  element
$z$ when it is in a unary predicate $A$ and has paths to an $A$ node of finite length.

\myparagraph{Existential rules}
Semantic relationships between relations can be described using \emph{existential rules} - also referred to as \emph{Tuple Generating Dependencies} (TGDs) - that are logical sentences of the form
    $\forall \vec x ~ \lambda(\vec x) \rightarrow \exists \vec y ~ \rho(\vec x, \vec y)$,
where
$\lambda(\vec x)$ and $\rho(\vec x,\vec y)$ are conjunctions of relational
atoms whose free variables are contained in ${\vec x}$, and ${\vec x
\cup \vec y}$ correspondingly.
The left-hand side of a TGD (i.e., the conjunction $\lambda(\vec x)$)
is the \emph{body} of the dependency, and the right-hand side is the
\emph{head}.
By abuse of notation, we often treat heads and bodies as sets of
atoms, and we commonly omit the leading universal quantifiers.
The variables that appear in both the head and the body are the \emph{frontier}
of the rule, and are also said to be the \emph{exported variables}.

Let $\inst$ be an instance and let $\tau$ be a TGD with head and body as above.
The notion of dependency $\tau$
\emph{holding} in $\inst$ (or $\inst$ \emph{satisfying} $\tau$, written ${\inst \models
\tau}$) is the usual one in first-order logic.

Given CQs $Q_1$ and $Q_2$ along with dependencies $\Sigma$, we say
$Q_1$ is contained in $Q_2$ relative to $\Sigma$ if for every
instance $\inst$ satisfying $\Sigma$, $\doutput{Q_1}{\inst} \subseteq \doutput{Q_2}{\inst}$.
We also write $\Sigma \wedge Q_1 \models Q_2$ in this case.
Given a set of facts $D$, Boolean CQ $Q$, and dependencies $\Sigma$, we
say $D$ and $\Sigma$ \emph{entails} $Q$ if $Q$ holds in every instance
containing $D$ and satisfying $\Sigma$. We also write
$\Sigma \wedge D \models Q$.
We can similarly talk about ``finite entailment'', where ``every instance containing
$D$'' is replaced by ``every finite instance containing $D$''.
For a class of queries and existential rules, we say \emph{entailment is finitely controllable} if entailment agrees
with finite entailment for each finite instance $D$ and each query and rules in the class.

\myparagraph{Special classes of rule}
A rule is \emph{linear} if the body has a single atom, and \emph{full} if there are no existential quantifiers in the head.
It is \emph{frontier-guarded} (FGTGD) if there is a body atom  that contains
all exported variables. It is \emph{frontier-one} ($\frone$) if there is only one
variable in the frontier.
A $\frone$ linear TGD with no constants or repeated variables
is a {\em Unary Inclusion Dependency} (UID).
A set of rules $\Sigma$ is \emph{Source-to-Target} if for every $\rho, \rho' \in \Sigma$ the head-predicates of $\rho$ do not appear in the body of $\rho'$.


\myparagraph{The chase}
In certain arguments we  use
the characterization of logical entailment between CQs in the presence
of existential rules in terms of the chase procedure \cite{maier,fagindataex}.
The chase modifies an instance by a sequence of \emph{chase steps} until all
dependencies are satisfied. Let $\inst$ be an instance,  and consider a TGD
 $\tau=\forall \vec x. \lambda  \rightarrow \exists \vec y. \rho$.
 Let  $h$ be a \emph{trigger} -- 
a
homomorphism from $\lambda$ into $\inst$. Performing a chase step for
$\tau$ and $h$ to $\inst$ extends $\inst$ with each fact of the conjunction ${h'(\rho(\vec
x,\vec y))}$, where $h'$ is a substitution such that ${h'(x_i) = h(x_i)}$ for
each variable ${x_i \in \vec x}$, and $h'(y_j)$, for each ${y_j \in \vec y}$,
is a fresh labeled null that does not occur in $\inst$. 
For $\Sigma$ a set of TGDs  and $\inst$ an instance, we use $\chase_\Sigma(\inst)$ to denote
any instance (possibly infinite) formed from $\inst$ by iteratively applying chase steps, where
$\Sigma$ holds. We say that \emph{$\chase_\Sigma(\inst)$ is finite} if we can perform finitely many chase
steps and obtain an instance satisfying $\Sigma$. We say that $\Sigma$ \emph{has terminating chase} if $\chase_\Sigma(\inst)$ is finite
for every finite $\inst$.

\myparagraph{Views and rewritability} 
A \emph{view} over some relational schema $\aschema$ 
is a tuple $(V, Q_V)$ where $V$ is a view relation and $Q_V$ is an associated 
query over $\aschema$ whose arity  matches that of $V$. $Q_V$ is referred to as the \emph{definition} of view $V$.
By $\views$ we denote a collection of views over a schema
 $\aschema$. We sometimes refer to the vocabulary of the definitions $Q_V$ as the \emph{base schema} for $\views$, denoting
it as $\Sigma_{\baseschema}$, while the predicates  components $V$ are referred to as the \emph{view schema},
denoted $\Sigma_{\views}$.
For an  instance $\inst$ and set of views $\views= \{(V, Q_V) \mid V \in \Sigma_\views \}$, the \emph{view image} of $\inst$, 
denoted by $\views(\inst)$, 
is the instance over $\Sigma_\views$ where each view predicate $V \in \Sigma_\views$ is  interpreted by $\doutput{Q_V}{\inst}$.
Recall from the introduction that query $Q$ is monotonically determined over views $\views$ relative to rules $\Sigma$ if
for each $\inst, \inst'$ satisfying $\Sigma$ with $\views(\inst) \subseteq \views(\inst')$, $Q(\inst) \subseteq Q(\inst')$.
Given views $\views$, rules $\Sigma$ and a query $Q$, a query $\rewriting$ over the view schema $\Sigma_\views$ is a 
\emph{rewriting}
of $Q$ with respect to $\views$ and $\Sigma$ if:
for each $\inst$ over $\aschema$ satisfying $\Sigma$, the output of $\rewriting$ on $\views(\inst)$ 
is the same as the output of $Q$ on $\inst$.
A rewriting  that can be specified in a particular
language $L$ (e.g. Datalog, CQs) is an \emph{$L$}-rewriting of $Q$ w.r.t. $\views$ and $\Sigma$,
and if this exists we say $Q$ is \emph{$L$-rewritable} over $\views, \Sigma$. We drop $\views$ and/or $\Sigma$ from the notation when it
is clear from context.

It is clear that if $Q$ has a rewriting in a  language that defines only monotone queries,
like Datalog, then $Q$ must be monotonically determined. We will be concerned with the converse
to this question.
The main questions we will consider, fixing languages  $L_Q$, $L_\views$, $L_\Sigma$  for the queries, views, and rules (e.g.
Datalog, fragments of Datalog for the first two, special classes of existential rules for the third) are:
\begin{compactitem}
\item can we decide whether a $Q$ in $L_Q$ is monotonically determined over $\views, \Sigma$ in $L_\views, L_\Sigma$?
\item fixing another language $L$ for rewritings,
if $Q$ is monotonically determined over $\views, \Sigma$, does it necessarily have a rewriting
in $L$?
\end{compactitem}

\myparagraph{Finite and unrestricted variants}
We have noted in the introduction that there are variants of our definitions of monotonic
determinacy and rewritability depending on whether only finite instances are considered, or all instances.
In the definitions above, we make unrestricted instances the default.
We say that query $Q$ is monotonically determined over views $\views, \Sigma$ \emph{over
finite instances} if in the definition ``for any each instances $\inst, \inst'$ satisfying \ldots''
is replaced by ``for each finite instances $\inst, \inst'$ satisfying \ldots''.
We similarly arrive at the notion of  a query $\rewriting$
being a rewriting over $\views, \Sigma$ for finite
instances by changing ``for each $\inst$ over $\aschema$'' in the definition of
rewriting to ``for each finite $\inst$ over $\aschema$''. If the finite and unrestricted versions coincide, we say that \emph{monotonic determinacy is finitely controllable}
for a class of view, queries, and rules.
In \cite{uspods20journal} finite controllability was observed in the absence of rules,
while Example \ref{ex:fc} shows that it fails even for simple rules.

\myparagraph{Monotonic determinacy characterized using approximations and the chase}
Figure \ref{alg:query-mondet-constraints} gives an abstract procedure for checking monotonic determinacy of Boolean Datalog query $Q$ over
Datalog views $\views$ with respect to rules $\Sigma$. We universally choose an approximation $Q_n$ of $Q$, and then chase the canonical database of
this query  with $\Sigma$. One can think of each of the results as a generic instance satisfying $Q$ and $\Sigma$. We let $\viewinst_n$ be the 
view image of such  a database: thus $\viewinst_n$ is a set of view facts. We then take each fact $\views_i(\vec c)$ in $\viewinst_n$, where $\views_i$ is  a view predicate,
and \emph{non-deterministically choose witnesses facts for it}. For each $\views_i(\vec c)$ we choose  an approximation for the query
$Q_{\views_{i}}$ associated with $\views_i$, and then take the canonical database of this with $\vec c$ substituted in for the
free variables. For an instance $\viewinst$ of the view schema,
we let $\backview_{\views}(\viewinst)$ be the instances of the base schema that results from this process.
In the case that the views are CQs, we can assume there is only one instance in $\backview_{\views}(\viewinst)$, namely the
chase of $\viewinst$ with rules of the form $\forall \vec x ~ [\view_i(\vec x) \rightarrow Q_{\views_{i}}(\vec x)]$. 
When views are UCQs or Datalog, $\backview_{\views}(\viewinst)$  can be thought  of as the result of  a ``disjunctive chase''. Next, we
chase   instances in $\backview_{\views}(\viewinst_n)$ using $\Sigma$ -- thus chasing a second time. 
We have described a non-deterministic process that generates instances $\inst'$ in the base schema, where each $\inst'$ satisfies
$\Sigma$, and has a view image containing one of the $\viewinst_n$. That is, each of these ``test instances'' has a view image
containing the view image of an instance satisfying $Q$ and $\Sigma$, and thus monotonic determinacy states they should satisfy $Q$.
The process is non-deterministic because we have guesses for the initial approximation $Q_n$,
and  guesses for the approximation witnessing each view fact. 
If each of the databases $\inst'$ resulting from this process satisfies 
the original query 
$Q$,  monotonic determinacy holds.

\begin{figure}
\hspace{2mm}$\mondet(Q,\views, \Sigma)\textbf{:}\hfill$ 
\begin{algorithmic}[1]
    \For{$Q_n$ approximation of  $Q$} \Comment unfold the query
    \State $C_n \defeq \chase_{\Sigma}(\canondb(Q_n))$                                            \label{alg:query-determinacy-constraints:init} \Comment Chase an unfolding
     \State $\viewinst_n \defeq \views(C_n)$                       \Comment Apply views \label{alg:query-determinacy-constraints:view-forward}
        \For{$Q'_{m,n} \in \backview_{\views}(\viewinst_n)$}             \label{alg:query-determinacy-constraints:view-acc-inverse}
\Comment Guess a witness for each view fact
         \State $C'_{m,n} \defeq \chase_{\Sigma}(Q'_{m,n})$
\Comment Chase again\label{alg:second-chase}
        \State{\kw{If} ~ {$C'_{m,n} \not \models Q$} \Comment Check if $Q$ holds   \label{alg:query-determinacy-constraints:return} \\
        \phantom{a}\hspace{1.5cm}\textbf{ return} \textbf{false}}        
         \EndFor
      \EndFor
            \State \textbf{return} \textbf{true}
\end{algorithmic}
\caption{Process  for checking monotonic determinacy.}\label{alg:query-mondet-constraints}

\end{figure}

The fact that this process captures monotonic determinacy with respect to existential rules is straightforward (see also Lemma 5.4 in \cite{uspods20journal})
\begin{proposition} \label{prop:pipeline} $Q$ is monotonically determined over $\views$ w.r.t. $\Sigma$ if and only if the process 
of Figure \ref{alg:query-mondet-constraints}
 returns true.
\end{proposition}

The ``process'' above is \emph{not} an algorithm,  even in the absence of rules,  since there are infinitely many choices for $Q_n$ and infinitely
many choices for the approximations to substitute in  Step \ref{alg:query-determinacy-constraints:view-acc-inverse}.
With rules, there is also the problem that the chase may be infinite. Nevertheless, all of our results on
monotonic determinacy will make use of this characterization.
Intuitively, one way to  analyze monotonic determinacy is by \emph{moving forward} in the process: getting effective representations
of the intermediate artifacts produced in each step.
A second technique is to \emph{move backward}, rewriting away steps of the process to  get a simpler algorithm that does not use
these steps. We will make use of both of these approaches in our results.

\vspace{-5pt}
\section{Tools for the positive results} \label{sec:tools}
Our undecidability results will be direct reductions, and thus do not require much prior machinery.
But as mentioned in the introduction, for our positive results we develop some techniques for analyzing
rules, views,  and queries involving Datalog, which we then apply to the process of Figure \ref{alg:query-mondet-constraints}.

\myparagraph{Review of the forward-backward approach and the bounded tree-width view image property}
We now review an  idea from \cite{uspods20journal}: we can ``capture the intermediate artifacts in Figure \ref{alg:query-mondet-constraints} using tree automata''. And we can sometimes move  from tree automata in the view signature to Datalog over the
views (backward mapping). 

For a number $k$
a \emph{tree decomposition of adjusted width $k$} for an 
instance $\inst$ is a pair $\TD = (\tau, \lambda )$ consisting
of a rooted directed tree  $\tau = (\vertices,E)$, either finite or countably infinite, and a
map $\lambda$ associating a tuple
of distinct elements $\lambda(v)$ of length at most $k$
(called a \emph{bag}) to each vertex $v$ in $\vertices$
such that the following conditions hold:
\begin{inparaenum}
  \item[--] for every atom $R(\vec{c})$ in $\inst$, there is a vertex $v\in \vertices$ with
$\vec{c} \subseteq \lambda(v)$;
  \item[--] for every element $c$ in $\inst$, the subgraph of $\tau$ induced over the set  $\{\, v\in \vertices \mid c \in \lambda(v) \,\} $ is connected.
\end{inparaenum}
A tree decomposition $\TD$ of an instance $\inst$ can be associated with a labelled tree $T$ where labels describe the facts
holding on the elements associated with a node. We use a standard encoding -- see \appwrap{Appendix \ref{sec:appprelims}}.
Any such tree $T$ will be a $k$ \emph{tree code} of the instance $\inst$. Given a tree code $T$, we denote the instance it encodes with $\decode{T}$.

Above we abuse notation slightly by re-using $\lambda(v)$ to indicate  the underlying set of elements, as
well as the tuple. 
The \emph{adjusted treewidth} of an instance $\inst$, $\atw(\inst)$, is the minimum adjusted
width
of a tree decomposition of $\inst$.
For a tree decomposition $\TD$ of data instance $\inst$
let $\tspan(\TD)$ (the \emph{treespan} of the decomposition)
be the maximum over elements $e$ of $\inst$ of the number of bags containing $e$.

A counterexample to monotonic determinacy consists of instances $\inst, \inst'$ satisfying $\Sigma$ such that $\views(\inst) \subseteq \views(\inst')$, $\inst \models Q$, and
$\inst' \models \neg Q$. We will consider such a pair as a single instance in the signature with two copies of the base schema, one interpreted as in $\inst$ while the other
 is interpreted as in $\inst'$, along with one copy of the view predicates, interpreted as in $\views(\inst)$.
A \emph{tree code of a counterexample to monotonic determinacy} is a tree code for the instance formed from a counterexample in the vocabulary above.
The following proposition, an application of Courcelle's theorem \cite{courcelle},  states that for any fixed $k$, we can recognize counterexamples that are of low treewidth, using automata.

\begin{proposition} \label{prop:kpipeline} [Forward Mapping] For each $Q$ in FGDL, $\views$ in FGDL or FO, $\Sigma$ existential rules, and each $k$ there is a tree automaton that accepts  all finite $k$ tree codes of counterexamples
to monotonic determinacy of $Q, \views, \Sigma$. There is a  a B\"uchi automaton over infinite trees that holds exactly  when there is
an arbitrary (possibly infinite) $k$ tree code of such an instance. 
\end{proposition}

Of course, for monotonic determinacy, we are interested not just in treelike counterexamples, but arbitrary counterexamples. However, in some cases, low treewidth is enough.
A triple $(Q, \views, \Sigma)$ has the \emph{bounded treewidth view image property} (BTVIP) if we can compute  a  $k$ 
such that for every approximation $Q_n$ and some chase $C$ of $\canondb(Q_n)$ under $\Sigma$ it holds that
 $\atw(\views(C)) \leq k$. If we can choose a finite $C$, we say $(Q, \views, \Sigma)$ has
the \emph{finite bounded treewidth view image property} (FBTVIP).
%

\begin{proposition} \label{prop:ksuffices} If $\Sigma$ are $\fgtgd$s, $(Q, \views, \Sigma)$ has the BTVIP (resp. FBTVIP), $\views$ is in 
Datalog, then there is $k$, computable
from $(Q, \views, \Sigma)$, such that whenever monotonic determinacy
fails, there is some counterexample (resp finite counterexample) of treewidth $k$. 
\end{proposition}
In particular, the two propositions  tell us  that, when the hypotheses of both propositions apply, it suffices to check that the automaton from Proposition \ref{prop:kpipeline}  is nonempty. This will
allow us to get decidability results on monotonic determinacy in the presence of the BTVIP.

For getting results on rewriting, it is useful to move backward from tree automata accepting certain codes to a Datalog query
accepting their decodings.
Here is one formalization of the backward mapping, a variation of  \cite[Proposition 7.1]{uspods20journal}

\begin{theorem} \label{thm:auttodlog} [Backward Mapping]
 Let $\sigma$ be a relational signature, $k \in \omega$ and $A$ 
a tree automaton over the  signature for  tree codes of treewidth $k$ structures over $\sigma$. Then 
there is a Datalog program $E_A$ such that for every
$\sigma$-structure $\struct$:  $\struct \models E_A$
iff there is a finite tree code $\tree$ over $\sigma$ accepted by $A$ with a
homomorphism from $\decode{\tree}$ to $\struct$. 
\end{theorem}

In our rewriting theorems --  see Theorems \ref{thm:fgviewsplusdatalogqueryrewriting} and  \ref{thm:fgviewsplusfgdlquery} in Section \ref{sec:decide} --  the idea
is to first apply the forward mapping of Proposition \ref{prop:kpipeline} 
to get an automaton accepting treelike counterexamples.
We then project  to get an automaton over codes of view instances. We apply backward mapping of Theorem \ref{thm:auttodlog}
to get a Datalog program.
We emphasize that the same process was  used in \cite{uspods20journal} for the rewriting results in Figure \ref{fig:noconrewrite}.
The main difference is that we will need a  modification  dealing with the fact that the treecodes involved may be infinite.
This required us to extend to automata over infinite trees in Proposition \ref{prop:kpipeline}, and in  Theorem 
\ref{thm:fgviewsplusdatalogqueryrewriting} it
will require us to expand the rewriting language from Datalog to the larger logic $\poslfp$.

\myparagraph{Certain answer rewritings}
The alternative to ``moving forward'' in the process of Figure \ref{alg:query-mondet-constraints} is to go backwards, eliminating
steps in the process via \emph{certain answer rewriting results}. This will require us to obtain new results eliminating
entailment steps.

For an instance $\inst$, logical sentences $\Sigma$, and Boolean query $Q$,
we write $\inst \wedge \Sigma \models Q$ if $Q$ returns
true on every instance that includes
$\inst$ and satisfies $\Sigma$. This is just the usual logical entailment
when $\inst$ is considered as a conjunction of facts. We also say
that $Q$ is \emph{certain} for $\inst, \Sigma$. 
A \emph{certain-answer rewriting} ($\certainanswerrewriting$) of $Q$ 
with respect to 
 rules $\Sigma$ is a
query $Q_\Sigma$ such that running $Q_{\Sigma}$ on every $\inst$ will tell
whether $Q$ is certain for $\inst, \Sigma$. For a query language $L$, we talk about an $L$-$\certainanswerrewriting$.
Informally, $Q$ has an $L$-$\certainanswerrewriting$ over $\Sigma$ if we can check whether $Q$ is certain w.r.t $\Sigma$ with an $L$ query.

The following is easy to derive from prior results.

\begin{theorem} \label{thm:bbcrewriting} UCQs have
UCQ $\certainanswerrewriting$s 
  over linear TGDs \cite[Thm. 3.3]{calirewriting}. UCQs 
have $\fgdl$ $\certainanswerrewriting$s over $\fgtgd$s \cite[Thm 5.6]{bbc}.
\end{theorem}

In all cases where we write that something has a $\certainanswerrewriting$, the generation of the rewriting
is effective. We omit the precise bounds here.
We can refine the argument from \cite{bbc} to show that for $\frone$ rules, the rewriting 
is in MDL:

\begin{theorem} \label{thm:rewritefrontieronemdl}  CQs have MDL
$\certainanswerrewriting$s  over  $\frone$
TGDs.
\end{theorem}

\noindent
Less is known  when  $Q$ is in Datalog.
For $Q \in \fgdatalog$ we get a result from Thm. \ref{thm:bbcrewriting}
and ``moving Datalog rules into the existential rules'':

\begin{proposition} \label{prop:fgdlcertain}
FGDL queries have FGDL $\certainanswerrewriting$s over $\fgtgd$s.
\end{proposition}

\myparagraph{A new finite controllability result} 
We now consider a subset of Datalog that is less restrictive than MDL or even $\fgdatalog$.
The \emph{Extensional Gaifman graph} of a Datalog rule is the graph whose nodes are the variables
in the head and whose edges connect two variables if there is an atom
over an extensional relation that connects them.
A query or program is \emph{extensionally-connected} (EC) if in each rule the Extensional Gaifman graph
is connected. EC Datalog subsumes $\fgdl$
and hence MDL.
Recall the definition of ``entailment is finite controllable'' from Section \ref{sec:prelims}. We can show:

\begin{theorem}\label{thm:fcecdatalog}
For the class of $\frone$ rules and EC-Datalog queries,
entailment is finitely controllable.
\end{theorem}

This is \emph{the first non-trivial finite-controllability result we know of for Datalog queries with arbitrary arity}.
The proof proceeds along the general lines used in finitely controllability for description logics, see in particular
\cite{finitecontrol1}.

\newcommand{\ch}{\mathcal{C}}
\newcommand{\rs}{\Sigma}
\newcommand{\dl}{Q}
\newcommand{\db}{\mathcal{D}}

\begin{proof} Fix the database instance $\db$, the set of frontier-one TGDs $\rs$, and an EC Datalog query $\dl$. Let $\ch$ denote the chase of $\db$ by $\rs$, let $n$ denote the maximal size of any rule in both $\rs$ and $\dl$, and let $N = 4 \cdot n^2$.
In this sketch we make some drastic simplifications to convey the idea:

\begin{inparaenum}
	\item The database instance $\db$ is a single unary atom. 
	\item The EDB relational symbols are at most binary. 
	\item The rules of $\rs$ have exactly one frontier variable.
	\item No atoms of the shape $E(x,x)$ appear in $\db, \rs,$ and  $\dl$ 
	\item Heads of rules of $\rs$ are trees.
\end{inparaenum}
We shall disregard the direction of binary predicates whenever discussing graph-theoretic notions, such as distances, trees, or cycles

From the assumptions we easily see:

\begin{proposition}\label{prop:fc-chase-treebody}
The chase of $\db$ by $\rs$ forms a regular tree.
\end{proposition}

\begin{definition}[Unabridged]
Given an instance $\inst$ and a cycle $C$ in $\inst$ we say that $C$ is {\em unabridged} if for every pair of elements 
$t,t'$ of $C$ the shortest path between them in $\inst$ goes through $C$. 
\end{definition}
Note that in the above we treat instances as undirected graphs.

\begin{definition}[Trimming]
Consider a tree $T$ with a node $v$. The process of removing all the children of node $v$ is referred to as {\em trimming} at node $v$. Note, that trimming at leaves is allowed, but it has no effect.
\end{definition}

\begin{definition}[Unfolding tree]
For a Datalog program $P$ without repeated variables in rule heads, the \emph{unfolding tree} is any tree derived from a CQ-approximation tree $T$ of $P$ through the process of trimming at one or more nodes of $T$. 

Then, an \emph{unfolding} of $P$ is a conjunction of labels from any unfolding tree of $P$. Note that the resulting CQ may include IDB predicates. 
\end{definition}

\begin{definition}[Succession]
We say that an unfolding tree $T$ is a \emph{direct successor} of an unfolding $T'$ if $T'$ can be obtained from $T$ by trimming it at a node with only leaves as its children. We define \emph{succession} as the transitive closure of direct succession. These notions naturally extend to unfoldings.
\end{definition}

To prove Theorem \ref{thm:fcecdatalog} we need to show that
iff $\ch$ does not entail $\dl$ then there exists a finite model $M$ of $\rs$ containing $\db$ that does not entail $\dl$ as well.
We give only the construction of $M$ here.

\begin{definition}[Ancestor]
Given an infinite tree $\mathcal{T}$, a natural number $m$, and a node $u$ of $\mathcal{T}$, we define the \emph{$m$-ancestor} of $u$ as the ancestor of $u$ at a distance of $m$, if it exists; otherwise, the $m$-ancestor of $u$ is the root of $T.$.
\end{definition}

\begin{definition}[Perspective]
Given a natural number $m$ and a node $u$ of an infinite tree $\mathcal{T}$, we define the \emph{$m$-perspective} of $u$ as the pair $\pair{T', u}$ where $T'$ is the subtree of $T$ that is rooted at the $m$-parent of $u$. We consider $m$-perspectives up to isomorphism.
\end{definition}

Let $type(u)$, for a term $u$ in $\ch$, consist of two values:

\begin{inparaenum}
	\item The depth of $u$ in $\ch$ modulo $N$.
	\item The $N$-perspective of $u$. Note that \cref{prop:fc-chase-treebody} indicates there are only a finite number of such perspectives, keeping in mind that we count only up to isomorphism.
\end{inparaenum}
Define $M$ as a structure that is a quotient of $\ch$ using the ``is of the same type'' relation, where $type$ is defined as above.

Using an analysis of how $\dl$ can be satisfied, we can show that $M$ witnesses finite controllability:
 $M$ extends $\db$, satisfies $\rs$, and does not satisfy $\dl$.
\end{proof}

\myparagraph{Compactness of entailment}
Let us go back to the forward approach for analyzing pipelines of chasing and views, such as the process of
Figure \ref{alg:query-mondet-constraints}. We start by using an automaton to represent the view images of chases of unfoldings of the query, which
we can do when we have the BTVIP and some extra conditions on the views. 
But usually we need an automaton over  infinite trees to do this, as in Proposition \ref{prop:kpipeline}.  This is unfortunate, because our backward mapping result,
Theorem \ref{thm:fgviewsplusdatalogqueryrewriting}, requires an ordinary finite tree automaton to map backward into Datalog.
What will help us is that we are interested in cases where the entire process succeeds -- which means monotonic determinacy holds.
We would like to argue that this depends on only a finite part of the view image of the chase, and later conclude that a 
finite tree automata suffices to capture this part. We give a general result about entailment that will be useful for those purposes.

Fix a set of views $\views$, rules $\Sigma$, and a query $Q$,
and let $\vinst$ be an instance of the view schema.
A \emph{$\views, \Sigma$-sound realization of $\vinst$} is an instance of the base schema satisfying $\Sigma$ whose view image
includes $\vinst$.  An instance $\vinst$ is said to be \emph{$Q$-entailing} (w.r.t $\Sigma, \views$) if 
every
$\views, \Sigma$-sound realization of $\vinst$ satisfies $Q$.
We  also write $\vinst \models_{\views,\Sigma} Q$.
Monotonic determinacy can be restated as: for every instance $\inst$ satisfying $Q$ and $\Sigma$,  its view image is $Q$-entailing w.r.t. 
$\Sigma$,$\views$.

It is easy to see that if $\Sigma$ consists of existential rules, and views are CQs, then when $\vinst$ is $Q$-entailing there is  a finite subinstance $\vinst_0$ 
that is $Q$-entailing w.r.t. $\Sigma$, $\views$. We only need enough facts from $\vinst$ to fire the chase rules needed to generate a match of $Q$.
In the case where the views are in Datalog, this is not clear. But it turns out we can often obtain  this ``compactness property'':

\begin{theorem} \label{thm:compact} [Compactness for view and rule entailment] Let $\Sigma$ consist of FGTGDs, $\views$ a set of views defined
by Datalog queries, and $Q$ an $\fgdl$ query.
For every $\vinst$ 
such that $\vinst$ is $Q$-entailing w.r.t. $\Sigma$, $\views$, there is
$\vinst_0$ a finite subinstance of $\vinst$ such that $\vinst_0$ is $Q$-entailing.
\end{theorem}

Note that under the hypotheses, if we take a particular \emph{unfolding of $\vinst$} -- a choice of witness for each view fact --  then  $Q$ will hold in the chase of the
corresponding instance, and the satisfaction of $Q$ will depend on only finitely many facts in the chase, thus on only finitely many facts
in $\vinst$. The issue is that there are infinitely many witness unfoldings, and one worries that more and more facts from
$\vinst$ are required for different witnesses.  The proof of Theorem \ref{thm:compact} works by observing that the set of witnesses
needed will only  depend
on the  $j$ quantifier rank type in guarded second order logic, for sufficiently large $j$. There are only finitely many such types,
and thus we need only finitely many facts. Note that when we move to general Datalog $Q$, this result fails: see \appwrap{the appendix} for details.

\myparagraph{Recognizing  $Q$-entailing instances with automata}
Now consider a \emph{tree-like} view  instance $\vinst$:  one with treewidth $k$, given by some tree code $T$.
Theorem \ref{thm:compact} tells us that there is a finite prefix $T_0$ of $T$ whose decoding is $Q$-entailing.
It turns out that we can often recognize these $Q$-entailing finite prefixes by running
an ordinary finite tree automaton. This result and Theorem \ref{thm:compact} together will allow us to reduce monotonic determinacy
to reasoning with finite tree automata, rather than dealing with infinite trees.

\begin{theorem} \label{thm:autpre} [Automata for entailing witnesses] Let $\Sigma$ consist of FGTGDs, $\views$ a set of views defined
by $\fgdl$ queries, and $Q$ a $\fgdl$ query.
Let $k$ be a number.
There is a tree automaton $T_{Q, \views, \Sigma, k}$ running over $k$-tree codes in the view signature
such that for  finite $\vinst_0$ of tree-width $k$,
$T_{Q, \views, \Sigma, k}$ accepts a code of $\vinst_0$ if and only if $\vinst_0 \models_{\views, \Sigma} Q$.
\end{theorem}

The proof of this theorem involves  \emph{elimination of quantification over extensions of a tree
in favor of quantification within a tree}, an idea inspired by prior work in embedded finite model theory over trees
\cite{treeextension}.
Consider the set $E$ of tree codes of finite instances $\vinst_0$ such that $\vinst_0 \models_{\views, \Sigma} Q$.
The entailment relation quantifies over all extensions of $\vinst_0$. But in our setting, it equivalently quantifies over all
tree-like instances, finite or infinite. Since all the queries involved are well-behaved, we can write a monadic second order sentence $\phi$ quantifying
over infinite trees such that a finite tree is in $E$ if and only if all of its extensions -- ``suffixes'' that add on additional branches, both
 finite and infinite --  satisfy $\phi$.
But we can argue that quantifying over all infinite extensions in the way does not take us out of regular tree languages.

\vspace{-5pt}
\section{Applying the tools for decidability and rewritability}  \label{sec:decide}
We first apply the tools from Section \ref{sec:tools} to give \emph{decidability results}.
We focus on three decidable cases: 
 \emph{FGDL TGDs, FGDL queries, and FGDL views}
-- that is  ``everything is guarded''; 
\emph{Fr-1 TGDs, MDL queries and CQ views} -- which we shorten as ``Fr-1 rules and queries, CQ views''.
and finally,
 \emph{linear TGDs, CQ query, CQ views} -- that is, ``linearity and CQs'';
For brevity we omit
 complexity bounds in the statements, but elementary upper bounds are easy to derive. 
What we want to highlight here is that \emph{the conditions are very restrictive, but just afterwards we will show that they are necessary}.


\myparagraph{Decidability via bounding treewidth}
In the first case mentioned above, we establish decidability via the forward-reasoning approach in the previous section, and more specifically the \emph{Forward mapping tools}:

\begin{theorem} \label{thm:datalogqfgdlviewsdecide} Let $Q$ range over $\fgdl$ queries, $\Sigma$ over FGTGDs, and $\views$ over FGDL views.
Then monotonic determinacy is decidable.
\end{theorem}

\begin{proof}
It is easy to show that we have the BTVIP in the case above: the approximations of $Q$ are tree-like, chasing with FGTGDs
adds on additional branches to the tree, and applying  FGDL preserves the tree structure. We just check non-emptiness of the  automaton of Proposition \ref{prop:kpipeline}.
By Proposition \ref{prop:ksuffices} this is sufficient
\end{proof}

The same method applies to the \emph{second case above}, MDL queries and CQ views, and  $\frone$ rules:

\begin{theorem} \label{thm:mdlqcqviewslindecide} Let $Q$ range over MDL queries, $\Sigma$ over $\frone$ TGDs, and $\views$ over CQ views.
Then monotonic determinacy is decidable.
\end{theorem}
\begin{proof}
When we approximate an MDL query and  chase, we get a structure with bounded tree-width and
low treespan. That is, a value appears in a fixed number of bags. When we apply CQ views, the
treewidth is bounded  \cite[Lem 6.5, Thm 8.2]{uspods20journal}.
\end{proof}

\myparagraph{Decidability based on certain answer rewriting} 
We tackle the final decidable case by giving decidability via the  ``backwards analysis'', using \emph{certain answer rewriting techniques from
the previous section}.
Recall that in the case of CQ views,  Step \ref{alg:query-determinacy-constraints:view-acc-inverse} in
the process of Figure \ref{alg:query-mondet-constraints}
amounts
to chasing with respect to the source-to-target rules $\backview_\views$ mentioned earlier. Thus the entire process can be considered to consist
of chase steps, and we can proceed by a rewriting approach that  removes some of these steps:
 we first find a $\certainanswerrewriting$ $R_1$ of $Q$ with respect to $\Sigma$, and then a $\certainanswerrewriting$ $R_2$ of
$R_1$ with respect to $\backview_\views$.  
Then we have monotonic determinacy of $Q$ over $\views$ with respect to $\Sigma$ exactly when $Q$ entails
$R_2$ with respect to $\Sigma$.
We can apply this straightforwardly in the case of linear TGDs, where Theorem \ref{thm:bbcrewriting} tells us we can find UCQ $\certainanswerrewriting$s $R_1$, which
will generate another UCQ $\certainanswerrewriting$ $R_2$.

\begin{theorem} \label{thm:linearcqcq}
If $\Sigma$ ranges over linear TGDs,
then monotonic determinacy of UCQ  $Q$ over CQ views $\views$ relative
to $\Sigma$ is decidable.
\end{theorem}

\myparagraph{Decidability in the finite using the finite controllability result} 
We have sketched the argument for decidability of monotonic determinacy over all instances for three cases.  In  these cases, we
 can show that monotonic determinacy is also decidable in the finite as well.
We focus on a special case of the second decidable case, which will make use of the \emph{finite controllability tool} (Theorem \ref{thm:fcecdatalog}) from the previous section.

\begin{theorem} \label{thm:finconfone}
Let $Q$ be a CQ  query, $\views$ a set of  CQ views, and $\Sigma$ is a set of $\frone$ TGDs.
If $Q$ is monotonically determined by $\views$ with respect to $\Sigma$ over finite structures, then the same holds over all structures.
\end{theorem}
\begin{proof}
Using our result on certain-answer rewriting,  Theorem \ref{thm:rewritefrontieronemdl},
we  get an MDL certain answer rewriting $R_1$ of $Q$ w.r.t. $\Sigma$.
By a direct argument we get an MDL c.a. rewriting $R_2$ of $R_1$ with respect to the rules that correspond to  view definitions.

Now if monotonic determinacy fails over all instances, we know, by Proposition \ref{prop:pipeline} and the properties
of c.a. rewritings,  that $R_2$ is not entailed by $Q$ and $\Sigma$.
But now we can apply Theorem \ref{thm:fcecdatalog} to conclude  there is a finite counterexample $\inst_1$, satisfying
$Q \wedge\Sigma \wedge \neg R_2$.
The view image of $\inst_1$, and its chase under the backward mapping rules is likewise a finite instance $\inst'_1$.
We claim that there is a finite instance  $\inst_2$ containing $\inst'_1$, satisfying $\Sigma \wedge \neg Q$. 
If this were the 
case,  $\inst_1$ and $\inst_2$ together would contradict monotonic determinacy in the finite.
We obtain $\inst_2$ by simply applying \cref{thm:fcecdatalog} again for $\inst'_1, \Sigma$, and $R_1$.
\end{proof}

We now apply the tools to investigate what kind of rewritings we can obtain.  

\myparagraph{View rewritings based on certain answer rewritings}  
One approach to obtaining rewritings is  ``working backwards'' on the pipeline in  Fig. \ref{alg:query-mondet-constraints}, applying the
tools related to \emph{certain answer rewritings} in Section \ref{sec:tools} to remove the last steps in the process.


\begin{theorem} \label{thm:fgdlcqviews}
Let $\Sigma$ be a set of $\fgtgd$s, $Q$ be a $\fgdl$ query, and
$\views$ a set of CQ views. Then if $Q$ is monotonically determined over $\views$
with respect to $\Sigma$, there is a rewriting of $Q$ in Datalog. 
\end{theorem}

\begin{proof}
We apply Proposition \ref{prop:fgdlcertain}
to get a Datalog certain answer rewriting $R_1$ of $Q$ under $\Sigma$, thus
reversing the final step in the process. We then get a certain answer rewriting $R_2$ of $R_1$ under the rules $\backview_{\views}$,
which are source-to-target TGDs.   We can check that $R_2$ is the desired rewriting.
\end{proof}

\myparagraph{View rewriting based on forward-backward}
The certain-answer rewriting approach applies only to CQ views. 
We now provide an  approach based on bounding treewidth.

\begin{theorem} \label{thm:fgviewsplusdatalogqueryrewriting} 
 Let $\Sigma$ be a set of $\fgtgd$s, 
$Q$ a Boolean Datalog query and $\views$ a set of $\fgdatalog$ views.
Then if $Q$  is monotonically
determined by $\views$ $\wrt$ $\Sigma$ there is a 
$\poslfp$ rewriting.
\end{theorem}
\begin{proof}
We use a modification of the forward-backward approach mentioned after Theorem \ref{thm:auttodlog}.
Using the same argument as Proposition \ref{prop:kpipeline}  we can get a B\"uchi automaton for tree-like view images of unfolding of
$Q$.
We can project on to the view signature to get a B\"uchi automaton accepting view images.
And because  $\Sigma, Q, \views$ has the BTVIP it is sufficient to check such tree-like counterexamples. 
If we ignore  acceptance conditions in this automaton, we get an ordinary finite tree automaton. 
We apply  the backward mapping of  Theorem \ref{thm:auttodlog} to this
to get Datalog rules. We   can add nested fixpoints to capture
the B\"uchi acceptance conditions.
\end{proof}

We do not know if Datalog suffices for $\fgtgd$s. But we can show $\poslfp$ is necessary in order to rewrite
$Q$ monotonically determined over $\views$ w.r.t.  $\Sigma$ for cases when we have the BTVIP:
\begin{theorem} \label{thm:noncompact} There are $\Sigma$,  Datalog $Q$ and $\views$ such that the BTVIP holds, 
with $Q$ monotonically determined by $\views$ w.r.t. $\Sigma$, where there is a $\poslfp$ rewriting but no Datalog rewriting over all structures.
\end{theorem}
\noindent 
The argument will rely on the \emph{failure} of ``compactness of entailment'' for this class. We provide an example where
compactness fails, and show that this failure implies that a Datalog rewriting is impossible.

\myparagraph{Better rewritings using Compactness of Entailment and Automata for Entailing Witnesses}
We can also use \emph{Compactness of Entailment} and \emph{Automata for entailing witnesses} from Section \ref{sec:tools} to show that if $Q$ is in $\fgdatalog$, we can do better than $\poslfp$:

\begin{theorem} \label{thm:fgviewsplusfgdlquery}
 Let $\Sigma$ be a set of $\fgtgd$s,
$Q$  a Boolean $\fgdatalog$  query, $\views$ a set of $\fgdatalog$ views.
Then if $Q$ is monotonically determined by $\views$ \wrt $\Sigma$, then
there is a Datalog rewriting.
\end{theorem}

\begin{proof}
Since $Q, \views, \Sigma$ has the BTVIP, there exists natural $k$ such that for every CQ approximation of $Q$ there exists its chase $C_n$ under $\Sigma$ having $\atw(\views(C_n)) \leq k$.

Applying  \emph{Compactness of Entailment}, Theorem \ref{thm:compact}, for any view instance of treewidth $k$ that entails (via $\backview$ and $\Sigma$) $Q$, there
is a finite  $\viewinstance$ contained in it that entails $Q$ the same way. By  the \emph{Automata for entailing witnesses result}, Theorem \ref{thm:autpre}, there is an ordinary
tree automaton $A$ that captures the codes of these finite subinstances.
We  convert $A$ to Datalog, using the backward mapping of \cref{thm:auttodlog}.
\end{proof}

Using a similar technique, we can also conclude rewritability when $Q$ is a CQ but $\views$ are arbitrary Datalog.
\begin{theorem} \label{thm:datalogviewsandqquery}
Suppose $Q$  is a Boolean UCQ (resp. CQ), $\views$ a set of Datalog
views, with
 $Q$  monotonically determined by $\views$ \wrt $\Sigma$.  Then
there is a UCQ (resp. CQ) rewriting of $Q$ over $\views$ relative to $\Sigma$.
\end{theorem}



\vspace{-5pt}
\section{Surprising Undecidability}  \label{sec:undecide}
The hypotheses that give us decidability in Section \ref{sec:decide}  are much stronger than in the absence of rules, where all the cases without recursion
are decidable.
Remarkably, they cannot be loosened. If we just mix the constraint classes from the second and third decidable case  above, we get undecidability
\emph{even with CQ  queries and views}.

\begin{theorem} \label{thm:undecfrontonepluslinearcqboth}
The problem of monotonic determinacy is undecidable when $\Sigma$ ranges over combinations of Linear TGDs
and Frontier-$1$ TGDs, $Q$ over Boolean
CQs, and $\views$ over CQ views. If we allow MDL queries, we have undecidability with just Linear TGDs.
\end{theorem}

The idea is that we can code a cellular automaton
in a monotonic determinacy problem.

If we allow UCQ views  -- adding disjunction, but no recursion -- we can even get undecidability
for rules in the intersection of  the two well-behaved rule classes:
\begin{theorem} \label{thm:undecfrontonecqucq}
The problem of monotonic determinacy is undecidable when $\Sigma$ ranges over UIDs, which are linear and frontier-$1$, $Q$ over Boolean
atomic UCQs, and $\views$ over UCQ views.
\end{theorem}

Here the idea is to code a tiling problem (or a non-deterministic Turing Machine). In the process of
Figure \ref{alg:query-mondet-constraints}, chasing the canonical database of one disjunct of the UCQ --
the ``start disjunct'' -- with the rules
will generate two infinite  axes, while one disjunct of a UCQ view will generate the cross product of the
axes: all$(x,y)$ co-ordinate pairs.
 ``Reverse chasing'' with the other disjuncts of the views will generate grid points for each such pair, with
each grid point tagged with a tile predicate.
Another disjunct of the CQ -- a ``verification disjunct'' -- will return true if the tiles violate one of the
forbidden patterns.
Note that  the use of a UCQ will allow us to code non-deterministic computation, while a CQ
view (as in the previous theorem) allows us only to mimic deterministic computation.

With more effort we can get undecidability even with $Q$ a CQ:
\begin{theorem} 
\label{thm:undecuidcqucq}
The problem of monotonic determinacy is undecidable when $\Sigma$ are UIDs, $Q$ is a Boolean
CQ, and $\views$ are UCQ views.
\end{theorem}
Similarly to the proof of \cref{thm:undecfrontonecqucq} we code a tiling problem.
The main new challenge is  we do not have multiple disjuncts in $Q$ to test
for different violations.  So now \emph{conjuncts} in our CQ will represent different violations of a correct
tiling. We arrange that if there is one violation of correctness of a tiling, there are 
``degenerate ways'' to map the remaining conjuncts of the CQ.
Details are in \appwrap{the appendix}.

Thus far  we have focused on FGTGDs, where the chase may not terminate.
For full TGDs  (no existentials in the head) the chase terminates. It follows that monotonic determinacy is finitely controllable.
Some of the rewritability results easily carry over from the case without rules. For example, for the case of CQ views, we can 
infer that monotonically determined Datalog queries are Datalog rewritable, using the inverse rules algorithm \cite{inverserules}.
But we can show that \emph{even for full TGDs, some of the decidable cases  from Fig. \ref{fig:nocondecide}
become undecidable}. In particular:

\begin{theorem}\label{thm:mdl-cq-fulltgd}
The problem of monotonic determinacy is undecidable when $Q$ ranges over MDL, $\views$ over unary atomic views, and $\Sigma$ over full TGDs.
For such $Q, \views, \Sigma$ the problems of Datalog rewritability, CQ rewritability, and the variants of these problems over finite instances are
also undecidable.
\end{theorem}

We proceed here by constructing a family of $Q, \Sigma, \views$ where chasing the approximations of $Q$ allows us to build an
arbitrarily large 
 deterministic computation graph. The views will just indicate whether this computation is accepting.
This will allow us to encode a  deterministic machine acceptance problem as monotonic determinacy.   Because $\views$ are so restricted (unary atomic!), monotonic determinacy coincides with CQ rewritability for this family.


\begin{figure}
\centering

\scalebox{0.7}{

\begin{tabular}{|
>{\columncolor[HTML]{EFEFEF}}c |ccc|}
\hline
\cellcolor[HTML]{C0C0C0}$Q \backslash \views$            & \cellcolor[HTML]{EFEFEF}CQ                                                                                                                                                                                                                            & \cellcolor[HTML]{EFEFEF}\begin{tabular}[c]{@{}c@{}}MDL,\\ FGDL\end{tabular}                                                                                                                                 & \cellcolor[HTML]{EFEFEF}\begin{tabular}[c]{@{}c@{}}UCQ,\\ DL\end{tabular} \\ \hline
\begin{tabular}[c]{@{}c@{}}CQ,\\ UCQ\end{tabular}        & \multicolumn{1}{c|}{\begin{tabular}[c]{@{}c@{}}Fr-1 $\cup$ Lin - \textbf{Und.}, Thm. \ref{thm:undecfrontonepluslinearcqboth}\\ \textit{Lin - Dec.}, \textit{Thm. \ref{thm:linearcqcq}}\end{tabular}}                                                  & \multicolumn{1}{c|}{FGTGD - Dec., Thm. \ref{thm:datalogqfgdlviewsdecide}}                                                                                                                                   &                                                                           \\ \cline{2-3}
MDL                                                      & \multicolumn{1}{c|}{\begin{tabular}[c]{@{}c@{}}Lin - \textbf{Und.}, Thm. \ref{thm:undecfrontonepluslinearcqboth}\\ full\,TGD - \textbf{Und.}, Thm. \ref{thm:mdl-cq-fulltgd}\\ \textit{Fr-1 - Dec., Thm. \ref{thm:mdlqcqviewslindecide}}\end{tabular}} & \multicolumn{1}{c|}{\begin{tabular}[c]{@{}c@{}}\phantom{a}\\ full TGD - \textbf{Und.}, Thm. \ref{thm:mdl-cq-fulltgd}\\ \textit{FGTGD -} \textit{Dec., Thm. \ref{thm:datalogqfgdlviewsdecide}}\end{tabular}} &                                                                           \\ \cline{2-2}
FGDL                                                     & \multicolumn{1}{c|}{Lin - \textbf{Und.}, Thm. \ref{thm:undecfrontonepluslinearcqboth}}                                                                                                                                                                & \multicolumn{1}{c|}{}                                                                                                                                                                                       &                                                                           \\ \cline{2-3}
\begin{tabular}[c]{@{}c@{}}\phantom{a}\\ DL\end{tabular} & \multicolumn{3}{c|}{UID - \textbf{Und.}, Thm. \ref{thm:undecuidcqucq}}                                                                                                                                                                                                                                                                                                                                                                                                                                                                          \\ \hline
\end{tabular}

}

\caption{Decidability with Rules} \label{fig:decidecon} 
\end{figure}
\begin{figure}
\centering

\scalebox{0.7}{

\begin{tabular}{|c|cccc|}
\hline
\rowcolor[HTML]{EFEFEF} 
\cellcolor[HTML]{C0C0C0}$Q \backslash \views$                                    & CQ                                                                    & \begin{tabular}[c]{@{}c@{}}MDL, \\ FGDL\end{tabular}                                  & UCQ                                                                                     & DL                                       \\ \hline
\cellcolor[HTML]{EFEFEF}CQ                                                       & \multicolumn{1}{c|}{\textcolor{blue}{CQ, \dag}}                       & \multicolumn{1}{c|}{}                                                                 & \multicolumn{1}{c|}{}                                                                   & CQ Thm. \ref{thm:datalogviewsandqquery}  \\ \cline{2-2} \cline{5-5} 
\cellcolor[HTML]{EFEFEF}UCQ                                                      & \multicolumn{1}{c|}{\textcolor{blue}{\makecell{UCQ,\\ Thm 2.8 \dag}}} & \multicolumn{1}{c|}{}                                                                 & \multicolumn{1}{c|}{\multirow{-2}{*}{\textcolor{blue}{\makecell{UCQ,\\ Thm 2.8 \dag}}}} & UCQ Thm. \ref{thm:datalogviewsandqquery} \\ \cline{2-2} \cline{4-5} 
\cellcolor[HTML]{EFEFEF}\begin{tabular}[c]{@{}c@{}}MDL, \\ FGDL\end{tabular}     & \multicolumn{1}{c|}{\textcolor{red}{DL}, Thm. \ref{thm:fgdlcqviews}}  & \multicolumn{1}{c|}{\multirow{-3}{*}{DL, Thm. \ref{thm:fgviewsplusfgdlquery}}}        & \multicolumn{2}{c|}{}                                                                                                              \\ \cline{2-3}
\cellcolor[HTML]{EFEFEF}\begin{tabular}[c]{@{}c@{}}\phantom{a}\\ DL\end{tabular} & \multicolumn{1}{c|}{Open}                                             & \multicolumn{1}{c|}{\textbf{PosLFP}, Thm. \ref{thm:fgviewsplusdatalogqueryrewriting}} & \multicolumn{2}{c|}{\multirow{-2}{*}{n.n. DL}}                                                                                     \\ \hline
\end{tabular}

}

\caption{Rewritability with Rules; \dag \protect\cite{thebook}} \label{fig:rewritecon} 
\end{figure}

\vspace{-5pt}
\section{Conclusions} \label{sec:conc}
We  investigated reasoning problems involving the interaction of views, queries, and background knowledge. We developed 
tools  which can be
applied to get positive results about monotonic determinacy.
And we show that even when you combine
simple queries, views, and rules, static analysis problems involving all of them can become undecidable.

The goal in this paper was \emph{not} to discuss all we know about monotonic determinacy with rules.
Many results follow easily from prior work.
But for completeness,
Figures \ref{fig:decidecon} and \ref{fig:rewritecon} summarize our knowledge.
The results apply to $\fgtgd$s, except where
we  indicate  restrictions (e.g. LIN for Linear rules).
In the tables {\bf bold font} highlights a significant difference from the case without rules.
In Figure \ref{fig:rewritecon}  \textcolor{blue}{Blue} indicates that results from prior work \cite{thebook}
can be used out of the box.
\textcolor{red}{Red}  indicates use of certain answer-rewriting approaches,  while \textcolor{black}{Black} indicates
an approach based on bounding treewidth. For readability, we omit some results for full existential rules in the table.
As shown in the table, a number of cases are left open, both for decidability and rewritability. 

\section*{Acknowledgements}
Piotr Ostropolski-Nalewaja was supported by the European Research Council (ERC) Consolidator Grant 771779 (DeciGUT).
\phantom{\cite{this-paper-full-version}} 
\bibliographystyle{kr}
\bibliography{det}

\newpage
\onecolumn
\appendix
\section{Additional details for codes,  automata,  Datalog, and the characterization of monotonic determinacy, Section \ref{sec:prelims}} \label{sec:appprelims}
\subsection{Basics of tree codes, tree automata, and Monadic Second Order Logic}
\label{subsec:tree codes}


A key component of our arguments involves
coding  bounded treewidth structures by labelled trees with a fixed signature -- that is, the notion of
tree code, referred to in the body of the paper. 
We now review the formal details, as well as the notation that we will use throughout the appendix.
We emphasize that the encoding we are using is quite standard: see, for example, the appendix of \cite{visible}.

Fix  a number $k > 0$ 
and fix a relational schema
$\sigma$. We define a schema $\codes{\sigma}{k}$ of unary predicates
such that it contains the following predicates:
\begin{enumerate}
 \item For every two numbers $l,l' \in \{1,\dots,k\}$ a predicate
$T^=_{l,l'}$.
 \item For every $m$-ary relation symbol $R$ from $\sigma$ and $m$-tuple
$\ntuple = (l_1,\dots,l_m) \in \{1,\dots,k\}^m$ a predicate
$T^R_\ntuple$. In the case where $m = 0$ this means that we add just one
predicate $T^R_\varepsilon$ for the empty tuple $\varepsilon = {()}$.
 \item For every partial injective map $g$ from $\{1,\dots,k\}$ to
$\{1,\dots,k\}$ a predicate $T_g$.
\end{enumerate}

Informally, the predicates allow a single vertex of a tree to encode the $\sigma$ predicates holding on up to $k$ elements of
a structure. A number  in $1 \ldots k$ occurring in a predicate of a vertex is a  \emph{local name} of the vertex. 
The predicates $T^=_{l,l'}$ tell which pairs of local names code the same element. The predicates
$T^R_\ntuple$ indicate which tuples of elements coded in a vertex  satisfy relation $R$.
The predicates $T_g$ describe how the elements coded in  a vertex relate to the elements coded in its parent.

A \emph{finite tree} is a finite directed graph with no cycles, in which every element has at most one predecessor and
exactly one element, denoted the root of the tree, has no predecessor.
An \emph{$\omega$-tree} is as above, but where the set of vertices is countably infinite.
When we refer to a \emph{tree} we mean either a finite tree or an $\omega$-tree.

A $\codes{\sigma}{k}$-tree $\tree$ is a  tree whose
vertices are labelled with predicates from $\codes{\sigma}{k}$. We
assume that there is an $r$ such that every node that is not a leaf has exactly $r$ children
which are ordered. If we want to emphasize the number of children,  we call $\tree$ a
tree with \emph{branching width $r$}.

A $\codes{\sigma}{k}$-tree $\tree$ is \emph{coherent} if
\begin{enumerate}
 \item For every vertex $v$ of $\tree$ the assignment of the predicates
$T^=_{l,l'}$ to $v$ satisfies the axioms of equality. Precisely,
this means that for all $l,l',l'' \in \{1,\dots,k\}$
\begin{enumerate}
 \item the predicate $T^=_{l,l}$ is true of $v$,
 \item if $T^=_{l,l'}$ is true of $v$ then $T^=_{l',l}$ is true of
$v$, and
 \item if $T^=_{l,l'}$ and $T^=_{l',l''}$ are both true of $v$ then
$T^=_{l,l''}$ is also true of $v$.
\end{enumerate}
 \item For every vertex $v$ of $\tree$ the assignment of predicates is
closed under substitution of equal local names. This means that if
$T_{l,l'}$ is true at $v$ and $T^R_\ntuple$ is true at $v$ for some
$m$-ary relation symbol $R$ and $M$ tuple $\ntuple = (l_1,\dots,l_k)$
such that $l_i = l$ then $T^T_\mtuple$ is also true $v$, where $\mtuple
= (h_1,\dots,h_m)$ is such that $h_i = l'$ and $h_j = l_j$ for all $j
\neq i$.
 \item If $v$ is a child of $w$, $T_g$ is true at $v$ and $T^=_{l,l'}$
is true at $w$ then $l \in \dom(g)$ iff $l' \in \dom(g)$. Moreover, if
$l \in \dom(g)$ and $l' \in \dom(g)$ then $T^=_{g(l),g(l')}$ is true at
$v$.
 \item For every nullary predicate $P$ in $\sigma$ and vertices $v$ and
$w$ of $\tree$ we have that $T^P_\varepsilon$ is true of $v$ iff
$T^P_\varepsilon$ is true of $w$.
 \item For every vertex $v$ of $\tree$ there is exactly one partial
injection $g$ such that $T_g$ is true of $v$.
\end{enumerate}

Coherent $\codes{\sigma}{k}$ trees are also called \emph{tree codes}.
The \emph{decoding} of a tree code $\tree$ is  the structure
$\decode{\tree}$ that is defined as follows: 
\begin{itemize}
\item The domain of
$\decode{\tree}$ consists of all the pairs $(v,l)$, where $v$ is a
vertex in $\tree$ and $l \in \{1,\dots,k\}$, quotiented by the smallest
equivalence relation $\equiv$ such that
\begin{itemize}
 \item if $T^=_{l,l'}$ is true at $v$ then $(v,l) \equiv (v,l')$, and
 \item if $w$ is the parent of $v$, $T_g$ is true at $v$ and $l \in
\dom(g)$ then $(w,l) \equiv (v,g(l))$.
\end{itemize}
We write $[v,l]$ for the equivalence class of the pair $(v,l)$ under
$\equiv$. 
\item A tuple $(c_1,\dots,c_m)$ of equivalence classes is put in the
interpretation of the $m$-ary predicate $R$ from $\sigma$ iff there are
$l_1,\dots,l_m$ and a vertex $v$ such that $(v,l_i) \in c_i$ holds for
all $i = 1,\dots,m$.
\end{itemize}

We call a tree code $\tree$ \emph{active} if there is at least one
element in the active domain of $\decode{\tree}$.

We stress again that this is the standard definition. In the standard setting, we would
assume also that \emph{tree codes, and the instances they define via decoding, are finite}.
But the definitions also make sense for infinite trees (say, with
countable depth). The following proposition connecting codings and decodings is routine.

\begin{proposition}
 For every   $\sigma$-structure $\struct$ of adjusted treewidth $k - 1$ and
number $r \geq 2$ there is a tree code $\tree$ with branching width $r$ such
that $\decode{\tree}$ is isomorphic to $\struct$. 
\end{proposition}

\myparagraph{Tree automata and MSO} We deal with two kinds of tree automata in this work. 
Finite tree automata are associated with a label alphabet $\Sigma$ and a branching factor $r$. They 
are given by a finite set of states $Q$ and an initial state $q_0 \in Q$, a transition relation that
is  a subset of  $Q^r \times \Sigma \times Q$, and a subset of $Q$, the final
states. The notion of an accepting  run of such  an automaton on a tree with branching $r$ is the usual one. We  annotate
each node of the tree with a state, such that the root gets an initial state and leaves get a final state.
The states assigned in the run to a parent and child must be consistent with the transition relation. An automaton
accepts a tree if it has an accepting run, and the language of an automaton is the set of trees it accepts.

Recall that \emph{MSO} is \emph{Monadic Second Order Logic} over trees. We have first and second order variables, a binary relation symbol for the parent
child relationship, and unary relations for each alphabet symbol. Atomic formulas include $S(x)$ for $S$ a second order variable, and
the usual atomic formula involving the parent-child relation and the unary predicates. We build up using Boolean operations, first order
quantifiers, and second order quantifiers. The notion of a tree satisfying an MSO sentence is the standard one, and the language
of an MSO sentence is the set of trees satisfying it.

We will use the well-known fact \cite{thomas}:

\begin{proposition} Tree automata and MSO  express the same set of languages of finite trees.
\end{proposition}

These are the \emph{regular tree languages}.

We also  make use of automata over \emph{infinite trees}, such as  \emph{B\"uchi  tree automata}. The latter are specified
in the same way as finite tree automata. They accept an $r$-branching tree whose depth is infinite. An accepting run
associates each node of the tree with a state. The root must be labelled with the initial state, and the transition function
must be respected as with finite tree automata.  But now the accepting
states must occur infinitely often down every branch.   Consider  \emph{Weak Monadic Second Order Logic} (WMSO), 
where second order quantification is only over finite subsets of the vertices.
It is known \cite{weakbuchi} that a WMSO collection of trees is recognizable  by a B\"uchi  tree automaton. Indeed the  WMSO definable sets
are exactly those such that both the set and its complement are recognizable by B\"uchi automata.

\myparagraph{More on automata and tree codes}
Naturally, we can  run tree automata on tree codes. For this purpose we make the
assumption that the tree codes have a constant branching width of $r$
for all nodes that are not leafs and that the children of any parent
node are ordered. We also think of $\codes{\sigma}{k}$-trees as trees
over the alphabet $\Sigma$ that contains letters for all 
sets of  predicates in $\codes{\sigma}{k}$. If $\tree$ is coherent then
any such set contains only one predicate of the from $T_g$. In this
case we write $\code^g_B$ for the letter $\{T_g\} \cup B$, where $B$ is
a set of predicates of the form $T^=_{l,l'}$ or $T^R_\ntuple$. We call
$B$ a \emph{bag}.

Given a subsignature $\sigma_V$ of $\sigma$ we can define for every
$\codes{\sigma}{k}$-tree $\tree$ its restriction
$\restr{\tree}{\sigma_V}$ to be a $\codes{(\sigma_V)}{k}$-tree in which
from every letter we drop all predicates $T^R_\ntuple$ such that $R$ is
not in $\sigma_V$. It is easy to see that this is equal to the restriction of
$\decode{\tree}$, which is  defined using the full signature $\sigma$, to a smaller label set.

When we restrict to codes that are finite, we will run finite tree automata.
We will use standard results about closure of regular tree languages under \emph{projection}.
That is,  for every $\codes{\sigma}{k}$-tree automaton $\aut$ there is a
$\codes{(\sigma_V)}{k}$-tree automaton $\restr{\aut}{\sigma_V}$ such
that for all $\codes{(\sigma_V)}{k}$-trees $\tree$
\[
 \restr{\aut}{\sigma_V} \mbox{ accepts } \tree \iff \mbox{ there exists some
$\codes{\sigma}{k}$-tree $\tree^+$ accepted by $\aut$ such that }
\restr{\tree^+}{\sigma_V} = \tree.
\]

\subsection{More details on approximating Datalog with CQs}
In the body of the paper we referred to CQ approximations of a Datalog query. This is  a well-known notion that originates
in \cite{chaudhurivardi}. We give the formalization that we use in this work, which is taken from \cite{uspods20journal}, see in particular
Appendix C.

We will deal with Head-Unconstrained (HU) Datalog queries, which do not allow repeated variables in the head, except
in the goal predicate, which cannot occur in the body of rules.  Every Datalog query can be rewritten
to an HU Datalog query \cite[Proposition 3.1]{uspods20journal}.  HU Datalog is easier to deal with, since we can always unify
an intensional fact in a rule body with any rule head using the same predicate.

\begin{definition}[CQ approximation]
A \emph{CQ approximation tree} for a HU Datalog query $(\Pi, \goal)$ will be a tree with each node 
labelled with an atom, with that atom using only variables. Non-leaf nodes will always be 
labelled with atoms over intensional relations of $\Pi$, while leaf nodes will be labelled with 
extensional atoms of the input schema. Approximation trees will never have the root as a leaf.
 We require that the root label of a CQ approximation 
tree does not contain repeated variables, unless it involves the $\goal$ predicate.
For example,  using $\goal(x,y,x,z)$ as a root label
is allowed, but $A(x,y,x,z)$ for $A$ a different predicate is not. Each non-leaf node $n$ will also be associated with a rule of $\Pi$, 
denoted $\ruleof(n)$.

The idea of the repeated variable restriction is: we do not need  repeated variables in head in intermediate predicates, because
we can push the repetition into the rule body atoms that will unify with the heads. The only reason we need repeated variables
in heads is if the final query we want to perform involves repetition.

We define the depth $n$ CQ approximation trees for a Datalog query, rooted at an arbitrary intensional fact $A$ with no repeated
variables, by induction on $n$.
The base case is a tree consisting of a single root node 
labelled with an intensional atom $A$ having all variables distinct, whole children are all extensional atoms, with the children matching the rule bodies
of a rule  for $A$.

For the inductive case, consider the case of a non-goal rule of $\Pi$, $\rho \datalogarrow U(\vec x) \datalogarrow \beta_1, \ldots \beta_k$ and
let $\beta_{n_1} \ldots \beta_{n_j}$ be intensional atoms in the rule body.
Let  $\mathbb{T} = T_{n_1} \ldots T_{n_j}$
be a tuple of CQ approximation trees such that the relation of $\beta_{n_i}$ 
is the same as the relation of the root label of $T_{n_i}$, for every $i$. 
Then the tree $T$ built in the following way is a CQ approximation tree:
\begin{itemize}
    \item Set $U(\vec x)$ to be the label of the root node $n$ of $T$ and set $\ruleof(n) = \rho$.
    \item For each $n_i$ Let $h_{n_i}$ be a homomorphism taking the atom $A_{n_i}$ in the root label of $T_{n_i}$ to $\beta_{n_i}$. Such
a homomorphism always exists if the relations match, since the variables in the intensional atom $A_{n_i}$ are distinct.
    \item Let $T'_{n_i}$ be obtained from $T_{n_i}$ by replacing $y$ in every label of $T_{n_i}$ 
    with $h_{n_i}(y)$ for every variable $y$ of the root label of $T_{n_i}$, and by replacing other variables of $T_{n_i}$
with fresh ones.
     \item  For the extensional atoms $\beta_i$ in the rule body, we let $T_i$ be a tree consisting only of the atom.
    \item The nodes of $T$ will be the root node defined above along with the nodes of $T'_i$.  We set 
the children of the root node to be the root of every $T'_i$. The edge relations and labels outside of the
root are those inherited from  $T'_i$.
\end{itemize}
\end{definition}

Every CQ approximation tree rooted at the goal predicate of $(\Pi, \goal)$ defines a CQ, by conjoining the leaf atoms. This is a 
\emph{CQ approximation of the Datalog query}. It is well-known and straightforward to see that:

\begin{proposition} A Datalog query is the union of its approximations.
\end{proposition}

We now claim that approximations of a Datalog
program can be captured with tree automata.
The following is  a variation of Proposition 6.9 of \cite{uspods20journal}, proven in Appendix C. But the idea goes back
to Chaudhuri and Vardi, see \cite{chaudhurivardi}:

\begin{proposition} \label{prop:boundedtreewidthapprox}
 If $P$ is a Datalog query and $k$ is a bound on the size of the
bodies of rules in $P$ then every CQ approximation of $P$ has a cannonical tree decomposition, which witnesses that it has
tree-width $k$. And there is a finite tree automaton that accepts exactly these tree decompositions.
\end{proposition}

Note that we do \emph{not} require that $Q$ is in FGDL.  The approximations are well-behaved for any Datalog query. This result will be used as a component of
showing that certain examples have the BTVIP -- e.g. in   Theorem  \ref{thm:datalogqfgdlviewsdecide}. It will also be used in some of our rewriting results via
the forward-backward method.

Of course, we will generally
need $Q$ to be in Datalog to apply this result to obtain decidability, since without that we cannot effectively check whether $\neg Q'$  holds on a tree, which is part of being
a counterexample to monotonic determinacy.

\subsection{Tree-like structure for the chase of instances, and the chase of approximations} \label{app:chaseapproxregular}

In the previous subsection, we argued that the tree codes of approximations of a Datalog program form a regular set of trees.
This shows that the first step in the process of Figure \ref{alg:query-mondet-constraints} is always treelike and regular.
Recall that the BTVIP property states that the first $3$ steps of the process of the figure can be captured with a regular tree language, in the sense
that we can always find some representation that is treelike and regular.

We will make use of the fact that for $\fgtgd$s, the chase can be given a tree-like structure, a result that is well-known since the discovery
of the $Datalog^{\pm}$ family \cite{datalogpmj} and of $\fgtgd$s \cite{bagetwalking}. If we have a finite instance $\inst$  and $\fgtgd$s $\Sigma$,
the chase can be given the structure of a tree with the root containing all facts of $\inst$. We will ensure that
for every trigger for $\fgtgd$ $\tau$, it \emph{fires locally in the tree} -- its  image consists of facts in a single node $n$. The facts generated by that chase step
will be in a child $c$ of $n$, and $c$ will inherit all facts from $n$ that are guarded in $c$ -- that is, facts of $n$ whose are arguments
are contained in some fact of $c$.   In order to ensure that triggers always fire locally,
acts will be propagated from children to parents  but only facts that  are guarded in the parent. 
In particular, we have the following two properties:
\begin{itemize}
\item Bounded treewidth of the chase. Although the chase can be infinite,
the tree structure described above forms a tree decomposition in the usual sense, with the width
bounded by the size of $\inst$ and  the maximal size of a rule head in $\Sigma$.
\item  Guarded interfaces. For each subtree of the root, the elements of $\inst$ that appear in the subtree form a guarded
set in $\inst$.
\end{itemize}

In the case that $\inst$ is a CQ approximation tree $Q_n$ of a Datalog query $Q$, it already  comes with a tree decomposition whose width
is independent of $n$, by Proposition \ref{prop:boundedtreewidthapprox}. When we chase such an approximation
with $\fgtgd$s, we will do so preserving this tree structure, rather than putting all of the elements of $Q_n$ in the root note.
When  we make the first non-full chase steps, we will be creating children of nodes that may be in the interior of the
approximation tree: that is, we grow the chase tree  off of each node in the approximation.
Although the chase generates new facts on the domain of the canonical database of $Q_n$, by the guarded interface property above,
each such fact will be guarded by a fact in the canonical
database of $Q_n$. In particular, the following result,  will be used later in the appendix, see for example the proof of 
 Theorem  \ref{thm:datalogqfgdlviewsdecide}.

\begin{proposition} \label{prop:chaseapprox} For any Datalog query, there is a uniform bound on the treewidth of $\chase_\Sigma(Q_n)$.
We can effectively form an automaton on infinite trees that  accepts the canonical tree structures associated with chasing
the approximation trees of $Q_n$ in the way described above, or a WMSO sentence that describes  these trees.
\end{proposition}

\subsection{Proof of Proposition \ref{prop:pipeline}: the process of Figure \ref{alg:query-mondet-constraints} is correct}

Recall the process of Figure \ref{alg:query-mondet-constraints}:

\begin{figure}[h!]
\hspace{2mm}$\mondet(Q,\views, \Sigma)\textbf{:}\hfill$
\begin{algorithmic}[1]
    \For{$Q_n$ approximation of  $Q$} \Comment unfold the query
    \State $C_n \defeq \chase_{\Sigma}(\canondb(Q_n))$                                            \Comment Chase an unfolding
     \State $\viewinst_n \defeq \views(C_n)$                       \Comment Apply views 
        \For{$Q'_{m,n} \in \backview_{\views}(\viewinst_n)$}             
\Comment Guess a witness for each view fact
         \State $C'_{m,n} \defeq \chase_{\Sigma}(Q'_{m,n})$
\Comment Chase again
        \State{\kw{If} ~ {$C'_{m,n} \not \models Q$}\textbf{ return} \textbf{false}}        \Comment Check if $Q$ holds    
         \EndFor
      \EndFor
            \State \textbf{return} \textbf{true}
\end{algorithmic}
\caption{Process  for checking monotonic determinacy.}

\end{figure}

Many of the results  in the submission rest on Proposition \ref{prop:pipeline} in the body, stating that this process
is correct for monotonic determinacy with rules  over all instances. We recall the statement
of the proposition:

\medskip

$Q$ is monotonically determined over $\views$ w.r.t. $\Sigma$ if and only if the process
of Figure \ref{alg:query-mondet-constraints}
 returns true.

\medskip

We now overview the proof. \rebut{We stress that this it follows the same argument as Lemma 5.4 in \cite{uspods20journal}}. There are two extra steps involved in the process,
involving chasing, and these are reflected in the proof.

\rebut{In one direction, suppose that the process in the figure fails (returns false).  
Let $C_n$ and $C'_{m,n}$ as above witness this. Note that $\views(C_n) \subseteq \views(C'_{m,n})$ as $\views(Q'_{m,n}) = \views(Q_n)$ and $Q'_{m,n} \subseteq C'_{m,n}$. 
Moreover $C_m \models Q$ and $C'_{m,n} \not\models Q$. Thus $C_n$ and $C'_{m,n}$ serve as a counterexample to monotonic determinacy.}


\rebut{Before we prove the other direction, note the following propositions:}
\begin{proposition}\label{prop:there-is-a-subset-in-v-back-v}
For any instance $\inst$  there exists an instance $\inst' \in \backview_{\views}(\views(\inst))$ such that $\inst'$ homomorphically maps to $\inst$.
\end{proposition}

\begin{proposition}\label{prop:v-back-v-preserves-homs}
Let $\inst$ and $\inst'$ be two instances such that there is a homomorphism from $\inst$ to $\inst'$. Then for every instance $B \in \backview_{\views}(\views(\inst'))$ there exists an instance $A \in \backview_{\views}(\views(\inst))$ such that $A$ homomorphically maps to $B$.
\end{proposition}

Beginning the argument for the other direction, suppose that monotonic determinacy fails with witness instances $\inst_1$ and $\inst_2$.
\rebut{That is, $\inst_1 \models Q$, $\inst_1$ and $\inst_2$ satisfy $\Sigma$, the view image $\views(\inst_2)$ of $\inst_2$ contains all facts
in the view image $\views(\inst_1)$ of $\inst_1$, but $\inst_2$ does not satisfy $Q$. We shall show, that the above procedure returns false. Let $Q_n$ be an unfolding of $Q$ that  holds in $\inst_1$. Such a $Q_n$ exists because $\inst_1 \models Q$. 
We let $C_n$ and $\viewinst_n$ be as in Figure \ref{alg:query-mondet-constraints}. 
Note that there exists a homomorphism $h$ from $C_n$ to $\inst_1$. Consider $\views(h(C_n))$ and  $\backview_{\views}(\views(h(C_n)))$. 
From \cref{prop:there-is-a-subset-in-v-back-v} there exists an instance  $W$ of  $\backview_{\views}(\views(h(C_n)))$ that homomorphically maps to $\inst_2$ as $\views(h(C_n)) \subseteq \views(\inst_1) \subseteq \views(\inst_2)$. 
Finally note, from \cref{prop:v-back-v-preserves-homs}, that there exists an instance, call in  $Q'_{m,n}$, that is in
 $\backview_{\views}(\viewinst_n)$ and is mappable to $W$.  Note that $C'_{m,n} = \chase_{\Sigma}(Q'_{m,n})$ maps to $\inst_2$, thus $C'_{m,n} \not\models Q$. Therefore, the procedure returns false.}



\section{Additional material related to tools for positive results, Section \ref{sec:tools}}
\subsection{More details concerning the forward-backward method}

\myparagraph{Forward mapping results}
We recall the results on forward mapping, beginning with Proposition
 \ref{prop:kpipeline}: 

\medskip

 For each $Q$ in FGDL, $\views$ in FGDL or FO, $\Sigma$ TGDs, and each $k$ there is a tree automaton that accepts  all finite $k$ tree codes of countereexample
to monotonic determinacy of $Q, \views, \Sigma$. There is a nondeterministic B\"uchi automaton over infinite trees that holds exactly  when there is
an arbitrary (possibly infinite) $k$ tree code of such an instance.

\medskip

The argument works for a larger logic -- \emph{Guarded Second Order Logic (GSO)}.
A tuple is said to be \emph{guarded} in an instance if  there is a fact $F(c_1\ldots c_n)$ in the instance such that every component of the tuple is one of the $c_i$. A set of $k$-tuples (a $k$-ary relation) is said to be \emph{guarded} in an instance if every tuple in it is guarded.
Guarded Second Order logic is defined, following \cite{gho}, as a semantic restriction of SO requiring that the second order quantifiers range over guarded relations. For example, when the base schema has only a binary relation $R(x,y)$, GSO allows the usual MSO quantification and also quantification over sets of edges: thus GSO
can express the existence of paths. The following is a variant of
 \emph{Courcelle's theorem} \cite{courcellestheorem}:

\begin{theorem} \label{thm:courcelle}
 For every guarded second-order sentence $\varphi$ over the signature
$\sigma$ and $k \in \omega$ there is a tree automaton $\aut_\varphi$ over $\codes{\sigma}{k}$
 such that for all finite $\codes{\sigma}{k}$-trees $\tree$ we
have that $\aut_\varphi$ accepts $\tree$ iff $\tree$ is coherent and
$\decode{\tree} \models \varphi$.
\end{theorem}

Proofs can be found many places. For example, see Theorem 10 of \cite{readingcourse} for an argument that matches the phrasing  above.

\begin{proof}
The first part of the proposition follows immediately from Theorem \ref{thm:courcelle}. The statement
that an instance is a counterexample involves statements that:
\begin{itemize}
\item The rules are satisfied on both instances. Since the rules are TGDs, they are FO, hence in GSO.
\item The views predicates match their definitions.
 FGDL and FO are both subsets of GSO, and GSO is closed under Boolean operations and universal quantification. Thus
this is expressible in GSO.
\item The query holds in one instance and does not hold in another: this again is in GSO if $Q$ is in FGDL.
\item The view facts in one instance are contained in the view facts of the other instance: this is a simple universal
containment, hence in FO, hence in GSO.
\end{itemize}
Thus being a counterexample can be expressed in GSO and Theorem \ref{thm:courcelle} applies to it.

The second part of the theorem can be shown either by redoing the proof of Courcelle's theorem for infinite trees,
or by showing definability in Weak Monadic Second Order Logic, and using the result mentioned earlier in the appendix
that WMSO definability implies recognizabiilty by a nondeterministic B\"uchi Automata \cite{weakbuchi}.
\end{proof}

We now  recall Proposition \ref{prop:ksuffices}:

\medskip

If $(Q, \views, \Sigma)$ has the BTVIP, $\views$ in Datalog, then there is $k$, computable
from $(Q, \views, \Sigma)$, such that whenever monotonic determinacy
fails, there is e some counterexample of treewidth $k$.

\medskip

Note that the BTVIP is only defined when $Q$ is in Datalog and $\Sigma$ are TGDs, since
we need to talk about the chase of approximations.

\begin{proof}
We make use of the correctness of the process in 
Figure \ref{alg:query-mondet-constraints}.
By the BTVIP assumption there is $k$ such that for each approximation fo $Q$, 
there is a  finite chase of the view images with treewidth $k$. Let $A_0$ be this set of instances.

We then consider the instances $A_1$ obtained by adding on fresh witnesses for each view fact. Since we can
use new trees for these witnesses, the treewidth of  $A_1$ is also bounded by the maximum of $k$ and
the size of each rule body.

We now let $A_2$ be the result of chasing these instances
chasing the view facts of these instances with the backward views. Since $\Sigma$ are $\fgtgd$s, we can
perform the usual tree-like chase on it --  see, for example,  \cite{treelikechase} or Section \ref{app:chaseapproxregular} of this submission
 for a description.
Note that when we chase  an $\inst_2 \in A_2$ we are  adding on facts over values in $\inst_2$, but only facts
that are already guarded in $\inst_2$.  So we do not break the running intersection property of the existing decomposition of $\inst_2$.
In addition our chasing adds on new structure, possibly infinite. But the structure consists of components of bounded treewidth,
 interfacing with $\inst_2$ in a guarded set \cite{datalogpmj,bbc}. Thus we can simply union the tree decomposition
of the components  with the decomposition of $\inst_2$.
\end{proof}

\myparagraph{Backward mapping results}
We now recall Theorem \ref{thm:auttodlog} from the body, which deals again with the case of finite trees. It says
that if we start with a standard finite tree automaton,  we can go backward to  Datalog:

\medskip

 Let $\sigma$ be a relational signature, $k \in \omega$ and $\aut$ an
automaton for the tree signature $\codes{\sigma}{k}$ of $\sigma$. Then
there is a Datalog program $E_\aut$ such that for every
$\sigma$-structure $\struct$:  $\struct \models E_\aut$
iff there is a finite  active tree code $\tree$ over $\sigma$ with a
homomorphism from $\decode{\tree}$ to $\struct$ such that $\aut$ accepts
$\tree$.

\medskip

The  proof of this will be more involved, but we emphasize that it is essentially the
same argument as in Proposition 7.1 of \cite{uspods20journal}.

\begin{definition} \label{def:dlogfromaut}
 Given the $\codes{\sigma}{k}$-tree automaton $\aut$ we define a Datalog
program $E_\aut$ as follows: The extensional predicates are all the
predicates from $\sigma$. As intensional predicates we have the $0$-ary
goal predicate $\goal$, the $1$-ary predicate $\adom$, for every
bag $L$ the $k$-ary predicate $\local_L$ and $k$-ary
predicates $P_{q,g}$ for every state $q$ of $\aut$ and partial injective
function $g : \{1,\dots,k\} \to \{1,\dots,k\}$. The program contains the
following rules:
For every $n$-ary extensional predicate $R$ in $\sigma$ it contains
the rule
\[
 \adom(x_i) \datalogarrow R(x_1,\dots,x_n), \quad \mbox{where } i \in
\{x_1,\dots,x_n\}.
\]
For every accepting state $q$ of $\aut$ and partial injective $g$ it
contains the rule
\[
 \goal \datalogarrow P_{q,g}(x_1,\dots,x_k).
\]
For every bag $L$ we have the rule
\[
 \local_L(x_1,\dots,x_n) \datalogarrow \bigwedge_{i = 1}^{n} \adom(x_i)
\land \bigwedge_{T^=_{l_0,l_1} \in L} x_{l_0} = x_{l_1} \land
\bigwedge_{T^R_{\ntuple} \in L} R(x_\mtuple),
\]
where $x_\mtuple = x_{l_1},\dots,x_{l_m}$ for $\mtuple =
(l_1,\dots,l_m)$. For every transition of the form $\code^g_L
\rightarrow q$ in $\aut$ the program $E_\aut$ contains the rule
\[
 P_{q,g}(x_1,\dots,x_k) \datalogarrow \local_L(x_1,\dots,x_k).
\]
For every transition of the form $q_1,\dots,q_r,\code^g_L \rightarrow q$
in $\aut$ and for all partial injections $g^1,\dots,g^r$ it contains the
rule
\begin{equation} \label{eqn:inductive_rule}
 P_{q,g}(x_1,\dots,x_k) \datalogarrow  \bigwedge_{j = 1}^r \left[
P_{q_j,g^j}(x_1^j,\dots,x_k^j) \land \bigwedge_{l \in \dom(g^j)} x_l =
x^j_{g^j(l)}   \right] \land \local_L(x_1,\dots,x_k).
\end{equation}
\end{definition}

We refer to the   Datalog query $E_A$ as the \emph{backward mapping Datalog query}, and
the rules as \emph{backward mapped rules}.

\begin{lemma} \label{lem:datalog_to_automaton}
 If $\struct \models E_A$ then there is an active tree code $\tree$ with
a homomorphism from $\decode{\tree}$ to $\struct$ such that $A$ accepts
$\tree$.
\end{lemma}
\begin{proof}
 Assume that $\struct \models E_A$. We define the tree $\tree$ by
induction on the evaluation tree that captures that witnesses $\struct
\models E_A$. We show that every intensional predicate of the form
$P_{q,g}(x_1,\dots,x_k)$ that holds of some tuple $(s_1,\dots,s_k)$ of
elements of $\struct$ there is a $\codes{\sigma}{k}$-tree $\tree$
together with a run of $\aut$ that ends at the root $v$ of $\tree$ in
state $q$ and a homomorphism $h$ from $\decode{\tree}$ to $\struct$ such
that $h([v,l]) = s_i$ for all $l \in \{1,\dots,k\}$. This yields the
desired result, because $\goal$ can only be true if
$P_{q,g}(x_1,\dots,x_k)$ for an accepting state $q$ is true of some
elements in $\struct$.

In the base case $P_{q,g}(x_1,\dots,x_k)$ holds of the tuple
$(s_1,\dots,s_k)$ because of the rule $P_{q,g}(x_1,\dots,x_k) \datalogarrow
\local_L(x_1,\dots,x_k)$ for some leaf transition $\code^g_L \rightarrow
q$. In this case we define $\tree$ to consist of just one leaf $v$ that
is labelled with $\code^g_L$. Because of the transition $\code^g_L
\rightarrow q$ there is a run of $\aut$ on $\tree$ that ends in $q$.
This definition of $\tree$ means that the structure $\decode{\tree}$ has
as its domain the elements $[v,1],\dots,[v,k]$ and a predicate $R$ from
$\sigma$ is defined to be true of $([v,l_1],\dots,[v,l_m])$ iff
$T^L_\ntuple \in L$ for $\ntuple = (l_1,\dots,l_m)$. We define the
homomorphism $h$ such that $h([v,l]) = s_l$ for all $l \in
\{1,\dots,k\}$. This is well-defined as a function because whenever
$[v,l] = [v,l']$ then this is because there is the predicate
$T^=_{l,l'}$ in the label $\code^g_L$ of $v$ and hence $s_l = s_{l'}$
holds, because $\local_L(x_1,\dots,x_k)$ is true of $s_1,\dots,s_k$. The
function $h$ preserves all the facts in $\decode{\tree}$ because the
last conjunct of $\local_L$ requires the facts from $L$ to hold at
$(s_1,\dots,s_k)$ in $\struct$.

In the inductive case $P_{q,g}(x_1,\dots,x_k)$ holds of the tuple
$(s_1,\dots,s_k)$ because of the rule from \eqref{eqn:inductive_rule},
which is in the program $E_A$ because of the transition
$q_1,\dots,q_r,\code^g_L \rightarrow q$ in $\aut$. We then have that for
every $j \in \{1,\dots,r\}$ there is some partial injection $g^j$ such
that the predicate $P_{q_j,g^j}(x^j_1,\dots,x^j_k)$ holds of some
elements $s^j_1,\dots,s^j_k$ in $\struct$. By the inductive hypothesis
we obtain for every $j$ a tree code $\tree_j$ and a homomorphism $h_j$
from $\decode{\tree_j}$ to $\struct$ such that $\aut$ has a run in
$\tree_j$ ending with $q_j$ and $h_j([v^j,l]) = s^j_l$ for all $l \in
\{1,\dots,k\}$. We define $\tree$ such that it consists of the disjoint
union of the trees $\tree_j$ for all $j \in \{1,\dots,r\}$ together with
a new root $v$ that is labelled with $\code^g_L$ and has all the roots
$v_j$ of the $\tree_j$ as its children. Note that with this construction
we have that $[v,l] = [v_j,g^j(l)]$ holds in $\decode{\tree}$ for all $j
\in \{1,\dots,r\}$ and $l \in \dom(g^j)$. There is a run of $\aut$ in
$\tree$ ending at $q$ because we can combine all the runs in the
$\tree_j$ ending in $q_j$ and use the transition
$q_1,\dots,q_r,\code^g_L \rightarrow q$. The homomorphism $h$ from
$\decode{\tree}$ to $\struct$ is defined such that $h([v,l]) = s_l$ for
all $l \in \{1,\dots,k\}$ and for $w \in \tree_j$ we set $h([w,l]) =
h_j([w,l])$. To show that this is well-defined one uses the observation
that whenever some $[w,l]$ for $w \in \tree_j$ is equal to either some
$[v,l']$ or some $[w',l']$ for $w' \in \tree_{j'}$ with $j \neq j'$ then
$[w,l] = [v_j,l'']$ for some local name $l''$ such that $l'' = g^j(l')$
for some $l' \in \dom(g^j)$. Then one exploits the equality $x_{l'} =
x^j_{g^j(l')}$ in \eqref{eqn:inductive_rule}. For all equality that are
enforced on local names at $v$ we use the corresponding conjunction in
the body of $\local_L(x_1,\dots,x_k)$, which holds of $s_1,\dots,s_k$
Our definition yields a homomorphism for the facts in the bag of $v$
because of the corresponding conjunct in the body of
$\local_l(x_1,\dots,x_k)$.
\end{proof}

\begin{definition}
 Let $\tree$ be an active tree code and fix an element $b$ in the active
domain of $\decode{\tree}$. For every $v$ in the tree $\tree$ we define
a $k$-tuple $a^v = (a_1,\dots,a_k)$ of elements of $\decode{\tree}$ such
that $a_l = [v,l]$, if $[v,l]$ is in the active domain of
$\decode{\tree}$, and $a_l = b$ otherwise.
\end{definition}

\begin{lemma} \label{lem:local}
 For every node $v$ in an active tree code $\tree$ and bag $L$ the
predicate $\local_L(x_1,\dots,x_k)$ is true of $a^v$ in
$\decode{\tree}$.
\end{lemma}
\begin{proof}
 We show that the body of
\[
 \local_L(x_1,\dots,x_n) \datalogarrow \bigwedge_{i = 1}^{n} \adom(x_i)
\land \bigwedge_{T^=_{l_0,l_1} \in L} x_{l_0} = x_{l_1} \land
\bigwedge_{T^R_{\ntuple} \in L} R(x_\ntuple),
\]
is true of $a^v = (a_0,\dots,a_k)$. By the statement of the lemma it is
guaranteed that every $a_l$ for $l \in \{1,\dots,k\}$ is in the active
domain of $\decode{\tree}$. Therefore, the first conjunct of the body
holds. That the other conjuncts about equality and the atomic relation
symbols hold follows immediately from the definition of
$\decode{\tree}$.
\end{proof}

\begin{lemma} \label{lem:dlogontree}
 If the active tree code $\tree$ is accepted by $\aut$ then
$\decode{\tree} \models E_\aut$.
\end{lemma}
\begin{proof}
 To prove that $\decode{\tree} \models E_\aut$, we show with an induction
on the depth of a run that if there is a run of $\aut$ ending at some
node $v$ of $\tree$ in a state $q$ of the automaton then for every
partial injection $g$ the predicate $P_{q,g}(x_1,\dots,x_k)$ is true of
$a^v$ in $\decode{\tree}$. Because of the rule $\goal \datalogarrow
P_{q,g}(x_1,\dots,x_k)$, for accepting $q$, this then means that if
there is an accepting run of $\aut$ at the root of $\tree$ then $E_A$ is
true in $\decode{\tree}$.

In the base case of the induction we have a run that ends at a leaf $v$
of $\tree$ and corresponds to a transition of the form $\code^g_L
\rightarrow q$ in $\aut$. It then follows immediately that
$P_{q,g}(x_1,\dots,x_k)$ is true of $a^v$ because by Lemma~\ref{lem:local}
the predicate $\local_L(x_1,\dots,x_k)$ is true of $a_v$ and the program
$E_\aut$ contains the rule $P_{q,g}(x_1,\dots,x_k) \datalogarrow
\local_L(x_1,\dots,x_k)$.

In the inductive case we have a run ending in a state $q$ and at a node
$v$ of $\tree$ that is not a leaf. We need to establish that for all
partial injections $g$ the predicate $P_{g,q}(x_1,\dots,x_k)$ is true of
$a^v$ in $\decode{\tree}$. Consider the last transition of the run that
ends at $v$, which must be of the form $q_1,\dots,q_r,\code^g_L
\rightarrow q$. Thus, for each of the children $v_1,\dots,v_r$ of $v$ in
$\tree$ we have that there are accepting runs of $\aut$ at $v_j$ ending
in state $q_j$. By the induction hypothesis we have that for each $j$
and any partial injection $g'$ that the predicate
$P_{q_j,g'}(x^j_1,\dots,x^j_k)$ is true of $a^{v_j}$ in
$\decode{\tree}$. In particular this holds for the partial injection
$g^j$ that is at the label $\code^{g^j}_{L_j}$ of the child $v_j$,
meaning that the predicate $P_{q_j,g'}(x^j_1,\dots,x^j_k)$ is true of
$a^{v_j}$. Moreover, by the construction of $\decode{\tree}$ we have
that $[v,l] = [v_j,g^j(l)]$ for all $l \in \dom(g^j)$. Therefore, if
$[v,l]$ is in the active domain of $\decode{\tree}$ then so is
$[v^j,g^j(l)]$ because it is the same element. It follows that in this
case the equality $x_l = x^j_{g^j(l)}$ holds. If on the other hand
$[v,l]$ is not in the active domain of $\decode{\tree}$ then
$[v_j,g^j(l)]$ is also not in the active domain and the equality $x_l =
x^j_{g^j(l)}$ holds because both variables are interpreted as $b$. Our
reasoning so far shows that for every $j$ the first conjunct in the body
of \eqref{eqn:inductive_rule} holds. For the second conjunct,
$\local_L(x_1,\dots,x_k)$, we can
use Lemma~\ref{lem:local}.
\end{proof}

\begin{proof}[Proof of Theorem ~\ref{thm:auttodlog}]
Let $E_\aut$ be defined from $\aut$ as in Definition~\ref{def:dlogfromaut}. 

The direction from left to right is proven in Lemma~\ref{lem:datalog_to_automaton}.

For the other direction assume that there is a tree code $\tree$ such
that $\aut$ accepts $\tree$ and there is a homomorphism from
$\decode{\tree}$ to $\struct$. By Lemma~\ref{lem:dlogontree} it follows
that $\decode{\tree} \models E_\aut$. Because Datalog is preserved under
homomorphisms it follows that $\struct \models E_\aut$.
\end{proof}

\subsection{Proofs of Certain Answer Rewriting results} \label{subsec:app-certain}

\myparagraph{A note on finite and infinite}
We note  that in the notion of certain answer rewriting, we are quantifying twice over instances or finite instances.
In the usual definition of rewriting, we say
$R$ is a certain answer rewriting of $Q$ w.r.t. $\Sigma$ if:

\medskip

for all finite instances $\inst$, $R(\inst)$ holds if and only if for every $\inst'$ extending $\inst$ with
$\inst'$ satisfying $\Sigma$, $\inst'$ satisfies $Q$.

\medskip

The second quantification over instances $\inst'$ cannot in general be exchanged with a quantification over finite
instances: this  exchange is valid when the rules are finitely controllable.

But in the first quantification, the distinction between finite instances and infinite instances will usually not be important
to us:

\begin{proposition}
If $\Sigma$ is a set of existential rules, and $Q$ and $R$ are in Datalog, ``finite instances'' can be exchanged with ``instances'' without
impacting the definition.
\end{proposition}

\begin{proof}
Clearly if we have the equivalence for all instances, we have it for finite instances. Conversely, suppose we have
an infinite $\inst$ where it fails. One possibility for failure is that $R$ holds on $\inst$ but the right hand implication does not hold.
But then the is some finite $\inst_0$ subinstance of $\inst$ where $R$ holds, and the right hand implication follows,.
The other possibility for failure is that the implication on the right holds for some
infinite instance $\inst$, but $R$ fails to hold on $\inst$.  But then $\chase_\Sigma(\inst)$ satisfies $Q$, since the chase
characterization still holds for Datalog $Q$ and infinite instances. Some approximation $Q_n$ of $Q$ witnesses this, and
only a finite set of facts $\inst_0$ from $\inst$ where needed to generate the match of $Q_n$ in the chase.
Thus $\inst_0$  is a counterexample of the statement for finite instances.
\end{proof}

We will make use of this propositions in many of the arguments that follow, since we will apply the certain answer rewritings
within pipelines of transformations that produce infinite instances.

We now turn to proving the results on certain answer rewriting from Section \ref{sec:tools} in the body of the paper.

Recall the statement of Theorem \ref{thm:rewritefrontieronemdl}:

\medskip

CQs have MDL $\certainanswerrewriting$s over   frontier-one TGDs.

\medskip

\begin{proof}[Sketch]
The rewriting algorithm follows the approach  in Theorem 5.5 and 5.6 of \cite{bbc}, where the same result is shown for guarded and frontier-guarded
TGDs, with the resulting rewriting being in frontier-guarded  Datalog.

Let $\Sigma$ be the rules. We can assume that our Boolean query  $Q$ consists of the atomic $0$-ary predicate $\goal$: if it is not, we
add a new TGD to $\Sigma$ of the form $Q \rightarrow \goal$

For simplicity we deal with the case where there are no constants in the query or in the rules.
Let us refer to a fact over our initial instance $\inst$ as a ``ground fact''.
We can see, arguing from the chase, that for any $\inst$,
ground facts that are inferred by $\Sigma$ are ``unary'': that is,
they have only one element in them, like $U(c)$ or $R(c,c)$.

Say that an instance $\inst$ is \emph{fact saturated with respect to $\Sigma$} if every  ground fact 
that is derived from $\inst$ and $\Sigma$ is already in $\inst$.

In deriving a new ground fact, from $\Sigma$, we may make use of an unbounded number of
facts in $\inst$: think of a graph $G(x,y)$ and $\Sigma$ with reachability rules, 
of the form s $Reaches(x) \wedge G(x,y) \rightarrow Reaches(y)$. To derive $Reach(a)$ we may use many $G$ facts.
 However, to achieve fact saturation with respect to frontier-guarded rules,
it suffices to ensure saturation for small sets. In the above example, once we know that the facts on each pair of nodes
cannot generate any new facts, we know we are fact-saturated.
This is capture in
the ``bounded base lemma'', Lemma 5.8 of \cite{bbc}

\begin{lemma} \cite{bbc}
Letting $k$ be the maximal number of variables in a TGD of $\Sigma$,
and let $A$ be an instance such that for each subset $X$ of the active domain of $A$ with
$|X|  \leq k$, the induced substructure $A \upharpoonright$ is fact-saturated with respect to $\Sigma$
Then $A$ is fact-saturated with respect to $\Sigma$.
\end{lemma}

We now create a Datalog program $P_\Sigma$ that infers all ground facts.  We include:
\begin{itemize}
\item intensional predicates $U_A$ for each atom with a  single free variable $x$,
along with ``copy rules'' of the
form $U_A(x) \datalogarrow A(x)$.  Similarly for every $0$-ary predicate in the signature.
\item all  Datalog rules over the signature with the extensional predicates and the intensional predicates above,
with at most $k$ atoms in the body and a unary or $0$-ary predicate $U_A$ in the head, such that the rules follow from $\Sigma$ and the copy
rules.  We can decide whether  such a rule is derived using a decision procedure for the appropriate guarded logic, for example
the guarded negation fragment of first order logic (GNFO).
\item $U_\goal$ as the goal predicate
\end{itemize}

Clearly, the program only derives $U_A(x)$ when $A(x)$ is derivable from $\Sigma$.
To show the converse, it suffice to show that for any instance $\inst$ after running the Datalog program until a fixpoint is reached, and
applying the ``reverse copy rule'' $U_A(x) \rightarrow A(x)$, the resulting instance
is fact-saturated with respect to $\Sigma$.
By the lemma, it suffices to show that for each induced subinstance $A_k={a_1 \ldots a_k}$
with domain of size $k$, $A_k$  is fact-saturated with respect to $\Sigma$ and the reverse rules. 
So fix $A_k$ and a fact $F$ that is derived from $A_k$
using $\Sigma$ and the reverse rules. Clearly, the reverse rules are only useful to derive $F$ at the end, since we can replace
any use of $A$ by $U_A$ in a derivation. 
 But $A_k \rightarrow F$ is a consequence of $\Sigma$. Thus $A_k \rightarrow U_F$ will be in $P_\Sigma$, and
thus $F$ will be derived from $P_\Sigma$ and the reverse rules.
We conclude $A_k$ is fact-saturated, and we are done.
\end{proof}

Recall the statement of Proposition  \ref{prop:fgdlcertain}:

\medskip

 For any set $\Sigma$ be a set of frontier-guarded  TGDs and any Boolean $\fgdatalog$
query $Q$ there is a Boolean $\fgdatalog$ query that is a certain answer rewriting of $Q$ $\wrt$ $\Sigma$.

\medskip

\begin{proof}
As mentioned earlier, the corresponding assertion for a CQ query is in \cite{bbc}.
We reduce to this case, by moving the rules of $\fgdatalog$ $Q$ into additional rules of $\Sigma$, arriving at 
set of rules $\Sigma'$,
and letting $Q'$ be just $\goal$.  Note that since the rules of $Q$ are frontier-guarded, the rules $\Sigma'$
are all frontier-guarded TGDs. We claim that for any instance $\inst$,

\begin{align*}
\inst \wedge \Sigma \models Q \leftrightarrow \inst \wedge \Sigma' \models \goal
\end{align*}
Where $\inst$ refers to the conjunction of facts in the instance.

The equivalence is a standard observation about entailment in Datalog and entailment in Horn clause logic.
Note that the right hand entailment asserts that  for any $\inst$, and an arbitrary superinstance of $\inst$ satisfying $\Sigma$ and
satisfying the rules of $Q$ as Horn clauses, $\goal$ holds.
The left hand entailment asserts that for every instance $\inst$ satisfying $\Sigma$, the instance $\inst'$
that is the least fixed point of the rules satisfies $\goal$.
Thus for the right to left direction, we just note that  $\inst'$ is a superinstance of a structure satisfying $\Sigma$ which satisfies the rules,
and thus the right hand entailment implies to conclude $\goal$.
For the other direction, suppose we have  a counterexample to the right hand entailment, consisting of $\inst$ satisfying $\Sigma$
and a superinstance $\inst^*$ satisfying the rules of $Q$, where $\neg \goal$ holds in $\inst^*$. Since $\inst^*$ is a fixpoint, setting
$\inst'$ to be the least fixed point of $\inst$ under the rules. Since we have only removed facts, $\neg \goal$ still holds.
But then the left is contradicted.
\end{proof}

\myeat{
\medskip

Recall the statement of Proposition  \ref{prop:mdltoextconn} 

\medskip

Suppose that $\Sigma= \Sigma_1 \cup \Sigma_2$
where $\Sigma_2$ consists of source to target full TGDs with the body of each TGD being a connected
CQm and $\Sigma_1$ linear source-to-target TGDs with target disjoint from sources of $\Sigma_2$.
Then MDL queries have EC-Datalog $\certainanswerrewriting$s over $\Sigma$.


The result is motivated by the case where $\Sigma_2$ are definition-to-view rules that populate
a view predicate with with a connected CQ view definition, while $\Sigma_1$ are rules that go in the other direction, from
a view definition to its definition, in a copy of the original  signature.

\begin{proof}
We first rewrite MDL $Q$ with respect to $\Sigma_1$
into a $\fgdl$ query $Q'$. This can be done using the standard
inverse-rules algorithm, which produces something in MDL, as shown in \cite{uspods20journal}. Then we rewrite $Q'$ with respect to
$\Sigma_{\mathsf{forward}}$ into a Datalog query $Q''$ by substituting the EDB predicates in $Q'$ by the corresponding rule
bodies.
From connectedness of the bodies of $\Sigma_2$, it is clear that $Q''$ is extensionally-connected.
\end{proof}
}

\subsection{Proof of Theorem \ref{thm:fcecdatalog}: Finite controllability of EC Datalog under Frontier One TGDs}

In this section we present the proof of Theorem \ref{thm:fcecdatalog}:

\medskip

 For the class of frontier-one TGDs $\Sigma$ and EC-Datalog queries $Q = (\Pi, \goal )$ entailment
is finitely controllable.

\medskip

For the remainder of this section, fix the database instance $\db$, the set of frontier-one TGDs $\rs$, and an EC Datalog query $\dl$. Let $\ch$ denote the chase of $\db$ by $\rs$, let $n$ denote the maximal size of any rule in both $\rs$ and $\dl$, and let $N = 4 \cdot n^2$.\footnote{The constant $4$ is chosen here for convenience and is not in any way considered to be a \emph{tight} value.}

In the upcoming section, we will introduce a few minor assumptions. Near the conclusion of the proof, we will discuss how to remove  these assumptions.

\begin{enumerate}
	\item The database instance $\db$ is a single unary atom. \label{assume:singleun}
	\item The EDB relational symbols are at most binary. \label{assume:binary}
	\item The rules of $\rs$ have exactly one frontier variable. \label{assume:uniquefront}
	\item No atoms of the shape $E(x,x)$ appear in $\db, \rs,$ and  $\dl$ \label{assume:noself}
	\item Heads of rules of $\rs$ are trees. \label{assume:treerule}
\end{enumerate}

\textbf{In this section, we shall disregard the direction of binary predicates whenever discussing graph-theoretic notions, such as distances, trees, or cycles.}

\begin{proposition}\label{prop:fc-chase-tree}
The chase of $\db$ by $\rs$ forms a regular tree.
\end{proposition}
\begin{proof}
From Assumptions 1 to 5.
\end{proof}

\begin{definition}[Unabridged]
Given an instance $\inst$ and a cycle $C$ in $\inst$ we say that $C$ is {\em unabridged} if for every pair of elements 
$t,t'$ of $C$ the shortest path between them in $\inst$ goes through $C$. 
\end{definition}
Note that in the above we treat instances as undirected graphs.

\begin{definition}[Trimming]
Consider a tree $T$ with a node $v$. The process of removing all the children of node $v$ is referred to as {\em trimming} at node $v$. Note, that trimming at leaves is allowed, but it has no effect.
\end{definition}

\begin{definition}[Unfolding tree]
For a Datalog program $P$ that lacks repeated variables in its head, the \emph{unfolding tree} is defined as any tree derived from a CQ-approximation tree $T$ of $P$ through the process of trimming at one or more nodes of $T$. 

Then, an \emph{unfolding} of $P$ is a conjunction of labels from any unfolding tree of $P$. Note that the resulting CQ may include IDB predicates. 
\end{definition}

\begin{definition}[Succession]
We say that an unfolding tree $T$ is a \emph{direct successor} of an unfolding $T'$ if $T'$ can be obtained from $T$ by trimming it at a node with only leaves as its children. We define \emph{succession} as the transitive closure of direct succession. These notions naturally extend to unfoldings.
\end{definition}

\begin{proposition}\label{prop:fc-edb-idb-distance-preserved}
Let $U$ be an unfolding of $\dl$. Let $x$ and $y$ be distinct variables of $U$.
Let $p$ be a path in $U$ using binary EDB predicates,  
and suppose $p$ connects $x$ and $y$. Then
\begin{align*}\tag{$\clubsuit$} \label{eq:appfcecd-prop1}
|p| \leq n \cdot dist_{U}(x,y),\;
\end{align*}
\end{proposition}
\begin{proof}
First, note the following:
\begin{itemize}
	\item[$(\diamondsuit)$] For any three variables $x$, $y$, and $z$, and two paths $p$ (from $x$ to $y$) and $p'$ (from $y$ to $z$), if both $x, y, p$ and $y, z, p'$ satisfy \cref{eq:appfcecd-prop1}, then $x, y$, and the concatenation of $p$ and $p'$ also satisfy \cref{eq:appfcecd-prop1}.
	\item[$(\heartsuit)$] Bodies of rules of $\dl$ satisfy \cref{eq:appfcecd-prop1}.
\end{itemize}

We will prove \cref{prop:fc-edb-idb-distance-preserved} by a simple inductive argument.

(base) Let $U$ be a single-atom unfolding - it trivially satisfies \cref{eq:appfcecd-prop1}.

(ind. step) Let $U'$ be any direct predecessor of $U$ and let $U \wedge \alpha = U'$ where $\alpha$ is a conjunction of atoms. Let $\vec{x}$ be the tuple of variables shared between $U$ and $U'$. Note that $\alpha$ satisfies \cref{eq:appfcecd-prop1} from $(\heartsuit)$ and $U'$ satisfies \cref{eq:appfcecd-prop1} by ind. hypothesis.

1) If $\vec{x}$ is empty then $U$ trivially satisfies \cref{eq:appfcecd-prop1}.

2) If $\vec{x}$ consists of one or two variables $x$ then $U$ satisfies \cref{eq:appfcecd-prop1} from $(\diamondsuit)$

\end{proof}

\begin{lemma}\label{lem:fc-no-unabridged-cycles}
Every unfolding of $\dl$ contains no unabridged cycles of length greater than or equal to $N$ when restricted to EDB relations.
\end{lemma}
\begin{proof}
We will use \cref{prop:fc-edb-idb-distance-preserved} to prove the lemma. Assuming, for the sake of contradiction, that the statement above does not hold, let $U$ be the earliest counterexample --- meaning no unfolding that is a predecessor of $U$ is a counterexample.

Let $U'$ be any direct predecessor of $U$ and let $U' \wedge \alpha = U$ where $\alpha$ is a conjunction of atoms.
Note, as $U'$ is not a counterexample, it contains no large unabridged cycle over EDB relations. However, $U$ does contain such a cycle $C$. Thus, there exist two distinct variables $x$ and $y$ shared between $U'$ and $\alpha$ such that $C$ passes through $x$ and $y$ in $U$. Let $p$ be the path over EDB relations connecting $x$ and $y$ in $U'$. Note that $|p| \geq N - |\alpha|$ and $dist_{U}(x,y) \leq |\alpha| \leq n$. Together with \cref{prop:fc-edb-idb-distance-preserved} we get a contradiction: $4 \cdot n^2 - n \leq |p| \leq n \cdot dist_{U}(x,y) \leq n^2$.
\end{proof}

In order to prove Theorem \ref{thm:fcecdatalog} we shall show the following:
\begin{lemma}\label{lem:fc-main}
If $\ch$ does not entail $\dl$ then there exists a finite model $M$ of $\rs$ that does not entail $\dl$ as well.
\end{lemma}

The proof of the above lemma is the heart of the argument, and will require some machinery.

\begin{definition}[Ancestor]
Given an infinite tree $\mathcal{T}$, a natural number $m$, and a node $u$ of $\mathcal{T}$, we define the \emph{$m$-ancestor} of $u$ as the ancestor of $u$ at a distance of $m$, if it exists; otherwise, the $m$-ancestor of $u$ is the root of $T.$.
\end{definition}

\begin{definition}[Perspective]
Given a natural number $m$ and a node $u$ of an infinite tree $\mathcal{T}$, we define the \emph{$m$-perspective} of $u$ as the pair $\pair{T', u}$ where $T'$ is the subtree of $T$ that is rooted at the $m$-parent of $u$. We consider $m$-perspectives up to isomorphism.
\end{definition}

Let $type(u)$, for a term $u$ in $\ch$, consist of two values:

\begin{itemize}
	\item The depth of $u$ in $\ch$ modulo $N$.
	\item The $N$-perspective of $u$. Note that \cref{prop:fc-chase-tree} indicates there is only a finite number of such, keeping in mind that we count only up to isomorphism.
\end{itemize}
Define $M$ as a structure that is a quotient of $\ch$ using the ``is of the same type'' relation, where $type$ is defined as above.

\begin{proposition}\label{prop:m-is-a-model}
$M$ is a model of $\rs$.
\end{proposition}
\begin{proof}
Note that $N$-distance neighborhoods of elements of $\ch$ before and after identification are homomorphically equivalent to each other. Thus for any element of $\ch$ the satisfaction of unary CQs of size smaller than $N$ is identical before and after taking the quotient. Thus the result follows from the fact that $\rs$ is frontier-one, as trigger activation relies on satisfaction of unary CQs of size no greater than $n$. 
\end{proof}

\begin{proposition}\label{prop:fc-on-unabridged-cycles}
Any instance that can be mapped to $M$ but not to $\ch$ contains an unabridged cycle of length at least $N$.	
\end{proposition}
\begin{proof}

\textit{Claim: $M$ does not contain an unabridged cycle of length smaller than $N$.}

Assume, towards contradiction, that $M$ contains an unabridged cycle of length smaller than $N$. Let $C$ be the shortest such cycle with a length of $m$, and let $\mathcal{P}$ be a sequence of simple paths in $\ch$ that demonstrates the existence of $C$. Note that each vertex appears only once among the paths in $\mathcal{P}$.

Assume that $\mathcal{P}$ contains a single element, $P$. Let $u$ and $u'$ be, respectively, the first and last vertices of $P$. Note that, as $P$ is simple and shorter than $N$, and the depth of $u$ and $u'$ modulo $N$ is equal, there exists another vertex, $v$, in $P$ that is the common ancestor of both $u$ and $u$'. Let $w$ and $w'$ be children of $v$ on paths toward $u$ and $u'$. Observe that $w$ and $w'$ have equal types, and thus are identified in $M$. Therefore, $C$ is not a simple cycle, leading to a contradiction.

Assume that $\mathcal{P}$ contains more than one path, and let $P$ and $P'$ be its two first elements. We will show that there exists $\mathcal{P}'$ containing one less path than $\mathcal{P}$ which is witnessing the existence of another unabridged cycle of length $m$. Let $u$ be the last element of $P$, and let $u'$ be the first element of $P'$. As $u$ and $u'$ are identified, one can find a path equivalent to $P'$ in the $N$-perspective of $u$ that can be used instead. Thus, using a simple inductive reasoning, one can reduce this case to the above.

Finally, note that any instance that contains no cycle and can be mapped to $M$ can also be mapped to $\ch$. Using the above claim, one can conclude the proposition.
\end{proof}

Now, we are ready to prove \cref{lem:fc-main}. Assume, for the sake of contradiction, that $M$ entails $\dl$, and let $U$ be an unfolding of $\dl$ that serves as a witness. From \cref{prop:fc-on-unabridged-cycles}, we know that $U|_{EDB}$ contains an unabridged cycle of length at least $N$. From \cref{lem:fc-no-unabridged-cycles} we get a contradiction.

\myparagraph{On lifting assumptions}
The assumptions can be lifted using the following pre-processing steps.

\smallskip
\noindent
\textbf{1. The database instance $\db$ is a single unary atom:}\\
The database can be saturated $\db$ under Datalog rules first. Then, each constant can be placed into a unique unary atom. With some trivial but necessary adjustments to $\rs$, one can consider each database constant separately.

\smallskip
\noindent
\textbf{2. The EDB relational symbols are at most binary:}\\
This can be achieved through a simple reification. 
For example, the atom $R(x, y, z)$ can be transformed into the following conjunction:
$$R_1(w, x)\wedge R_2(w, y)\wedge R_3(w, z).$$

\smallskip
\noindent
\textbf{3. The rules of $\rs$ have exactly one frontier term:}\\
Trees built over atoms resulting from the application of frontierless rules can be considered separately, as in step 1.

\smallskip
\noindent
\textbf{4. No atoms of the shape $E(x,x)$ appear in $\db, \rs$, and  $\dl$:}\\
Similarly to step 2, one can replace $E(x, x)$ with $E_{loop}(x)$ in every atom of the input. Again, a trivial but necessary surgery of the ruleset and the Datalog query might be required.

\smallskip
\noindent
\textbf{5. Heads of rules of $\rs$ are trees:}\\
This can be ensured during reification in step 2.

\subsection{Proof of Theorem \ref{thm:compact}: compactness of entailment for views and rules} \label{subsec:app-compact}

Recall the statement of Theorem \ref{thm:compact} [Compactness for view and rule entailment]:

\medskip

Let $\Sigma$ consist of frontier-guarded TGDs, $\views$ a set of views defined
by Datalog queries, and $Q$ a $\fgdl$ query.

For every $\vinst$ 
such that $\vinst$ is $Q$-entailing with respect to $\Sigma$, $\views$, there is
$\vinst_0$ a finite subinstance of $\vinst$  that is $Q$-entailing with respect to $\Sigma$, $\views$.

\medskip

We write $\inst \equiv_k \inst'$ if Duplicator has a winning
strategy in the EF-game over $k$ rounds for guarded second order logic.
In every round of this game Spoiler can choose to either pick
an element or a guarded set in one of the instances. A guarded set is a
subset $R$ of $D^n$ for some $n \geq 1$, where $D$ is the domain of the
instance and every tuple in $R$ is already in the interpretation of one
of the relations symbols. Duplicator has to reply with an choice of the
same type in the other instance. After $k$ rounds Duplicator wins if
the pairs of elements that have been picked in the rounds where Spoiler
choice to play an element are a partial isomorphism. We clarify the 
notion of partial isomorphism we use. F
Assume that in the match the pairs
$(a_1,a_1'),\dots,(a_h,a_h')$ of elements have been played. We  write
$c^\inst$ for the interpretation of some constant $c \in \Lang$ in
 $\inst$. We require that the relation
\[
 \{(a_1,a_1'),\dots,(a_h,a_h')\} \cup \{(c^\inst,c^{\inst'}) \mid
\mbox{ for some constant } c \in \Lang\}
\]
is a partial isomorphism in the usual sense. We also need that the partial
isomorphism preserves the guarded relations that have been chosen in any
of the rounds, where Spoiler chose a guarded subset of $D^n$.

A basic result is that $\inst \equiv_k \inst'$ iff $\inst$ and
$\inst'$ cannot be distinguished by formulas of quantifier depth at
most $k$. Here, we do not distinguish between first and second order
quantifiers when computing the quantifier depth of a formula. If the
language is finite then one can show with an induction on $k$ that,
up-to logical equivalence, there are only finitely many formulas of
quantifier depth $k$. Thus, $\equiv_k$ has a finite index.

Given two instances $\inst$ and $\inst'$ for $\Lang$ their
\emph{disjoint sum} $\inst + \inst'$ is defined
as follows: Its domain is the disjoint union of the domains of $\inst$
and $\inst'$ quotiented by the smallest equivalence relation
$\sim$ that is generated by all pairs $a_1 \sim a_2$ such that for some
constant $c$ of $\Lang$ the element $a_1$ is the interpretation of $c$
in $\inst$, and $a_2$ is the interpretation of $c$ in $\inst'$. A
tuple $(u_1,\dots,u_n)$ is in the interpretation of an $n$-ary relation
symbol $P$ if there representatives $a_1,\dots,a_n$ such that $a_i$ is
in the equivalence class $u_i$ for all $i$, and the elements
$a_1,\dots,a_n$ either all belong to the domain of $\inst$ and are in
the interpretation of $P$ in $\inst$ or they all belong to the domain
of $\inst'$ and are in the interpretation of $P$ in
$\inst'$.

The following is a useful observation about guarded sets in the disjoint
sum:
\begin{lemma} \label{lem:guardedinsum}
 For every guarded relation $R$ in $\inst + \inst'$ there are
guarded relations $A$ in $\inst$ and $B$ in $\inst'$ such that
$R = A \uplus B$.
\end{lemma}

The following ``composition lemma'', also straightforward, says that $\equiv_k$ behaves well under disjoint union:
\begin{lemma} \label{lem:replacement}
 If $\inst \equiv_k \inst'$ and $\jinst \equiv_k
\jinst'$ then $(\inst + \jinst) \equiv_k (\inst' +
\jinst')$.
\end{lemma}

We now show that we can without loss of generality consider UCQ views rather than Datalog views:

\begin{lemma} \label{lem:datalog to ucq}
 Let $\view = \{\view_1,\dots,\view_n\}$ be a set of Datalog views and
$\varphi$ a GSO-sentence. Then there is a set $\view' =
\{\view'_1,\dots,\view'_n\}$ of UCQ views for the same view schema as
$\view$ such that for every
instance $\viewinst$ over the view schema
\begin{equation} \label{e:backwards equivalent}
 \inst \models \varphi \mbox{ holds for all } \inst \in
\viewback(\viewinst) \iff 
 \inst \models \varphi \mbox{ holds for all } \inst \in
\someviewback{\view'}(\viewinst).
\end{equation}
\end{lemma}
\begin{proof}
 Let $k$ be the quantifier depth of $Q$, thought of as a GSO formula.
For every view $V_i$ with answer variables $x_1,\dots,x_m$ we let
$\schema_i$ be the same as the schema $\schema$ of $\varphi$, but with added
constants $c_1,\dots,c_m$. Every approximation $\approxInst$ of the
Datalog program for $V_i$ can then be considered as an instance for the
schema $\Lang_i$ in which the constants $c_1,\dots,c_n$ denote answer
variables $x_1,\dots,x_m$ at the root of $\approxInst$. We partition the
set of all approximations of $V_j$ into its $\equiv_k$-equivalence
classes for the schema $\Lang_i$. Because the index of $\equiv_k$ is
finite we can choose finitely many representatives
$\rinst_{i,1},\dots,\rinst_{i,l_i}$ of all the $\equiv_k$-equivalence
classes of approximants of $V_i$. Recall that the \emph{canonical CQ} of an instance  is
the CQ whose variables are the elements of $\inst$, with atoms for each fact of the instance.
We define the UCQ view $\view'_i$ to
be the UCQ that is the disjunction of all the canonical CQs of the
instances $\rinst_{i,1},\dots,\rinst_{i,l_i}$, where the constants
$c_1,\dots,c_m$ become the answer variables.

The right-to-left direction of \eqref{e:backwards equivalent} holds
because every $\inst \in \someviewback{\view'}(\viewinst)$ also exists
in $\viewback(\viewinst)$. The reason is that every CQ component
of the UCQ view $\view'_i$ is also a CQ approximation  of the Datalog
view $\view_i$.

For the direction from left to right we show that for every $\inst \in
\viewback(\viewinst)$ there is an instance $\inst' \in
\someviewback{\view'}(\viewinst)$ such that $\inst \equiv_k \inst'$
holds for the schema $\schema$. 
For a fact $F=V_i(\vec c)$ over the view schema where the views are defined using  Datalog query $Q_V$, we
talk about an \emph{approximation} of the fact, meaning $Q_i(\vec c)$ for $Q_i$ a CQ approximation of $Q_V$.
The idea is to replace the  approximation
of any view fact in $\inst$ with the corresponding representative, and
then use Lemma~\ref{lem:replacement} to argue that this replacement
preserves equivalence. To make this precise we use an induction over an
enumeration $F_1,\dots,F_h$ of all the facts in $\viewinst$. We show by
induction over $j = 0,1,\dots,h$ that there is an $\inst_j$ such that 
\begin{enumerate}
 \item $\inst_j \equiv_k \inst$,
 \item $\inst_j$ contains approximations of the facts $F_1,\dots,F_j$ of
$\viewinst$ according to the views in $\view'$, and
 \item $\inst_j$ contains approximations of the facts $F_{j + 1}, \dots,
F_h$ from $\viewinst$ according to the views in $\view$.
 \item $\inst_j$ does not contain any facts that are not required by
the previous two items.
\end{enumerate}
From the second and third item it is clear that $\inst_h \in
\someviewback{\view'}(\viewinst)$ and thus we can then take $\inst' =
\inst_h$. In the base case of the induction we let $\inst_0 = \inst$.

In the inductive step we assume that we are given the instance
$\inst_j$, satisfying the three items above. We need to construct
$\inst_{j + 1}$, that satisfies the same items with $j + 1$ in place of
$j$. Thus, consider the fact $F_{j + 1}$ in $\viewinst$. It corresponds
to one of the views $\view_i$ in $\view$ and by the third item $\inst_j$
contains an approximation $\approxInst$ of the Datalog program $\view_i$
as the approximation  of $F_{j + 1}$. Consider $\inst_j$ as an instance over
the extended schema $\schema_i$, where the constants $c_1,\dots,c_m$
denote the elements at the root of $\approxInst$ that correspond to the
answer variables in $\view_i$ and to the constants in the fact $F_{j +
1}$. It is clear that we can then view $\inst_j$ as a sum $\approxInst +
\inst^-$ for the schema that contains these additional constants. Here,
$\inst^-$ contains all the facts of $\inst_j$ that are not in the
approximation  $\approxInst$ of $F_{j + 1}$. Let $\approxInst_{i,u}$ be the
representative of the $\equiv_k$ equivalence class of $\approxInst$ that
was chose above for the definition of the view $\view'_i$. We let
$\inst_{j + 1} = \approxInst_{i,u} + \inst^-$. Thus $\inst_{j + 1}$ is
just like $\inst_j$, but the fact $F_{j + 1}$ has been unfolded to
$\approxInst_{i,u}$ instead of $\approxInst$. By the choice of
$\approxInst_{i,u}$ we have that $\approxInst_{i,u} \equiv_k
\approxInst$ and thus it follows by Lemma~\ref{lem:replacement} that
\[
 \inst_{j + 1} = (\approxInst_{i,u} + \inst^-) \equiv_k  (\approxInst +
\inst^-) = \inst_j \equiv_k \inst.
\]
Note that here by Lemma~\ref{lem:replacement} the first equivalence
$\equiv_k$ holds for the schema $\schema_i$ of the view fact
$F_{j + 1}$. But since $\schema_i$ extends $\schema$ with constants, the
equivalence also holds if we consider the instances as instances for the
reduced schema $\schema$ without these constants.
\end{proof}

We now finish the proof of Theorem \ref{thm:compact} by handling the case of UCQ views:

\begin{lemma} \label{lem:easykonig}
 Assume that $\view$ consists of UCQ views.
 Let $Q$ be some Datalog query and consider a countably infinite
instance $\viewinst$ over view schema such that $\inst \models Q$ for all $\inst
\in \viewback(\viewinst)$. Then there is a finite $\viewinst^- \subseteq
\viewinst$ such that $\inst \models Q$ for all $\inst \in
\viewback(\viewinst^-)$.
\end{lemma}
\begin{proof}
 Let $F_0,F_1,\dots$ be an enumeration of all the atomic facts in
$\viewinst$. For every $n \in \omega$ define the finite subinstance
$\viewinst_n \subseteq \viewinst$ to consist of just the facts
$F_0,F_1,\dots,F_{n - 1}$. We claim that there is some $n$ such that
$\inst \models Q$ holds for all $\inst \in
\viewback(\viewinst_n)$. Thus, we can then set $\viewinst^- =
\viewinst_n$.

Define an infinite, but finitely branching, tree $T$, whose vertices are
pairs $(n,\inst)$ such that $n \in \omega$ and $\inst$ is a finite
instance over the original schema. The root is $(0,\emptyset)$, and for
every node $(n,\inst)$ that is already in the tree we add children as
follows: Consider the fact $F_n = \view_i (c_1,\dots,c_m)$ in $\inst$.
For every CQ $U$ in UCQ defining the view $\view_i$ we add a child
$(n,\inst')$ to the tree, where $\inst'$ is the extensions of $\inst$
with the canonical database of $U$, identifying the constants
$c_1,\dots,c_m$ in $\inst$ with the answer variables of $U$.

It follows from our definition of the tree $T$ that 
\[
 \viewback(\viewinst_n) = \{\inst \mid (n,\inst) \in T\}.
\]
Moreover, it also follows from the definitions that
\[
 \viewback(\viewinst) = \left\{\inst_\beta \mid
\beta \mbox{ infinite branch in } T \right\},
\]
where $\inst_\beta = \bigcup \{\inst \mid \beta_n = (n,\inst) \mbox{ for
some } n\}$.

To prove the claim assume for a contradiction that for every $n$ there
is some $\inst \in \viewback(\viewinst_n)$ such that $\inst \not \models
Q$. Consider the subtree $T'$ of $T$ consisting only of nodes
$(n,\inst)$ such that $\inst \not \models Q$. This is a connected
subtree because if $(n,\inst_p)$ is a parent of $(n+1,\inst_c)$ then
$\inst_p \subseteq \inst_c$ and thus $\inst_c \not \models Q$ implies
$\inst_p \not \models Q$. By our assumption that for every $n$ there is
some $\inst \in \viewback(\viewinst_n)$ such that $\inst \not \models Q$
it follows that $T'$ is infinite. It follows from K\"{o}nig's Lemma that
$T'$ contains an infinite branch $\beta$. Since $\inst_\beta \in
\viewback(\viewinst)$ we have that $\inst_\beta \models Q$. Because $Q$
is a Datalog query there are finitely many facts in $\inst_\beta$ that
are sufficient for $Q$ to hold in $\inst_\beta$. This entails that
$\inst \models Q$ holds already for a finite $\inst$ such that $\beta_n
= (n,\inst)$ for some $n$. This contradicts our definition of $T'$.
\end{proof}

Because every $\fgdl$  query is expressible in GSO we can use
Lemma~\ref{lem:datalog to ucq} to strengthen the result from the previous
lemma.

\begin{corollary} \label{cor:konig trick}
 Assume that $\view$ consists of Datalog views and that $Q$ is some
$\fgdl$ query. Consider a countably infinite instance
$\viewinst$ such that $\inst \models Q$ for all $\inst \in
\viewback(\viewinst)$. Then there is a finite $\viewinst^- \subseteq
\viewinst$ such that $\inst \models Q$ for all $\inst \in
\viewback(\viewinst^-)$.
\end{corollary}

\subsection{Proofs related to Theorem \ref{thm:autpre}: automata for entailment with views and rules over bounded treewidth structures} \label{subsec:app-certainanswertreewidth}

Recall the statement of Theorem
 \ref{thm:autpre} [Automata for entailing witnesses] 

\medskip

Let $\Sigma$ consist of frontier-guarded TGDs, $\views$ a set of views defined
by $\fgdl$ Datalog queries, and $Q$ a $\fgdl$ query.
Let $k$ be a number.

There is a tree automaton $T_{Q, \views, \Sigma, k}$ that runs over $k$-tree codes in the view signature
such that for every finite $\vinst_0$ of treewidth $k$,

$T_{Q, \views, \Sigma, k}$ accepts a code of $\vinst_0$ if and only if $\vinst_0 \models_{\views, \Sigma} Q$.

\medskip

Note that here we refer to an ordinary tree automaton over finite trees.

\begin{proof}
By Lemma~\ref{lem:datalog to ucq} it suffices to consider the case where
the views in $\view$ are UCQs.

Recall that GSO refers to Guarded Second Order Logic, where we have first order quantification and also quantification over guarded
sets of tuples.
We are going to show a stronger statement:

\medskip

For any $\varphi$ in GSO, we can define a MSO formula $\psi$ for the schema
of $\codes{(\viewschema)}{k}$-trees such that for every
$\codes{(\viewschema)}{k}$-tree $\tree$ we have that $\tree \models
\psi$ iff $\tree$ is a tree code and $\inst \models
\varphi$ for all $\inst \in \viewback(\decode{\tree})$. 

\medskip

To obtain the proof of Theorem  \ref{thm:autpre}, we apply the above  to the GSO sentence
 $\varphi=\Sigma \rightarrow Q$. We get an MSO  sentence $\tau$ that checks whether
all instances obtained by backward chasing the views, $\Sigma \rightarrow Q$ holds.
That is:

\begin{equation} \label{e:target formula}
 \tree \models \beta \iff \inst \models \varphi \mbox{ for all } \inst
\in \viewback(\decode{\tree}), 
\end{equation}

 It is clear that there is a MSO formula $\tau$
(or an automaton) that expresses that $\tree$ is a tree code. 
We then just convert $\psi = \tau \land \beta$ to a tree automaton, which we can do since over finite
trees, finite tree automata capture MSO properties.

We now turn to proving the stronger statement. Given an instance of the view schema, we refer to
an instance  of the base schema in $\viewback(\decode{\tree}$ as a \emph{realization}.
We call an instance of the
The main difficulty is to express the quantification over all realizations
in the right hand side of \eqref{e:target formula} with monadic
quantifiers over the treecode $\tree$. The proof strategy is to extend the
bags that are encoded at the nodes of $\tree$ with additional local
names that hold of the realizations of any of the view facts in the bag. To
this aim we define a notion of extended tree codes, which are tree codes
for the schema $\schema' = \viewschema \uplus \schema$ that are supposed
to represent the view instance $\decode{\tree}$, together with one of
its realizations $\inst \in \viewback(\decode{\tree})$. We provide more
details in the paragraphs below.

For every UCQ $\view_i$ in $\view$ let $a_i$ denote its arity. For every
UCQ $\view_i$ in $\view$ let $R_i$ be the finite set of all canonical
instances of the CQs in the UCQ $\view_i$. Each of the instances $\rinst
\in R_i$ can be seen as an instance for the schema $\schema_i$, which
extends the schema $\schema$ of $\varphi$ with additional constants
$c_1,\dots,c_{a_i}$, that mark the positions of the answer variables of
$\view_i$.

Let $b$ be  the maximal size of any of the $\rinst \in R_i$ for any
$i$ and set $k' = k + r \cdot (k + 1)^m \cdot b$, where $r$ is the
number of views in $\view$ and $m$ is the maximal arity of any of the
views in $\view$. It is clear that for every view $\view_i$ in $\view$,
any realization instance $\rinst \in R_i$,  and any tuple $\ntuple =
(l_1,\dots,l_{a_i})$, there is an injective function
$e^\rinst_{i,\ntuple}$ from the domain of $\rinst$ to $\{1,\dots,k'\}$
such that $e^\rinst_{i,\ntuple}(c_j) = l_i$, meaning that the element
corresponding to the $j^{th}$ answer variable in $\rinst$ is mapped to the
$j$-th element in the tuple $\ntuple$. The number $k'$ is chosen large
enough to ensure that the realizations of any view
facts overlap only at their answer variables. Formally, this means that
for all views $\view_i$ and $\view_{i'}$ in $\view$, $\ntuple =
(l_1,\dots,l_{a_i})$ and $\ntuple' = (l'_1,\dots,l'_{a_{i'}})$ such that
either $\view_i \neq \view_{i'}$ or $\ntuple \neq \ntuple'$, and
moreover, for all $\rinst \in R_i$, $\rinst' \in R_{i'}$, with elements
$a$ from $\rinst$ and $a'$ from $\rinst'$ such that $a$ is not answer
variable of $\rinst$ and $a'$ is not an answer variable of $\rinst'$, it
holds that $e^\rinst_{i,\ntuple}(a) \neq e^{\rinst'}_{i',\ntuple'}(a')$.

Define an \emph{extended tree code} to be a
$\codes{(\schema')}{k'}$-tree such that:
\begin{enumerate}
 \item \label{i:tree code} The tree $\tree$ is a
$\codes{(\schema')}{k'}$-tree code 
 \item \label{i:extension is local} For every $T_g$ that is true at a
vertex $v$ of $\tree$ the partial injection $g$ is such that its domain
and range are both contained in $\{1,\dots,k\}$.
 \item \label{i:no equality in extension} If $T^{=}_{l,l'}$ is true at a vertex $v$ of $\tree$ then $l,l'
\leq k$.
 \item \label{i:no view facts in extension} If $T^{\view_i}_\ntuple$ is
true at a vertex $v$ of $\tree$ for some view symbol $\view_i$ from
$\viewschema$ then $\ntuple \in \{1,\dots,k\}^{a_i}$, that is, $\ntuple$
contains only local names from $\{1,\dots,k\}$.
 \item \label{i:unfoldings are there} For every vertex $v$ of $\tree$
there is a function $\upsilon_v$ that maps any pair $(\view_i,\ntuple)$,
where $T^{\view_i}_\ntuple$ is true at a vertex $v$, to an element
$\upsilon_v(\view_i,\ntuple) \in R_i$ such that
\begin{enumerate}
 \item for every fact $R(b_1,\dots,b_h)$ in $\rinst =
\upsilon_v(\view_i,\ntuple)$ such that $T^{\view_i}_\ntuple$ is true at
$v$ the predicate $T^R_\mtuple$ is true at $v$ for $\mtuple =
(e^\rinst_{i,\ntuple}(b_1),\dots,e^\rinst_{i,\ntuple}(b_h))$, and
 \item if $T^R_\mtuple$ is true at $v$ then there is some
$T^{\view_i}_\ntuple$ that is true at $v$ and a fact $R(b_1,\dots,b_h)$
in $\rinst = \upsilon_v(\view_i,\ntuple)$ such that $\mtuple =
(e^\rinst_{i,\ntuple}(b_1),\dots,e^\rinst_{i,\ntuple}(b_h))$.
\end{enumerate}
\end{enumerate}
One can check that there is an MSO formula $\eta$ such that $\tree
\models \eta$ iff $\tree$ is an extended tree code.


Every $\codes{(\schema')}{k'}$-tree $\tree$ projects to a
$\codes{\schema}{k'}$-tree $\restr{\tree}{\schema,k'}$ if we
just forget all the predicates of the form $T^{\view_i}_\ntuple$ for
some view relation $\view_i$ in $\view$. Moreover, it is obvious that if
$\tree$ is a $\codes{(\schema')}{k'}$-tree code then its projection
$\restr{\tree}{\schema,k'}$ is a $\codes{\schema}{k'}$-tree
code.

Similarly, every $\codes{(\schema')}{k'}$-tree $\tree$ projects to a
$\codes{(\viewschema)}{k}$-tree $\restr{\tree}{\viewschema,k}$. We just
forget all the predicates that mention any local name $l > k$ or that
are of the form $T^R_\ntuple$ for some relation symbols $R$ from
$\schema$. Because of items \ref{i:tree code}, \ref{i:extension is
local} and \ref{i:no equality in extension} above we have that if
$\tree$ is an extended tree code then $\restr{\tree}{\viewschema,k}$ is
a $\codes{\viewschema}{k}$-tree code.

Extended tree codes capture the backwards mapping of views. This is
made precise by the following two properties of extended tree codes:
\begin{description}
 \item[Soundness:] For every extended tree code $\tree$ there is an $\inst \in
\viewback(\decode{\restr{\tree}{\viewschema,k}})$ such that $\inst \cong
\decode{\restr{\tree}{\schema,k'}}$.
 \item[Completeness:] For every $\codes{(\viewschema)}{k}$-tree code $\tree'$ and
$\inst \in \viewback(\decode{\tree'})$ there is an extended tree code
$\tree$ such that $\restr{\tree}{\viewschema,k} = \tree'$ and
$\decode{\restr{\tree}{\schema,k'}} \cong \inst$.
\end{description}
Using items \ref{i:no view facts in extension} and \ref{i:unfoldings are
there} from the definition of extended tree codes, one can check that
these properties hold.

We  observe that in MSO we can evaluate GSO formulas over the
models that are encoded over tree codes.  Thus, there is an MSO formula
$\varphi'$ such that for all $\codes{\schema}{k'}$-tree codes $\tree$
\[
 \tree \models \varphi' \iff \decode{\tree} \models \varphi .
\]
See the argument for the ``Courcelle-like results'' - Proposition \ref{prop:kpipeline} and Theorem \ref{thm:courcelle} -- elsewhere in the appendix.

 Also note that the only difference between
$\codes{\schema}{k'}$ and $\codes{(\schema')}{k'}$ is that the latter
contains additional predicate symbols. Hence, the formula $\varphi'$
can also be evaluated over extended tree codes $\tree$ for which we have
that
\[
 \tree \models \varphi' \iff \decode{\restr{\tree}{\schema,k'}} \models
\varphi .
\]

Define the MSO formula $\beta$ as follows:
\[
 \beta \equiv \forall Q_1\dots\forall Q_s (\eta \rightarrow \varphi'),
\]
where $Q_1,\dots,Q_s$ are all the predicate symbols that are in
$\codes{\schema}{k'}$. The idea on a $\codes{(\viewschema)}{k}$ tree is to
consider all extensions to a $\codes{(\schema')}{k'}$-tree, which need to
extended tree codes because of $\eta$, and then make sure that they all
satisfy $\varphi$. By the soundness and completeness properties of extended
tree codes it follows that this definition of $\beta$ satisfies
\eqref{e:target formula}.
\end{proof}

\section{Additional material, related to applying the tools from Section \ref{sec:tools}
to obtain  decidability and rewritability results, Section \ref{sec:decide}}
\subsection{More details on the proof of Theorems \ref{thm:datalogqfgdlviewsdecide} and \ref{thm:mdlqcqviewslindecide}: decidability of monotonic determinacy via the BTVIP}

Recall the statement of Theorem  \ref{thm:datalogqfgdlviewsdecide}:

\medskip

 Let $Q$ range over $\fgdl$ queries, $\Sigma$ over $\fgtgd$s, and $\views$ over FGDL views.
Then monotonic determinacy is decidable.


\medskip

\begin{proof}
We give more details here. Note that similar arguments for the final applications of this theorem
-- to CQ views and FGDL views -- are given in Theorems 9.1 and 9.2
of \cite{uspods20journal}. The main difference here relates to the fact that infinite structures are involved.

The heart of the argument is to show that we have the BTVIP. The fact that the approximations are treelike is Proposition \ref{prop:boundedtreewidthapprox}.
The fact that  the chase with $\fgtgd$s gives a low treewidth structure is well-known: see, for example Proposition
\ref{prop:chaseapprox} in Section \ref{app:chaseapproxregular},
or
\cite{treelikechase}. As mentioned earlier,  the idea goes back to 
\cite{datalogpmj}. Applying the FGDL views simply adds facts compatible with the same tree decomposition, since the views are guarded.

Proposition \ref{prop:ksuffices}  tells us that we can compute a $k'$
such that  it suffices to look for a $k'$ treelike counterexample -- possibly infinite -- to monotonic
determinacy. 

Proposition \ref{prop:kpipeline} tells us that we can compute a B\"uchi Automaton (or a WMSO sentence) that represents these counterexamples.
So it suffices to check whether this sentence is satisfiable, which we can do by decidability of nonemptyiness of these automata, or appealing
to decidability of satisfiability of Weak MSO over trees, which follow from
Rabin's theorem: see, for example, Theorem 11.3 of  \cite{thomas}.
\end{proof}

The argument for Theorem \ref{thm:mdlqcqviewslindecide} is similar, except the confirmation of the BTVIP requires the analysis of CQ views
as in Lemma 6.5 and Theorem 8.2 of \cite{uspods20journal}.

\subsection{Proof of Theorem \ref{thm:linearcqcq}: decidability of monotone determinacy for CQ views and
queries, for  linear TGDs}

Recall the statement of Theorem \ref{thm:linearcqcq}:

\medskip

Suppose $\Sigma$ ranges over linear TGDs,
Then monotonic determinacy of UCQ  $Q$ over CQ views $\views$ relative
to $\Sigma$ is decidable.

\medskip

\begin{proof}
We make use of the certain answer rewriting approach outlined in the body of the paper. That is, we will
use certain answer rewriting  to eliminate steps moving backward in the process of Figure  \ref{alg:query-mondet-constraints}.

We can produce a UCQ $R_1$ which is a certain answer rewriting of $Q$ with respect to $\Sigma$.

Consider a variation of the process of Figure \ref{alg:query-mondet-constraints}, where the chase  steps in
line $4$ produce facts in a primed signature, and the further steps after that are also in the primed signature.
This change of signature does not impact the correctness of the procedure.
Now consider the steps where we apply the views and then ``chase backwards'' with the view-to-definition rules.
We can consolidate these  into
a single step that chases with rules of the form:
\[
\phi_v(\vec x, \vec y) \rightarrow \exists \vec y ~ \phi'_v(\vec, \vec y)
\]
where $\phi_q$ is the view definition of a view $v(\vec x)$.
These are source-to-target rules, which are always first-order rewritable. That is,
we can get a UCQ $R_2$ that is certain answer rewriting of $R_1$ for these rules.

Now monotonic determinacy, the success of the whole process,  is equivalent to 
the entailment
\[
Q \models_\Sigma R_2
\]
This is an entailment of UCQs with linear TGDs, decidable in PSPACE.
\end{proof}

\subsection{More details on decidability and monotonic determinacy in the finite}
In Section \ref{sec:decide} in the body we mentioned that in  the three decidable cases
for monotone determinacy, we have decidability in the finite.

The \emph{second decidable case} was about Fr-1 TGDs, MDL queries and CQ views.
For UCQ queries and CQ views, finite controllability was proven in
Theorem \ref{thm:finconfone} in the paper. 
We can use the same ideas to handle the MDL case:  In order to extend to MDL queries, we need to use the trick of ``moving MDL rules into
Fr-1 TGDs'' that is applied in  the proof of Theorem \ref{thm:rewritefrontieronemdl}.

The \emph{third decidable case} had to do with \emph{CQ views and queries, linear TGDs}. We state the result here:

\begin{proposition} \label{prop:fclinear} Suppose $\Sigma$ consists of Linear TGDs.
Then
monotonic determinacy of a CQ query $Q$ over CQ views $\views$ with respect to $\Sigma$ is finitely controllable.
\end{proposition}

We give the proof of the proposition. It also relies on finite controllability results for certain answers, but in this case we use an old result.
 Recall that a  \emph{$k$-universal model} for an instance $\inst$ and set of dependencies
$\Sigma$ is a superinstance $\inst_k$
of $\inst$ that satisfies $\Sigma$
with the property that a CQ of size at most $k$ holds in $\inst_k$ if and only
if it is entailed by $\inst$ and $\Sigma$.

For any TGDs, The chase is a $k$-universal model for any $k$, but in general it is infinite.
However for linear TGDs it is known that we can do better:

\begin{theorem} \cite{rosati}
Linear TGDs have finite $k$-universal models for each instance $\inst$ and
each $k$.
\end{theorem}

We can easily use this to prove Proposition \ref{prop:fclinear}.

\begin{proof}
Compute a UCQ $Q^*$  that is a rewriting of $Q'$ (the primed version
of $Q$) with respect to $\Sigma'$.
Such a $Q^*$ exists by the BDDP.
Let $q$ be the size of $Q$,
 $k$ be the maximal number of atoms in $Q^*$, and let $w$ be the maximal number
of atoms in a view.

Choose a $k \cdot w$-universal model $\inst_1$ for the canonical database of $Q$
and $\Sigma$.
Let $\inst_2$ be the result of applying $\Sigma_{\views}$ to $\inst_1$, where
$\Sigma_{\views}$ are the source-to-target TGDs corresponding to $\views$, going
from the unprimed relations to the primed ones.
Let $\inst_3$ be a $q$-universal model for $\inst_2$ and $\Sigma$.

We claim that if the pair $\inst_1, \inst_3$ is
not a finite counterexample to monotonic determinacy of $Q$ over $\views$, then $Q$ is monotonically-determined
over $\views$. This would imply monotonic determinacy is finitely controllable.

Clearly $\inst_1, \inst_3$ form a counterexample to
monotonic determinacy
as long as $\inst_3$ does not satisfy $Q'$, so suppose $\inst_3$ satisfies $Q'$.
We know by $q$-universality of $\inst_3$ over $\Sigma'$ and $\inst_2$ that
$Q'$ must be entailed by $\Sigma'$ over $\inst_2$.
Hence the rewriting $Q^*$ must hold in $\inst_2$.
Thus $Q^*$ must be entailed by $\Sigma_\views$ over $\inst_1$, by universality of the usual
chase. Letting $Q''$ be a rewriting of $Q^*$ with respect to $\Sigma_\views$,
we know $Q''$ holds in $\inst_1$.
The maximal size of disjuncts in $Q''$ is at most $k \cdot w$, and thus
by $k \cdot w$-universality of $\inst_1$,
$Q''$ is entailed by $\Sigma$ over the canonical database of $Q$.

We argue  that $Q'$ is entailed by $\Sigma, \Sigma_{\views}, \Sigma'$
over canonical database of $Q$, which  would imply monotonic determinacy.
This follows using the properties of rewritings and the chase.
Since $Q''$ is entailed by $\Sigma$ it will hold in $\inst^\infty_1$, the chase
of the canonical database under $\Sigma$. Since $Q''$ is a rewriting of $Q^*$,
$Q^*$ will hold in $\inst^\infty_2$, the chase of $\inst^\infty_1$ by $\Sigma_{\views}$.
And since $Q^*$ is a rewriting of $Q'$, we infer $Q'$ must hold in the chase of $\inst^\infty_2$
by $\Sigma'$. Thus the ``pipeline'' described in
 Figure \ref{alg:query-mondet-constraints} returns true, and using Proposition \ref{prop:pipeline} we conclude
that $Q$ is monotonically determined over $\views$ relative to the rules.
\end{proof}

Thie leaves the \emph{first decidable case}. Recall that this case is \emph{``everything is guarded''}:  the views and queries are both guarded, and rules are $\fgtgd$s.
Here  we do not claim finite  controllability of monotonic determinacy. Instead we argue for decidability
of monotone determinacy in the finite directly.

\begin{theorem} \label{thm:finconcqcq}
Let $Q$ be a  UCQ query, $\views$ a set of CQ views, and $\Sigma$  a set of linear TGDs
If $Q$ is monotonically determined by $\views$ with respect to $\Sigma$ over finite structures, then the same holds over all structures.
\end{theorem}

We provide the proof of this theorem.
The \emph{Guarded Negation Fixpoint Logic} (GNFP) is a fragment of LFP with the restrictions:
\begin{itemize}
\item  Only existential quantification
\item Free variables in a negated formula are guarded. negation is of the form:
 $R(\vec x) \wedge \neg \phi(\vec x)$, where $R$ is an input relation, not a free second order variable
that gets bound by a fixpoint.
\item Fixpoints are guarded. To take a least fixpoint over predicate $U$ of a formula $\phi(U, \vec x)$,
$\vec x$ must be guarded by an input relation.
\end{itemize}
If we disallow fixpoints, we get the \emph{Guarded Negation Fragment} GNF.

Monotone determinacy of $Q$ over $\views$ with respect to $\Sigma$ can be rephrased as the satisfiability of:
$Q \wedge \Sigma \wedge  \Sigma_{\views} \wedge \Sigma' \wedge \neg Q'$
where $Q'$ is a copy of the query on a signature where each predicate $R$ is replaced by a primed copy $R'$, and similarly
for the rules $\Sigma$. $\Sigma_{\views}$ contains axioms:
\begin{align*}
 \forall \vec x~ \phi_V(\vec x) \rightarrow \phi'_V(\vec x)
\end{align*}
where $\phi_V$ is the definition of view $V$, and $\phi'_V$ is a copy in a primed signature.

$Q$ and $\neg Q'$  are in GNF, hence in GNFP.  In the latter case, this is  because there are no free free variables, recalling
that $Q$ is a Boolean CQ.
$\Sigma$ and $\Sigma'$ are  likewise in GNF. $Q$ is clearly in GNFP, since the rules are assumed guarded.
Similarly $\Sigma_{\views}$ can be expressed in GNFP, since the views are guarded by a base relation.

Now we can use a result from \cite{gnfj} that  entailment for GNFP is decidable in the finite.

\subsection{Proof of \cref{thm:finconfone}}
Recall the statement of \cref{thm:finconfone}:

\medskip

Let $Q$ be a CQ  query, $\views$ a set of  CQ views, and $\Sigma$ is a set of $\frone$ TGDs.
If $Q$ is monotonically determined by $\views$ with respect to $\Sigma$ over finite structures, then the same holds over all structures.

\medskip

\begin{proof}
Fix $Q$, $\views$, and $\Sigma$. Let $\sigma$ be the signature of $Q$ and $\Sigma$. Let $\sigma'$ denote primed the copy of $\sigma$. Below we shall use $\cdot'$ to denote copies of queries, rules, and instances over $\sigma$ to their respective versions in the primed signature $\sigma'$.

We show that the general and finite cases of monotonic determinacy coincide. To this end assume that a there exists a counterexample to monotonic determinacy. We shall show that there exists a finite counterexample. Note that the other direction is trivial.

As stated above, we let $Q'$ be the query $Q$ copied over to to $\sigma'$, the copy of the base signature, and similarly let $\Sigma'$ denote the copy of the rules on a primed signature.
 Using our results on certain-answer rewriting,  Theorem \ref{thm:rewritefrontieronemdl} ,
we can get an MDL certain answer rewriting $R'$ of $Q'$ with respect to $\Sigma'$. Note that $R'$ works over the copied
signature $\sigma'$. Note that since the views are CQs, we can treat view definitions and their ``inverses'' 
(view-to-definition rules) as TGDs.
One can rewrite each rule body of $R'$ against these rules, ignoring the intensional predicates. Let $\Theta$ denote the set of such TGDs, 
We let $R$ denote the resulting Datalog program, with the notation indicating that it is in the unprimed signature $\sigma$.  Note that in the process of rewriting $R'$ we did not change the heads of $R'$, meaning that $R$ is also MDL. 

We shall state the properties of $R$ and $R'$ that follow from their respective definitions. We have that: 
\begin{itemize}
\item $\spadesuit$) for any instance $\inst$ over $\sigma$ if $\inst \models R$ then $\inst \wedge \Theta \models R'$; 
\item $\diamondsuit$) for any instance $\inst'$ over $\sigma'$ if $\inst' \models R'$ then $\inst' \wedge \Sigma' \models Q$. 
\end{itemize}
Consequently we have  $\heartsuit$) $\inst \models R$ implies $\inst \wedge \Theta \wedge \Sigma \models Q$. 

As there exists a counterexample to monotonic determinacy, from \cref{prop:pipeline} we know that there exists one produced by the process of \cref{alg:query-mondet-constraints}. Note that when we apply this process in the case of CQ query and views, it is deterministic. 
Let $C_n$ 
be as in \cref{alg:query-mondet-constraints}. Note, that by the definition of $R$ and $C_n$ we have $C_n \models Q \wedge \Sigma \wedge \neg R$, otherwise, from ($\heartsuit$) the process of \cref{alg:query-mondet-constraints} would return true. From \cref{thm:fcecdatalog}, we obtain a finite model of $Q$ and $\Sigma$ that does not entail $R$. Let $C_f$ denote it. Note that $\backview_{\views}(\views(C_f))$ is a singleton, again since the views $\views$ are CQs. Take the instance and copy it  to the
$\sigma'$ signature, letting $M'_f$ denote the result.
Thus $M'_f$ is finite and is isomorphic to $\chase_\Theta(C_f)$. Therefore, from $(\spadesuit)$, we have that $M'_f$ does not entail $R'$. From $(\diamondsuit)$, we have that $M'_f$ and $\Sigma'$ do not entail $Q$. From a second application of \cref{thm:fcecdatalog} we obtain a finite model of $M'_f$ and $\Sigma'$ that does not entail $Q$. 
Let $C'_f$ denote that model. Finally observe that $\views(C_f) \subseteq \views(C'_f)$, $C_f \models Q$, and $C'_f \not \models Q'$. Therefore, as both $C_f$ and $C'_f$ are finite $C_f \cup C'_f \cup \views(C_f)$ form a finite counterexample to monotonic determinacy.
\end{proof}

\subsection{More details on Theorem \ref{thm:fgviewsplusdatalogqueryrewriting}: rewritability via the forward-backward method}

Recall Theorem \ref{thm:fgviewsplusdatalogqueryrewriting}:

\medskip

Let $\Sigma$ be a set of $\fgtgd$s, $Q$ a Boolean Datalog query and $\views$ a set of $\fgdatalog$ views.
Then if $Q$  is monotonically
determined by $\views$ $\wrt$ $\Sigma$ there is a
$\poslfp$ rewriting of $Q$ in terms of  $\views$ w.r.t $\Sigma$.

\medskip

\begin{proof}
For this proof we  provide a  sketch of the argument.

We proceed by first using Proposition \ref{prop:boundedtreewidthapprox}
to get a tree automaton $A_0$ that represents approximations of $Q$. 
We proceed by first using Proposition \ref{prop:boundedtreewidthapprox}
to get a tree automaton $A_0$ that represents approximations. We can also use definability in WMSO to argue that
the set of instances that contain the approximations, satisfy the rules $\Sigma$, and have their view instances correct are recognized
by a non-deterministic B\"uchi automata.

\begin{example} \label{ex:forward}
Consider a base schema with binary predicate $R, S$ and unary predicates $W$ and $Z$. Consider a Datalog query $Q$ that asserts that there is an $R$
path  from a $W$ node to a $U$ node:

\begin{align*}
\goal &\datalogarrow W(x), I(x) \\
I(x) &\datalogarrow Z(x) \\
I(x) &\datalogarrow R(x,y), I(y)
\end{align*}
The canonical expansions are just words that begin with W and traverse $R$ edges until reaching $Z$.

Consider $\Sigma$ consisting of the rules:

\begin{align*}
Z(x) \rightarrow \exists y ~ S(x,y) \\
S(x,y) \rightarrow \exists z ~ S(y,z)
\end{align*}

Thus the chase with $\Sigma$ will append to the final $Z$ node an infinite chain of $S$ edges. Note that the adjusted treewidth is $2$.

Consider $\views$ that copy the $S$ and $R$ edges, omitting the unary predicates. For brevity we call the view predicates $R$ and $S$.

The  B\"uchi tree automaton that we want will accept exactly infinite chains consisting of an initial chain of $R$ edges
and then an infinite chain of $S$ edges.  It will have
states $q_R, q_S$, with $q_R$ initial and $q_S$ accepting (i.e. the B\"uchi acceptance condition is $F=\{q_S\}$).  
The label alphabet will be sets of facts over two elements.

The transitions will be:

\begin{itemize}
\item $t_1$: if we have a parent with label including $R(p,q)$ and the parent is in state $q_R$, then in the unique child  we can transition
to $q_R$
\item $t_2$: if we have a parent with label including $S(p,q)$ and the parent is in state $q_R$, then in the unique child  we can transition
to $q_S$
\item  $t_3$: if we have a parent with label including $S(p,q)$ and the parent is in state $q_S$, then in the unique child  we can transition
to $q_S$
\end{itemize}
\end{example}

The transition function of the automaton is transformed into Datalog rules similar to
the backward mapping rules defined for finite
tree automata in the algorithm of Theorem \ref{thm:auttodlog} mentioned in the body. Recall that this refers back to
Section 7 of \cite{uspods20journal}, with the algorithm given in the beginning of that section and the correctness proven in Proposition 7.1 \cite{uspods20journal}. We now review the construction, which is quite intuitive, and then describe how we modify it.
The algorithm translates each transition in the automaton into a Datalog rule. Ignoring some additional parameters needed
for bookkeeping, in this translation  each state $q$ of the automaton over tree codes
with $k$ local names will translate into an intensional predicate $P_q(x_1 \ldots x_k)$ on the instance that is coded. Each transition between
parent state  $q$ and child states $q_1, q_2$ in the automaton will correspond to a rule  with $P_q$ in the head and $P_{q_1}, P_{q_2}$ in the rule body.

The distinction from  Theorem \ref{thm:auttodlog}, which dealt with standard tree automaton,
is that we have to deal with the acceptance conditions of the B\"uchi automaton. So now we create not a Datalog query, but a $\poslfp$ sentence.
We change the finitely many Datalog rule bodies associated to any intensional predicate $I_q$ into a disjunction of CQs.
So now we have a vector consisting of intensional predicates and associated UCQs.

\begin{example} Let us continue the examination of Example \ref{ex:forward}.
The backward mapping for the example will have binary intensional predicates for each state, $I_{q_R}(p,c)$, $I_{q_S}(p,c)$, and
extensional predicates for the binary predicates $R$ and $S$, recalling that these represent view predicates.

If we were doing the backward mapping for finite tree automata, we  would produce a Datalog program with rules
\begin{align*}
I_{q_R}(p,c) &\datalogarrow R(p,c) \wedge  I_{q_R}(c,c') \\
I_{q_R}(p,c)  &\datalogarrow  S(p,c)   \wedge I_{q_S}(c,c') 
\end{align*}
corresponding to transitions $t_1, t_2$ above, and similarly for $I_{q_S}$ and $t_3$.

In our construction, we will proceed as above for the predicates corresponding
to non-accepting states  -- $I_{q_R}(x,y)$ above. But we will proceed differently for the predicates corresponding to
accepting states, like  $I_{q_S}(x,y)$ above.
\end{example}

We  illustrate the construction of the $\poslfp$ formula for the case when  $A$ is \emph{shallow},  in the sense that when a parent
is in an accepting state, it can only transition to children that are also in accepting states.
The example above has this property:   the intuition for this example is that the accepting states 
correspond to tree code vertices generated by view images of the chase, a suffix, while non-accepting states will correspond
to treecode vertices that come from unfolding the approximation, a prefix.

Instead of creating Datalog rules as in the usual backward mapping, we  will create corresponding fixpoint formulas, where the non-accepting states are 
treated using least fixpoints,
as in traditional Datalog, while the accepting states are treating using  greatest fixpoints. 
We can consider these as vectorized fixpoints, binding a collection of predicates
at once. It is known by the ``Beki\'c principle'' \cite{bekic} that such vectorized fixpoints can be replaced by sequences of usual fixpoints.
The predicates corresponding
to accepting   states (informally, those corresponding to nodes generated in the chase, or witnesses to views), are put into an inner vector of  $\nu$ operators.  
Predicates corresponding to non-accepting states -- informally, those corresponding  to the CQ approximation -- 
are bound with a $\mu$ operator. This ordering here  -- accepting states within the scope of  non-accepting state --
corresponds in our case to the fact that we first choose  the approximation, then chase it and take the view image.

We explain the construction by continuing the example, leaving some details
 to the full version.

\begin{example} \label{ex:backwardposlpf}
Let us continue the examination of Example \ref{ex:forward}.
The backward mapping for the example will have binary intensional predicates for each state, $I_{q_R}(p,c)$, $I_{q_S}(p,c)$, and
extensional predicates for the binary predicates $R$ and $S$, recalling that these represent view predicates.

We produce the system of fixpoints: 
\begin{align*}
 \mu_{I_{q_R}}: ~  I_{q_R}(p,c) := ((\exists c'  ~ R(p,c) \wedge  I_{q_R}(c,c')) \vee \exists c' ~ ( R(p,c)   \wedge I_{q_S}(c,c') )) \\
\nu_{I_{q_S}}: ~ I_{q_S}(p,c) :=  (S(p,c) \wedge \exists c'  ~ I_{q_S}(c,c'))
\end{align*}

This asserts that we do an outer least fixpoint for $I_{q_R}$ and in each iteration do a greatest fixpoint on
$I_{q,S}$.
And then we add the quantification $\exists p c ~ I_{q_R}(p,c)$.

That is, our $\poslfp$ formula states that:
\begin{itemize}
\item there is an $ I_{q_R}$ pair
\item $I_{q_R}$ pairs are labelled correctly and have $R$ paths to an  $I_{q_S}$ pair.
\item  every $I_{q_s}$ pair is $S$-labelled  correctly  and links to another  $I_{q_s}$ pair
\end{itemize}

Thus the formula asserts exactly the structure we want within the view instance.
\end{example}

Now we explain how this example generalizes.
Recall that for $\fgtgd$s, the instance can be given a tree structure, where each tree vertex is associated with a collection of facts, as in a tree decomposition.  When a  $\fgtgd$ $\tau$ is fired, the trigger image of the body
is contained in some vertex $v$ of the tree. For a non-full TGD, we create a new child of $v$ which includes the facts in the head
along with any facts from $v$ that are guarded by the facts in the head. For a full TGD we create new facts in the same node, and
we may propagate a fact back from a child to its parent, if the fact is guarded in the parent.
We say that an $\fgtgd$ $\tau$ is \emph{speedily witnessed} in a tree code if whenever the body of the TGD matches via a homomorphism $h$ into a tree vertex $v$,
the head facts are found in either in $v$ or in the child created when the the chase step for $h$ is fired.
We say a tree code $t'$ extends another tree codes $t$ if there is an injective mapping $f$ from the nodes of $t$ to the nodes of $t'$ such that the labels of node $v$ are a subset of the labels of the nodes of $f(v)$.

\begin{claim} \label{clm:speed} For $\fgtgd$s $\Sigma$ with $k$ the maximum number of elements in a frontier, and any finite $k$ tree code  $t$ with width the maximum number of elements in a rule , there is a $k$ tree code $t'$ that extends $t$
where each rule in $\Sigma$ is speedily witnessed.
\end{claim}
\begin{proof}
This can be proven using standard results about the treelike chase for $\fgtgd$s, see for example \cite{treelikechase}: when we need a new speedy witness we create a child. In order to ensure that witness triggers for $\Sigma$ always have frontier
within a single node, we always propagate a new fact from child to parent if it is guarded in the parent. But since all such facts are guarded in the parent, this will not violate the bound $k$.
\end{proof}


The \emph{hybrid envelope trees} for a Datalog query
$Q$, views $\views$,  and $\fgtgd$s $\Sigma$ are the tree codes $t'$ for the combined base and view signature, such that:
\begin{itemize}
\item $t'$   extends an approximation tree corresponding to the canonical database of some approximation $Q_n$
\item every TGD  is speedily witnessed in $t'$
\item for each $\vec d$ coded in a single  node of $t'$, for each view definition $\phi(\vec x)$,
the corresponding  view fact $V_\phi(\vec d)$ is present in the tree node
\end{itemize}

Using the proof of Claim \ref{clm:speed} above, we can easily show:

\medskip

\begin{claim} \label{clm:universal}
For FGDL views $\views$ and a Datalog query $Q$, and any
number $n$, there is a hybrid envelope tree $t$ whose decoding is a universal model for $\canondb(Q_n) \wedge \Sigma$.
\end{claim}
\begin{proof}
We perform the treelike chase to obtain a tree satisfying the first two items: that is, we repeatedly extend as in
Claim \ref{clm:speed}. And then we simply evaluate each view
definition to add the corresponding facts. Since the  views are FGDL we do not break the tree decompositon when
we do this.
\end{proof}


A B\"uchi automaton is \emph{shallow} if for every accept state $q$ and every transition whose parent state is
$q$, the associate child states are also accepting.  That is, in an accepting run, accepting states form a suffix of the tree.

\begin{lemma} \label{lem:shallow} For Datalog $Q$, FGDL $\views$, and $\fgtgd$s $\Sigma$, there is a shallow automaton  that accepts exactly the hybrid envelop trees for $Q$, $\views$, $\Sigma$.
\end{lemma}
\begin{proof}
The first item in the definition can be enforced with a shallow
 B\"uchi automaton easily: we have non-accept states for the intermediate intentionsal predicates and then accept states
when we reach extensional facts.
The second item  can be enforced with an automaton having  a single accept state that is initial and a sink non-accepting state
The third item  can be enforced with a shallow automaton that guesses the intermediate intensional relations holding
at each node.

The automaton we want then is the product, and the product of shallow automaton is shallow.
\end{proof}

We  note that the projection of a shallow automaton is shallow. Let $A_V$ be the automaton obtained
from Lemma \ref{lem:shallow} by projecting away everything but the view facts.

We  let $\phi_{\views,\Sigma, Q}$ be  the $\poslfp$ formula  obtained by performing the backward mapping
of the shallow automaton $A_V$ above.
Our main claim is the following:

\begin{theorem}$\phi_{\views,\Sigma, Q}$ is a rewriting of $Q$ over $\views$ with respect to $\Sigma$.
\end{theorem}

\begin{proof}
In one direction suppose $Q$ holds in $\inst$ and let $\vinst$ be the view image of $\inst$.
There is some approximaton $Q_n$ with a homomorphism $h$ 
sending $Q_n$ into $\inst$. We apply Claim \ref{clm:universal}  to get
a hybrid envelope tree $t$ embedding $Q_n$ that is a universal model for $\canondb(Q_n) \wedge \Sigma$. We let $t'$ be the corresponding projected tree onto view facts. Then $h$
extends to a homomorphism $h_t$ from  the decoding of $t$ into $\inst$. By definition of the automaton $A_V$, $t'$ is accepted
by $A_V$. Since $A_V$ is shallow, there is a run $r$ of $A_V$ on $t'$ in which the domain of
accepting states form a suffix. Let $p'$ be the corresponding prefix.
We first give  a strategy for unfolding the outer least fixpoint of
$\phi_{\views,\Sigma, Q}$: we unfold until we reach the boundary of $h_t(p')$. At that point, we can unfold the inner fixpoints infinitely
often.  This witnesses that $\vinst$ satisfies $\phi_{\views,\Sigma, Q}$.

In the other direction, suppose $\phi_{\views,\Sigma, Q}$ holds in $\jinst$.

There is an unfolding of the outer fixpoint that witnesses
this, and a finite  tree  $p'$ where nodes are labelled with
the inner fixpoint, such that there is a homomorphism from $\decode{p'}$ into $\jinst$. 
 Repeatedly expanding the inner fixpoint at each node of $p'$, we will
obtain a tree with $p'$ as a prefix, thus yielding an infinite tree $t'$ accepted by $A_V$ (filled out with self loops or dummy nodes), along with a 
homomorphism $h'$ of $\decode{t'}$ into $\jinst$. 
Since $A_V$ is a projection of the automaton obtained from Lemma~\ref{lem:shallow}, there is an expansion of $t'$
to a tree $t_0$ with codes for base facts that embeds the canonical tree for
$Q_n$ into $\inst$, where the decoding of $t_0$  satisfies $\Sigma$, and where the view facts 
are contained in those given by the view definitions for  $\decode{t'}$. Let $\inst_0$ be the decoding of $t_0$. 


$\inst_0$ satisfies $Q$ since it contains a homomorphic image  of $Q_n$.
Using monotonic determinacy, we can see that $Q$ holds in $\inst$ as well.
\end{proof}

\end{proof}

\subsection{Proof of Theorem \ref{thm:noncompact}: Non-rewritability in Datalog} \label{app:norewrite}


Recall that if we have frontier-guarded TGDs, an FGDL query $Q$,  and Datalog  views $\views$, with $Q$ monotonically determined over
$\views$ with respect to $\Sigma$, then there is a rewriting of $Q$ over $\views$ in Datalog. This
is Theorem \ref{thm:fgviewsplusfgdlquery}. On the other hand, if $Q$ is in Datalog, our result only give a rewriting
in $\poslfp$. As mentioned in the body, n the case of frontier-guarded TGDs, or even linear TGDs,
we do not know if $\poslfp$ is needed.

We now show that there are TGDs with the BTVIP, where we cannot obtain rewritability in Datalog, we really need $\poslfp$.
This stated in Theorem \ref{thm:noncompact} in the body:

\medskip

 There are rules $\Sigma$, a Datalog query $Q$ such that the BTVIP holds, and $\fgdatalog$ views
with $Q$ monotonically determined by $\views$ w.r.t. $\Sigma$, where there is no Datalog rewriting.

\medskip 

Here we emphasize that, as in our default for the paper, we deal with all instances: so the example will
be monotonically determined over all instances, and \emph{there is no Datalog rewriting over all instances}.
For the example we give, there will be a Datalog rewriting over finite instances.
This is a case where the argument is specific to the setting of unrestricted instances. 

Recall,  the \emph{compactness of entailment} property discussed in the context of our positive results
on rewriting.
It deals with the situation where we have an infinite instance of the view schema, and it  is $Q$-entailing with respect  to 
$\Sigma, \views$ -- that is, it implies, using the view definitions and $\Sigma$, that $Q$ holds.
Then  there is a finite subinstance if the view image that is $Q$-entailing with respect to $\Sigma,\views$  well.
It is easy to see that this is a necessary condition to have a Datalog rewriting.
Theorem \ref{thm:compact} proved compactness of entailment in the setting where views and queries are frontier-guarded and the rules
are tame.
We give an example, where the rules still have the BTVIP, although they are not $\fgtgd$s, where  this property is  violated.
And using this failure, we will easily conclude that Datalog-rewritability fails.

\begin{proposition} 
 There is a set of TGDs $\Sigma$, a Boolean CQ
$Q$ a set of $\fgdl$ views $\views$ with the BTVIP such that $Q$ is
monotonically determined by $\views$ relative to $\Sigma$ but $Q$ has no
Datalog rewriting over the views in $\views$.
\end{proposition}
In fact, the argument will show that \emph{there is no rewriting as an infinite disjunction of CQs}.
\begin{proof}
Let the schema $\schema$ contain the nullary symbol $\goal$, the unary
relation symbol $U$ and the binary relation symbols $A$, $L$
and $R$.
Define the following set of rules for $\schema$
\begin{align*}
 \Sigma = \{ 
 & \goal \land A(x,x) \rightarrow \exists y ~ R(x,y), \\
 & \goal \land R(x,y) \rightarrow \exists z ~ R(y,z), \\
 & R(r,r') \land L(l,l') \land A(r,l) \rightarrow A(r',l') \\
 & A(r,l) \land U(l) \land L(x,y) \rightarrow \goal \land A(l,l)
\}.
\end{align*}

Let the view schema contain binary relations $V_A$ and $v_R$
along with a unary predicate $\reach_U$. The views $\views$ are defined
such that $V_A$ and $V_R$ have copy views, meaning the trivial view
definitions $V_A(x,y) \rightarrow A(x,y)$ and $V_R(x,y) \rightarrow R(x,y)$.
For the unary predicate $\reach_U$ we have the following Datalog view, which
returns the nodes that can reach a  $U$ node via an $L$-path:
\begin{align*}
 \reach_U(x) & \datalogarrow U(x) \\
 \reach_U(x) & \datalogarrow L(x,y) \land \reach_U(y)
\end{align*}

Define the Boolean CQ  $Q = \exists x ~ \goal \land A(x,x) \land U(x)$.

\begin{claim} $Q$ is monotonically determined by the views in $\views$ relative to rules $\Sigma$.
\end{claim}
\begin{proof}
Consider the process for checking monotonic determinacy with rules. We start with the canonical
database of $Q$, which has the goal predicate, along with a single element $x_0$ satisfying $U$, with a self-loop
$A(x_0, x_0$).
Then we consider the result of chasing  this instance with the rules. Chasing
 with the first two rules we create an infinite $R$ chain rooted at $x_0$, say
$x_0, x_1 \ldots$.  The two last rules do not fire,
since there are no $L$ atoms present.

Now applying the views, the copy of $R$ reveals the $R$ path within the view image. The copy of $A$ contains only
 the self-loop on $x_0$. Finally, we know that $\reach_U(x_0)$ holds. As $x_0$ can reach
a $U$ node  
via a path of $L$ edges. That is, $\reach_U(x_0)$ is generated from a degenerate, empty $R$-path.

Note that even though the rules are not frontier-guarded the BTVIP holds. Indeed,  on the approximations
of $Q$, the non-guarded rules never fire.

Thus in the ``reverse views'' step, we have instances $\inst_n$ each containing the $A$ self-loop on $x_0$, the
$R$-facts, and a
path $x_0=v_0 \ldots v_n$ of size $n$ connected by  $L$ edges, leading from $x_0$ to a $U$ node $v_n$.

Consider now chasing the instance $\inst_n$ with the rules.
The first two rules do not fire initially, since $\goal$ does not hold.
But now the second-to-last rule does fire, propagating  $A$-facts $A(x_i, v_i)$.
When we get $A(x_n, v_n)$, we deduce whole $Q$ using the final rule.
Thus $Q$ holds in every ``test'', and the algorithm returns true.
Therefore we have monotonic determinacy.
\end{proof}

We now claim that there is no Datalog rewriting of $Q$ over $\view$ that works over
finite and infinite $\Sigma$-instances.

Assume for a contradiction that a Datalog rewriting $\Pi$ of
$Q$ exists. We first argue that $\Pi$ must be true in the following infinite
instance of the view schema
\begin{align*}
 \viewinst_\infty =
 \{ & \reach_U(x_0), \\
& V_R(x_0,r_1),V_R(r_1,r_2),\ldots,V_R(r_n,r_{n + 1}), \ldots, \\
& V_A(x_0,x_0), V_A(r_1,l_1),V_A(r_2,l_2),\dots,V_A(r_n,l_n),\dots\}.
\end{align*}
The reason is that $Q$ is true in the following instance that satisfies $\Sigma$:
\begin{align*}
 \inst_\infty =
 \{ & \goal, U(x_0), \\
& R(x_0,r_1),R(r_1,r_2),\dots,R(r_n,r_{n + 1}), \dots, \\
& A(x_0,x_0),A(r_1,l_1),A(r_2,l_2),\dots,A(r_n,l_n),\dots\}.
\end{align*}
and one easily checks that the view image of these two instances are the same.
Because $\Pi$ is a rewriting of $Q$ it follows that $\Pi$ is true in
$\viewinst_\infty$.

Because $\Pi$ is in Datalog, if it
is true in some instance then it is already true in a finite subinstance
of that instance. Thus, there exists an $n$ such that $\Pi$ is true in the instance:
\begin{align*}
 \viewinst_n =
 \{ & \reach_U(x_0), \\
& V_R(x_0,r_1),V_R(r_1,r_2),\dots,V_R(r_{n - 1},r_n), \\
& V_A(x_0,x_0), V_A(r_1,l_1),V_A(r_2,l_2),\ldots,V_A(r_n,l_n)\}.
\end{align*}
But then consider the following instance:
\begin{align*}
 \inst_n =
 \{ & U(l_{n + 1}), \\
& R(x_0,r_1),R(r_1,r_2),\dots,R(r_{n - 1},r_n), \\
& A(x_0,x_0),A(r_1,l_1),A(r_2,l_2),\dots,A(r_n,l_n), \\
& L(x_0,l_1),L(l_1,l_2),\dots,L(l_{n - 1},l_n),L(l_n,l_{n + 1})\}.
\end{align*}
One can easily check that $\viewinst_n \subseteq \views(\inst_n)$. To see that $\inst_n$
satisfies $\Sigma$, observe that the absence of the $\goal$ fact means that  the first two rules
hold vacuously. The third rule  can easily be seen to hold.
And since the unique node satisfying $U$,  $l_{n + 1}$, is not the target of an $A$-edge, the last rule
holds vacuously.  

Note that $Q$ is
false in $\inst_n$, because $\goal$ is not true there. But then $\Pi$ can
not be a rewriting of $Q$, because $\viewinst_n \subseteq \views(\inst_n)$,
$\Pi$ is true in $\viewinst_n$, but $Q$ is false in $\inst_n$.
\end{proof}

\subsection{More detail on Proof of Theorem  \ref{thm:datalogviewsandqquery}: better rewriting using Compactness of Entailment
and Automata for Entailing Witnesses}
Recall the   forward-backward method for generating rewritings for monotonically determined queries over views
in the presence of rules. It is available
when $(Q,\Sigma, \views)$ has the BTVIP. We  form a tree automaton for the chase
of the view images of unfoldings, and then change them into a relational query. If we have the FBTVIP, we can
make due with a finite tree automata, and generate a Datalog query. If we have only the BTVIP -- as is the case when
the chase does not terminate -- then we need a tree automaton over infinite trees to capture the view images, and thus
get only a $\poslfp$ rewriting, as in  Theorem  \ref{thm:fgviewsplusdatalogqueryrewriting}.
But if we can argue for compactness of entailment, we can use
a finite tree automata to capture a sufficiently large portion of the view image, and use the backward mapping to get
a Datalog rewriting again. This is what happens in Theorem  \ref{thm:fgviewsplusfgdlquery}.

We can use compactness of entailment in a more straightforward way to prove
the Datalog rewriting result of Theorem \ref{thm:datalogviewsandqquery}. We restate the result here:

\medskip

Suppose $Q$  a Boolean UCQ, $\views$ a set of Datalog
views, and
 $Q$ is monotonically determined by $\views$ \wrt $\Sigma$. Then
there is a UCQ rewriting of $Q$ over $\views$ relative to $\Sigma$. When $Q$ is a CQ, the rewriting
can be taken to be a CQ as well.

\medskip

Note that when we apply the Datalog views to the chase, we may not get bounded treewidth.

\begin{proof}
For simplicity, we give the argument only when $Q$ is a CQ.
Let $V_0$ be the view image of $\chase_\Sigma(\canondb(Q))$: that is, $V_0$ is the instance formed
in the first few steps of the  process in Figure \ref{alg:query-mondet-constraints}.
Since $Q$ is monotonically determined over $\views$ relative to $\Sigma$, 
 $V_0$ must entail $Q$ using the view-reversing process.
Thus by the compactness result, Theorem \ref{thm:compact}, there is a finite subinstance $V^-_0$ of $V_0$ that entails this.
Changing $V^-_0$ into a CQ gives the desired rewriting.
\end{proof}


\section{Additional material related to the surprising undecidability results, Section \ref{sec:undecide}}
\subsection{Proof of Theorem \ref{thm:undecfrontonepluslinearcqboth}: the undecidability result for CQs and combinations of Linear and Frontier-$1$: } \label{sec:app-undecide}

Recall Theorem \ref{thm:undecfrontonepluslinearcqboth}:

\medskip

The problem of monotonic determinacy is undecidable when $\Sigma$ ranges over combinations of Linear TGDs
and Frontier-$1$ TGDs, $Q$ over Boolean
CQs, and $\views$ over CQ views. If we allow MDL queries, we have undecidability with just Linear TGDs.

\medskip

\begin{proof}
We first focus on the first part of the theorem. The second will follow by replacing the use of Frontier-one TGDs
by MDL rules.

We prove undecidability of monotonic determinacy by reduction from 
\emph{state reachability problem for deterministic one-dimensional cellular automata.}
A cellular automaton is a pair $\A = (\Q, \Delta)$ where 
$\Q = \{T_0, T_1, \dots, T_n\}$ is a set of cell states with 
$T_0$ a blank state, and $\Delta$ is a set of transitions
of the form $T_i T_j T_k \to T_l$ and $T_j T_k \to T_l$. All transitions
should be deterministic. This means that for each tuple of natural numbers $(i, j, k)$ there is at most one
$l$ with $T_iT_jT_k \to T_l \in \Delta$ and for each pair of numbers $(j, k)$ there is at most one
$l$ with $T_jT_k \to T_l \in \Delta$.  A \emph{tape state} is an infinite sequence 
$\sigma_0\sigma_1 \sigma_2 \dots$ of elements of $\Q$ 
such that  there is $i \in \mathbb{N}$ such for all $j \ge i$ $\sigma_j = T_0$.
We say that a tape state $\varrho = \varrho_0\varrho_1\varrho_2 \dots$ \emph{$\A$-follows}
a tape state $\sigma = \sigma_0\sigma_1\sigma_2 \dots$ if $\Delta$ contains
transitions $\sigma_0\sigma_1 \to \varrho_0$ and
$\sigma_{i-1} \sigma_{i} \sigma_{i+1} \to \varrho_{i}$ for $i \ge 1$. 
We say that $T_n$ is reachable in $\A$ if there is a tape state $\varrho$ containing $T_n$
such that the pair $(T_0T_0T_0T_0\dots, \rho)$ is in the transitive closure of the $\A$-follows relation.  
Since a deterministic $1$-dimensional cellular automata can model an arbitrary deterministic Turing machine, the
problem of determining given $\A$ and $T_n$ whether  $T_n$ is reachable in $\A$, is undecidable.

Given a cellular automaton $\A$, we define a Frontier-Guarded Datalog query
$Q$ and a set of CQ views $\views$ such that $Q$ is monotonically determined by $\views$ iff
$T_n$ is reachable in $\A$.

Our schema will use a relational representation of the two-dimensional grid shown in Figure~\ref{figure:procedure}. We
have a vertical axis $y$ which uses $\ysucc$ for the successor relation and a horizontal
axis $x$ which uses $\xsucc$ for the successor relation. Then we have the ``grid points'' in
the quadrant between the axes which are linked via $\vertproj$ and $\horproj$ relations
to their projections on the axes. Unary predicates $\zerox$ and $\zeroy$
mark the origins of the axes. We will also have  ternary relations $\guard$ and $\oguard$.

The following CQs will be useful in the sequel:
\begin{align*}
\horadj(z_1, z_2, y, x_1, x_2) = \vertproj(y, z_1)\wedge \vertproj(y, z_2) \wedge  \horproj(x_1, z_1)\wedge
\horproj(x_2, z_2)\wedge \xsucc(x_1, x_2)
\end{align*}

This says that $z_1$ and $z_2$ have the same $y$-projection, while the $x$-projection of $z_2$
is next to the $x$-projection of $z_1$.
The CQ $\vertadj$ is defined similarly.

It should be clear that $\rightof(z_1, z_2) = \exists y ~ x_1 ~ x_2 \horadj(z_1, z_2, y, x_1, x_2)$ holds
for grid points $z_1$ and $z_2$ iff $z_2$ is the right neighbour of $z_1$. Similarly we can define
$\vertadj(z_1, z_2, y_1, y_2, x)$ checking vertical adjacency and $\downto(z_1, z_2)$ by
existentially quantifying $\vertadj$
indicates that $z_2$ is the downward neighbor of $z_1$.
Note that we can define the  axes and the origin:
\begin{itemize}
\item $\bottomedge(z) = \exists y  ~ \vertproj(y, z) \wedge \zeroy(y)$
\item  $\axisy(z) = \exists x ~ \horproj(x, z) \wedge \zerox(x)$
\item  $\atorigin(z) = \bottomedge(z) \wedge \axisy(z)$
\end{itemize}

Consider the CQ Query: 
\[
Q = \exists z_0 ~ x_0 ~ x_1 ~ y_0 ~ y_1 ~ \guard(z_0, x_0, x_1) \wedge \xzero(x_0) \wedge \oguard(z_0, y_0, y_1) \wedge \yzero(y_0)
\]
Along with the CQ views:
\[
S(x, y) = \horproj(x, z) \wedge \vertproj(y, z)
\]

and the views $V_{\xzero} , V_{\yzero}, V_{\xsucc},V_{\ysucc}$.

Our rules $\Sigma$ consist of several groups:

\begin{itemize}
\item[(I)] Linear TGDs that ``build a grid'':
\begin{itemize}
\item  $\guard(z_0, x, x') \rightarrow \exists x'' ~ \guard(z_0, x', x'')$
\item  $\guard(z_0, x, x') \rightarrow \horproj(x, z_0) \wedge \xsucc(x, x')$
\item  $\oguard(z_0, y, y') \rightarrow \exists y'' ~ \oguard(z_0, y', y'')$
\item $\oguard(z_0, y, y') \rightarrow \vertproj(y,z_0) \wedge \ysucc(y, y')$
\end{itemize}

When chasing $Q$ with these rules we will  get an unbounded $\xsucc$ chain of elements,
with ``grid point'' $z_0$ projecting to all of them, and also a $\ysucc$ chain of
elements  with $z_0$ projecting to all of them.

\item[(II)] Full frontier-one TGDs that simulate the run of the automaton:

\begin{itemize}
\item $\atorigin(x) \rightarrow A(x)$
\item $A(x) \rightarrow T_0(x)$
\item $ \axisy(y) \wedge T_j(y) \wedge \rightof(y,z) \wedge T_k(z) \wedge \downto(y',y) \rightarrow T_l(y')$,

for every transition  $T_j T_k \to T_l \in \Delta$,
\item $T_i(x) \wedge \rightof(x,y) \wedge T_j(y) \wedge \rightof(y,z) \wedge T_k(z) \wedge \downto(y',y) \rightarrow T_l(y')$,

for every transition  $T_i T_j T_k \to T_l  \in \Delta$,
\end{itemize}

\item[(III)] Linear TGDs that mark acceptance:

\[
T_n(v) \rightarrow \exists x_0 ~ x_1 ~ y_0 ~ y_1 ~ z_0 ~ \guard(z_0, x_0, x_1) \wedge \xzero(x_0) \wedge \oguard(z_0, y_0, y_1) \wedge \yzero(y_0)
\]

for each final state $T_n$
\end{itemize}

\begin{claim} $T_n$ is reachable in $\A$ iff $Q$ is monotonically determined over $\views$ with respect to $\Sigma$.
\end{claim}

\begin{figure}
\vbox{
\centerline{
\scalebox{0.5}{\includegraphics{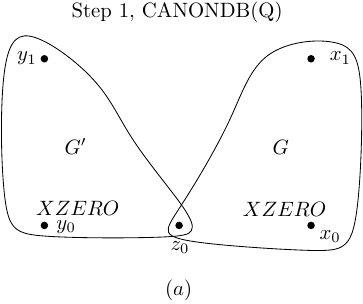}}\hspace*{1cm}
\scalebox{0.5}{\includegraphics{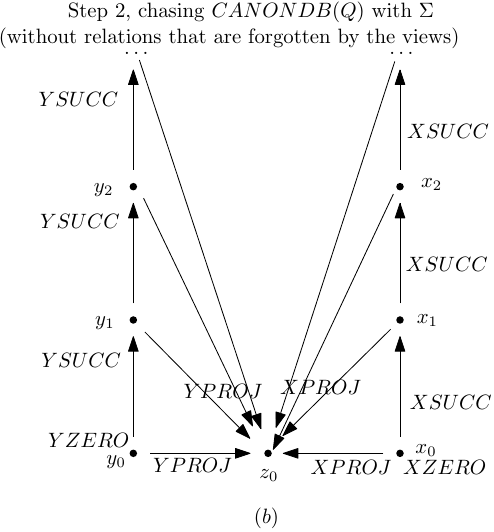}}\hspace*{1cm}
\scalebox{0.5}{\includegraphics{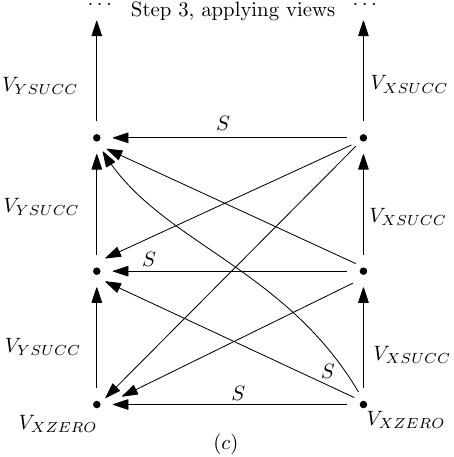}}}\vspace*{8mm}
\centerline{
\scalebox{0.7}{\includegraphics{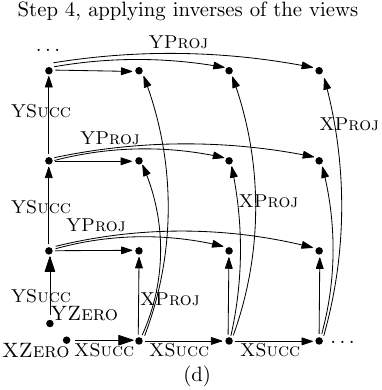}}\hspace*{1cm}%
\scalebox{0.7}{\includegraphics{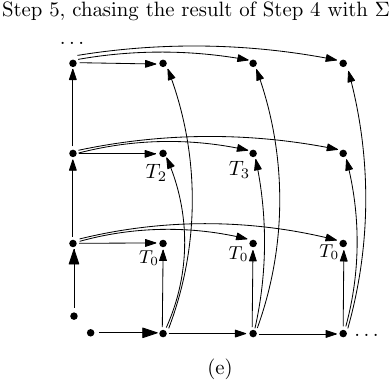}}
}}
\caption{Applying the monotonic determinacy process of Figure~\ref{alg:query-mondet-constraints} to $Q$, $\Sigma$ and views from Theorem~\ref{thm:undecfrontonepluslinearcqboth}.}
\label{figure:procedure}
\end{figure}

\begin{proof}
The idea of the proof is shown graphically in Figure~\ref{figure:procedure}. It shows how the monotonic determinacy checking procedure from
Figure~\ref{alg:query-mondet-constraints} works on $Q$, $\Sigma$ and views constructed in the proof of Theorem~\ref{thm:undecfrontonepluslinearcqboth}.
In particular, (a) shows the canonical database of $Q$, (b) shows the result of chasing it with $\Sigma$ omitting the symbols that are ``forgotten'' by the views, 
(c) shows the result of applying the views to (b), 
(d) shows the result of applying the ``view inverting rules'' (choosing a body for each view fact) to the result of (c), 
yielding a grid-like structure.
Then  (e) shows how chasing (d) with $\Sigma$ simulates the work of the cellular automaton $\A$ with transitions $T_0T_0 \to T_2$ and $T_0T_0T_0 \to T_3$.
It should be clear that due to Part 3 of $\Sigma$, $Q$ holds in the ultimate result of (e) iff $T_n$ is reachable in $\A$. 
\end{proof}

The first part of Theorem \ref{thm:undecfrontonepluslinearcqboth} follows immediately from the claim and the undecidability of 
the state reachability problem for deterministic one-dimensional cellular automata.

For the second part of Theorem \ref{thm:undecfrontonepluslinearcqboth}, we set $\Sigma'$ to consist only of TGDs of group (I). TGDs of groups 
(II) and (III) we convert into a MDL query $Q_{MDL}$ by interpreting TGDs of group (II) as rules of a Datalog program and adding the rule 
$\goal \datalogarrow T_n(x)$. We set $Q' = Q \lor Q_{MDL}$ and argue that

\begin{claim} $T_n$ is reachable in $\A$ iff $Q'$ is monotonically determined over $\views$ with respect to $\Sigma'$.
\end{claim}
 
 To see this, note that every CQ-approximation of  $Q'$ is either CQ-approximation of $Q$ or CQ-approximation of $Q_{MDL}$ and
 that CQ-approximations of $Q_{MDL}$ do not change after chasing by $\Sigma'$ and are preserved after applying views followed by 
 ``view-inverting rules''. Thus the monotonic determinacy checking procedure from
Figure~\ref{alg:query-mondet-constraints} always succeeds on CQ-approximation of $Q_{MDL}$. On the other hand, arguing similar to Claim 1, 
we can show that this procedure, when executed on $Q$, simulates $\A$, and succeeds if and only if $T_n$ is reachable in $\A$. 
\end{proof}

\subsection{Proof  of Theorem \ref{thm:undecfrontonecqucq}: Undecidability with UCQ views with only UIDs}

Recall the statement of Theorem
\ref{thm:undecfrontonecqucq}

\medskip
The problem of monotonic determinacy is undecidable when $\Sigma$ ranges over Unary
Inclusion Dependencies (which are Frontier-$1$ and linear),
$Q$ over Boolean
UCQs, and $\views$ over UCQ views.

\medskip

We start with a warm-up, in which we use a richer class of rules.

\begin{proposition} \label{prop:undecfrontonecqucq}
The problem of monotonic determinacy is undecidable when $\Sigma$ ranges over Frontier-$1$ TGDs, $Q$ over Boolean
atomic CQs, and $\views$ over UCQ views.
\end{proposition}

Note that this is the same as our goal theorem, but we allow frontier-$1$ TGDs, rather than only linear frontier-$1$ TGDs.

\begin{proof}
We consider a tiling problem consisting of a set of tiles and horizontal and vertical constraints. 
The goal is to find a tiling of the infinite quarter-plane satisfying the constraints.

We take $Q  = \exists x ~ \init(x) $.
We set $\Sigma$ to be the union of $\Sigma_{\qstart}$ and $\Sigma_{\qverify}$ where $\Sigma_{\qstart}$ is
$$
\begin{array}{rcl}
\init(x) &\rightarrow& A_1(x) \wedge A_2(x) \\[2mm]
A_1(x) &\rightarrow& \exists y ~ \xsucc(x,y) \wedge A_1(y) \\[2mm]
A_2(x) &\rightarrow& \exists y ~ \ysucc(x,y) \wedge A_2(y) \\[2mm]
\end{array}
$$
and $\Sigma_{\qverify}$ consists of the following TGDs:
\begin{enumerate}[{(R}1{)} ]
\item $A_1(x) \rightarrow \exists y ~ \init(y)$ and $A_2(x) \rightarrow \exists y ~ \init(y)$ 
\item $\horadj(z,z') \wedge T_i(z) \wedge T_j(z') \to \exists x ~ \init(x) $ \\[2mm] where
$\horadj(z,z')$ is $\exists x x' y ~ \horproj(z,x) \wedge \vertproj(z,y) \wedge \horproj(z',x') \wedge 
\vertproj(z', y) \wedge \xsucc(x,x')$, and $i,j$ range over pairs violating the horizontal constraint of the tiling;
\item and similar rules for the vertically incompatible tiles.
\end{enumerate}

The set of views $\views$ consists of only the \emph{grid-generating view}, which is a UCQ:
$$
\begin{array}{rcl}
S(x,y) & \datalogarrow & A_1(x) \datalogwedge A_2(y)\\
S(x,y) & \datalogarrow & \xproj(x,z) \datalogwedge T_i(z) \datalogwedge \yproj(y,z) \mbox{ for all } T_i \mbox{ in } Tiles;\\
\end{array}
$$

We claim that a tiling exists exactly when $Q$ is not monotonically determined over $\views$ with respect to $\Sigma$.

In one direction, suppose a tiling exists, and assume that its underlying domain is the set of pairs of positive integers.
 We let $\inst$ be obtained from chasing $\init(x_0)$ with
$\Sigma$, and let $\jnst$ be the view image of $\inst$. We can identify the domain of $\jnst$ with pairs of integers,
and let $\inst'$ be
obtained by taking the tiling and changing its domain to  be the domain of $\jnst$.
We see that $\inst, \inst'$ represent a counterexample to monotonic determinacy.

In the other direction, suppose monotonic determinacy fails. Then some 
iteration of the loop
in Figure \ref{alg:query-mondet-constraints} will fail, where an iteration consists of:
\begin{enumerate}[{(stage }1{)}]
\item taking the canonical database of $Q$ 
\item chasing with $\Sigma$ its canonical database
\item applying the views
\item choosing disjuncts in chasing the inverse of the view rules
\item chasing again with $\Sigma$
\end{enumerate}

The instance formed in the failing iteration $\inst'$ must be based on $\{\init(x_0)\}$. Note that after choosing disjuncts but prior to the final chase step,
$\inst'$ could not have any $A_1$ or $A_2$ atoms, because otherwise $\inst'$ after the final stage would satisfy
$Q$ because of (R1). Therefore $\inst'$, prior to the last step, must be a proper grid with tiles given as $T_i$-atoms. 
As we know that $\inst'$ does not satisfy  $Q$,
it follows that it yields a proper tiling of the quarter-plane.

\end{proof}

With the warm up over, we are now ready for the proof the theorem:

\begin{proof}
We consider a tiling problem consisting of a set of tiles, an initial tile, vertical constraints
and horizontal constraints. The goal is to find an infinite tiling satisfying the constraints.

We set $\Sigma$ to be as follows:
$$
\begin{array}{rcl}
\init(x) &\rightarrow& A_1(x) \wedge A_2(x) \\
A_1(x) &\rightarrow& \exists y ~ \xsucc(x,y) \wedge A_1(y) \\
A_2(x) &\rightarrow& \exists y ~ \ysucc(x,y) \wedge A_2(y) \\
A_1(x) &\rightarrow& \exists y ~ \init(y)\land \origin(y) \\
A_2(x) &\rightarrow& \exists y ~ \init(y)\land \origin(y) 
\end{array}
$$
These can clearly be converted to UIDs.

$Q_{\qverify}$ is the disjunction of the following CQs:
\begin{enumerate}[{(CQ}1{)} ]
\item $\horadj(z,z') \wedge T_i(z) \wedge T_j(z')$ where
$\horadj(z,z')$ is $\exists x x' y ~ \horproj(z,x) \wedge \vertproj(z,y) \wedge \horproj(z',x') \wedge 
\vertproj(z', y) \wedge \xsucc(x,x')$, and $i,j$ range over pairs violating the horizontal constraint of the tiling;
\item similar CQs for the vertically incompatible tiles;
\item $\origin(x) \wedge T_i(x)$ for any $i$ not equal to the (index of the) initial tile;
\end{enumerate}

The set of views $\views_{TP}$ consists of
\begin{itemize}
\item[--] the \emph{grid-generating
view}
$$
\begin{array}{rcl}
S(x,y) & \datalogarrow & A_1(x) \datalogwedge A_2(y)\\
S(x,y) & \datalogarrow & \xproj(x,z) \datalogwedge T_i(z) \datalogwedge \yproj(y,z) \mbox{ for all } T_i \mbox{ in } Tiles;\\
\end{array}
$$
\item[--] the \emph{atomic views} $V_\ysucc$, $V_\xsucc$, $V_{\origin}$, $V_{T_i}$ for EDBs 
$\ysucc,\xsucc, \origin$  and each
$T_i$ in $Tiles$;
\item[--] the following \emph{special} views
$$
\begin{array}{crcl}
(SP3) &  V_{\ha}(z_1, z_2, y,x_1,x_2) & \datalogarrow & \ha (z_1, z_2)\\
(SP4) & V_{\va}(z_1, z_2, y_1,y_2,x) & \datalogarrow & \va (z_1, z_2)\\
\end{array}
$$
\end{itemize}

Now we set $Q_{\qstart} = \init(x)\land \origin(x) $ and $Q = Q_{\qstart} \lor Q_{\qverify}$.

We claim that a tiling exists exactly when $Q$ is not monotonically determined over $\views$ with respect to $\Sigma$.

In one direction, suppose a tiling exists, and assume that its underlying domain is the set of pairs of positive integers.
 We let $\inst$ be obtained from chasing $\init(x_0)$ with
$\Sigma$, and let $V$ be the view image of $\inst$. We can identify the domain of $V$ with pairs of integers,
and let $\inst'$ be
obtained by taking the tiling and changing its domain to  be the domain of $V$.
We see that $\inst, \inst'$ represent a counterexample to monotonic determinacy.

In the other direction, suppose monotonic determinacy fails. Then some iteration of the loop
in Figure \ref{alg:query-mondet-constraints} will fail,
 where an iteration is based on:
\begin{enumerate}[{(stage }1{)}]
\item a choice of one of the disjuncts of $Q$ 
\item chasing with $\Sigma$ its canonical database
\item applying the views
\item choosing disjuncts in chasing the inverse of the view rules
\item chasing again with $\Sigma$
\end{enumerate}
Failing means that $Q$ does not hold after applying these choices.
In the first item, we cannot make a choice of $Q_{\qverify}$, since due to atomic and special views 
all such choices  succeed.

So we assume that the failing iteration  is based on $Q_{\qstart}$. Note that prior to the final chasing stage, the instance formed,
call it $\inst'$, could not have any $A_1$ or $A_2$ atoms, because otherwise $\inst'$ after the final stage would satisfy
$Q_{\qstart}$. Therefore $\inst'$ prior to the final chasing stage must be a proper grid with tiles, and the final chase
oes not introduce any new atoms. As we know that $\inst'$ does not satisfy  $Q_{\qverify}$, it follows
that it yields a proper tiling of the quarter-plane.

\end{proof}

\subsection{Proof of Theorem~\ref{thm:undecuidcqucq}: undecidability with only UIDs, where the query is a CQ and the views
are UCQs}
\newlength\tindent
\setlength{\tindent}{\parindent}
\setlength{\parindent}{0pt}
Recall the statement of Theorem \ref{thm:undecuidcqucq}:

\medskip

The problem of monotonic determinacy is undecidable when $\Sigma$ ranges over Unary
Inclusion Dependencies
$Q$ over Boolean
CQs, and $\views$ over UCQ views.

\medskip

\newcommand{\cols}{\mathcal{T}}
\newcommand{\forb}{\mathcal{F}}
\newcommand{\ver}{\mathtt{vertical}}
\newcommand{\horz}{\mathtt{horizontal}}
\newcommand{\gs}{GridSource}
\newcommand{\badmatch}[1]{\mathtt{BadMatch_{#1}}}
\newcommand{\net}{Net}
\newcommand{\badp}{Q_{\mathtt{bad}}^p}
\newcommand{\badpb}{Q_{\mathtt{bad}}^{p'}}
\newcommand{\ax}{Axis_{x}}
\newcommand{\ay}{Axis_{y}}
\newcommand{\gx}{Grid_{x}}
\newcommand{\gy}{Grid_{y}}
\newcommand{\sig}{\sigma}
\newcommand{\qstrt}{Q_{\mathtt{start}}}
\newcommand{\qnet}{Q_{\mathtt{net}}}
\newcommand{\qfo}{Q_{\mathtt{foundation}}}
\newcommand{\vps}{\vec{v_p}}

\newcommand{\xs}[1]{x^{*}_{#1}}
\newcommand{\ys}[1]{y^{*}_{#1}}
\newcommand{\vs}[1]{v^{*}_{#1}}

\newcommand{\vsps}{\vec{v^*_{p}}}
\newcommand{\qfree}{Q_{\mathtt{free}}}
\newcommand{\qwhole}{Q}

\newcommand{\vti}{V_{T_i}}
\newcommand{\vslot}{V_{\mathtt{slot}}}
\newcommand{\vnet}{V_{\mathtt{net}}}
\newcommand{\vchoice}{V_{\mathtt{chose}_{i}}}
\newcommand{\vchoicei}[1]{V_{\mathtt{chose}_{#1}}}
\newcommand{\vunf}{V_{\mathtt{unfaithful}}}

\newcommand{\gridnode}[2]{\textsc{TileAt}_{#1}^{#2}}
\newcommand{\thirdgrid}[1]{\textsc{Tiling}_{#1}}

\myparagraph{The Undecidable Problem}
We will use a reduction from the problem of coloring a quarter plane. Formally we define the instance of this problem as an pair $\pair{\cols, \forb}$ where $\cols$ is a set of tiles and $\forb$ is a set of forbidden vertical and horizontal pairs, encoded as a subset of $\cols \times \cols \times \set{\horz, \ver}$. We say that tiling admits $p \in \forb$ at coordinates $(n,m)$ if and only if tiles at $\pair{(n,m), (n + 1, m)}$ form $p$ if it is horizontal or tiles at $\pair{(n,m), (n, m + 1)}$ form $p$ if it is vertical. Given an instance  $\pair{\cols, \forb}$ of the tiling problem we ask if there exists a tiling of a plane that admits no forbidden pair, if that is the case we call the tiling {\em valid}.

Let $\inst = \pair{\cols, \forb}$ an instance of the tiling problem.
Having an instance $\inst$ of this problem, we show how to  construct an  instance of the monotonic determinacy
problem.

\myparagraph{Challenge} A key challenge in coding any determinacy problem is that
the query $Q$ plays two roles, highlighted in \cref{alg:query-mondet-constraints}.
On the one hand the query needs to serve as a starting point on which (using the rules)
we can generate something like a  large grid. 
On the other hand, the same query is used at the
end of the process, usually to check for some violation of a correctness property in a tiling
or a Turing Machine run.  We have seen two ways of achieving this ''dual role'' for the query, which we review below.

In the proof of Theorem 
\ref{thm:undecfrontonepluslinearcqboth}, the query is designed to play only the first role.
To allow it to play the second role, we 
use a  rule: when a violation is found a special rule is fired that makes
$Q$ hold.  But this rule is not a UID, so this technique will not be available.

A second technique would be to let $Q$ be a UCQ, with one disjunct playing the ``initialization'' role and the others representing detection of various violations. But this will also
not be available  to us, because we need to use a CQ. 

The main challenge  will be to code verification of a disjunction of the different
tiling violations using a CQ.

\myparagraph{The Reduction}
We now present the reduction.

\noindent
{\bf Signature.}
Our signature $\sig$ will consist of:
\begin{itemize}
    \item Unary relations $\gs$ and $T_i$ for each $i \in \cols$.
    \item Binary relations $\net$, $\ax$, and $\ay$.
    \item Ternary relations $\gx$ and $\gy$.
\end{itemize}

\noindent
{\bf Query construction.}

Define $\qstrt$ as the following 
CQ (see \cref{fig:q-start}):
$$\qstrt(v_0, x_0, x_1, y_0, y_1) = \ax(x_0, x_1) \wedge \ay(y_0, y_1) \wedge  \gs(v_0).$$
\begin{figure}[H]
    \centering
    \includegraphics[height=0.35\textwidth]{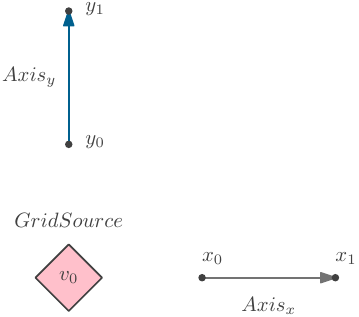}
    \caption{Visualization of $\qstrt$.}
    \label{fig:q-start}
\end{figure}

For $p \in \forb$, $p = \pair{i, j, \ver}$, define $\badp$ as the CQ:
$$\badp(v_p) = \exists{v'_p, x,x',y,y',y''}\; T_i(v_p) \wedge T_j(v'_p) \wedge \gx(v_p, x, x') \wedge \gx(v'_p, x, x') \wedge \gy(v_p, y, y')
\wedge  \gy(v'_p, y', y'').$$
We refer to such CQs as ``vertical $\badp$ queries'': See \cref{fig:bad-vertical}.

For $p \in \forb$, $p = \pair{i, j, \horz}$  define $\badp$ as the CQ:
$$\badp(v_p) = \exists{v'_p, x,x',x'',y,y'}\wedge  T_i(v_p) \wedge T_j(v'_p),\wedge \gx(v_p, x, x') \wedge \gx(v'_p, x', x'') \wedge
 \gy(v_p, y, y') \wedge \gy(v'_p, y, y)'.$$

These are ``horizontal $\badp$'' queries, illustrated in \cref{fig:bad-horizontal}.

Both queries are visualized in the following figures:

\begin{figure}[H]
\begin{minipage}{0.45\textwidth}
    \centering
    \includegraphics[height=0.9\textwidth]{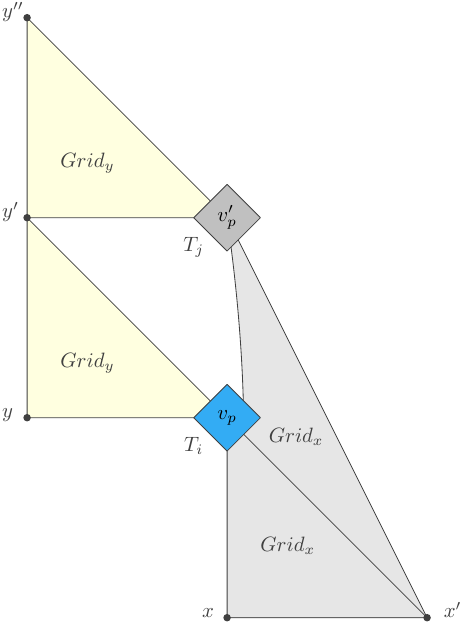}
    \caption{Visualization of a ``vertical $\badp$'' query.  }
    \label{fig:bad-vertical}
\end{minipage}
\hfill
\begin{minipage}{0.45\textwidth}
    \centering
    \vspace{1.175cm}
    \includegraphics[width=0.9\textwidth]{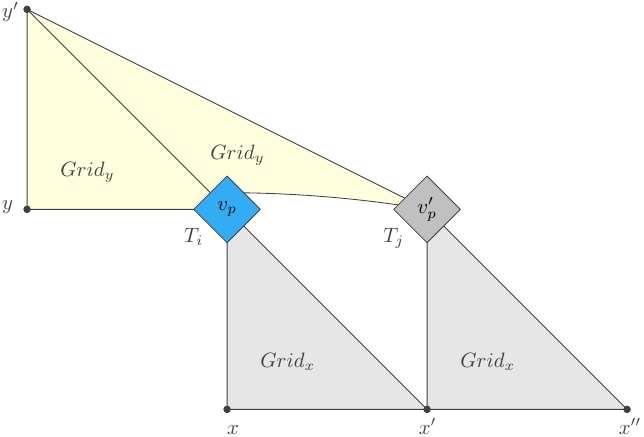}
    \vspace{1.175cm}
    \caption{Visualization of a ``horizontal $\badp$'' query.}
    \label{fig:bad-horizontal}
\end{minipage}
\end{figure}

Let $\vps$ be a vector $\pair{v_{p}}_{p \in \forb}$ of variables. Define $\qnet(v_0, \vps)$ as a $\net$-clique over $\set{v_0} \cup \set{v_p \mid p \in \forb}$ variables. See \cref{fig:q-net} for a visualization.
\begin{figure}[H]
\centering
\begin{minipage}{0.4\textwidth}
    \centering
    \includegraphics[height=0.9\textwidth]{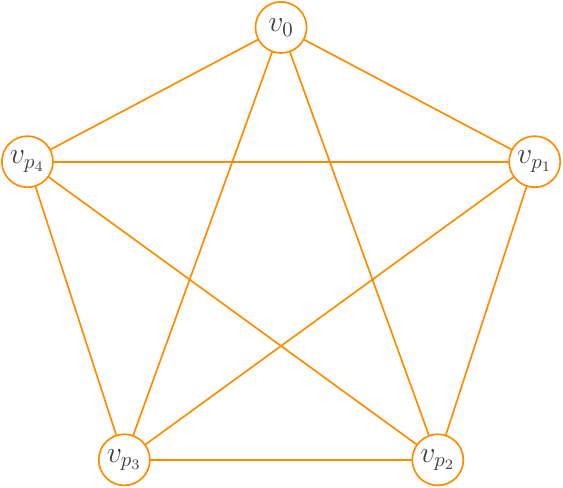}
    \caption{Complete depiction of $\qnet$ for $|\forb| = 4$.}
    \label{fig:q-net}
\end{minipage}
\end{figure}


Finally we define $\qfree(v_0, x_0, x_1, y_0, y_1, \vps)$ as: 
$$\qstrt(v_0, x_0, x_1, y_0, y_1) \wedge \qnet(v_0, \vps) \wedge \bigwedge_{p \in \forb}\badp(v_p).$$ And $\qwhole$ as the Boolean
CQ obtained by fully existentially quantifying $\qfree$. See \cref{fig:q-free} and \cref{fig:q-free-alternate}.

\begin{figure}[H]
\begin{minipage}{0.4\textwidth}
    \centering
    \includegraphics[height=0.9\textwidth]{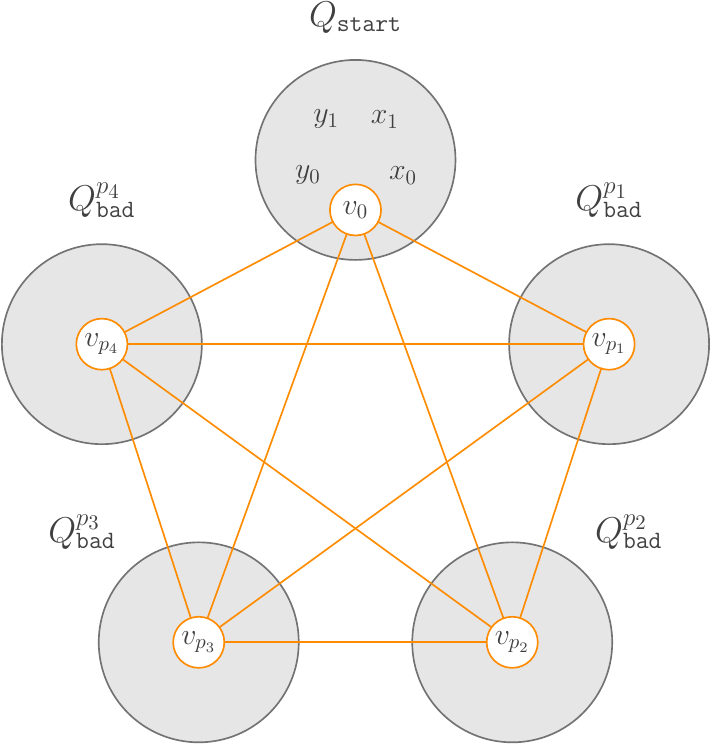}
    \caption{High level visualization of vertical $\qfree$ with free variables and $\net$ atoms exposed.}
    \label{fig:q-free}
\end{minipage}
\hfill
\begin{minipage}{0.4\textwidth}
    \centering
    \includegraphics[height=0.85\textwidth]{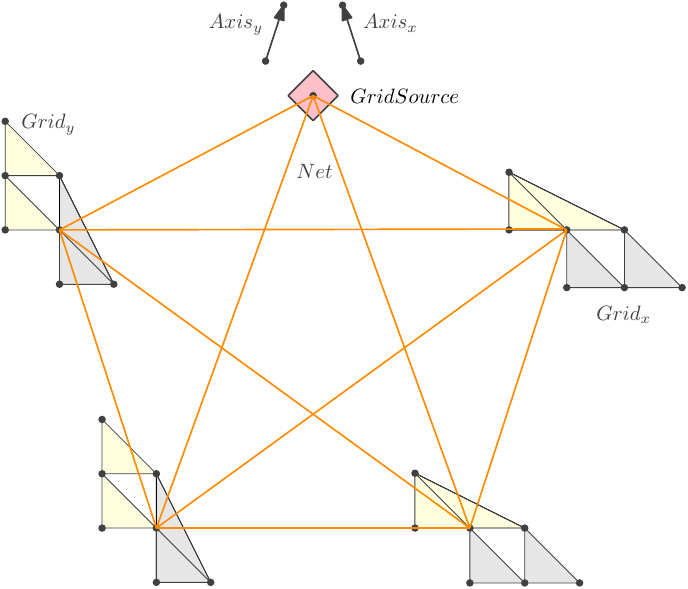}
    \caption{Structural depiction of $\qfree$ with every atom exposed.}
    \label{fig:q-free-alternate}
\end{minipage}
\end{figure}


{\bf Tuple Generating Dependencies.}
We will take our UIDs $\Sigma$ to  be:
$$ \ax(x,x') \rightarrow \exists{x''} \ax(x', x'') \quad\text{and}\quad \ay(y,y') \rightarrow \exists{y''} \ay(y', y'').$$

The structure $\chase_{\Sigma}(\canondb(Q))$
 will then consist of $\canondb(Q)$ with two infinite $\ax$ and $\ay$ chains. For clarity we name the elements of those chains 
$\xs{0}, \xs{1}, \xs{2}, \ldots$ and $\ys{0}, \ys{1}, \ys{2}, \ldots$ respectively. Moreover, to distinguish constants of  $\canondb(Q)$ from variables of $Q$ we will use a star in the upper-right subscript, for example a constant $\vs{0}$ of $\canondb(Q)$ will correspond to variable $v_0$ of $Q$. Finally, we will use $\vsps$ to denote domain elements of the chase that correspond to $\vps$. For further clarification see \cref{fig:m-structure}, which shows
the structure, with the canonical database at the bottom right.

\begin{figure}[H]
\begin{minipage}{0.50\textwidth}
    \centering
    \includegraphics[width=0.8\textwidth]{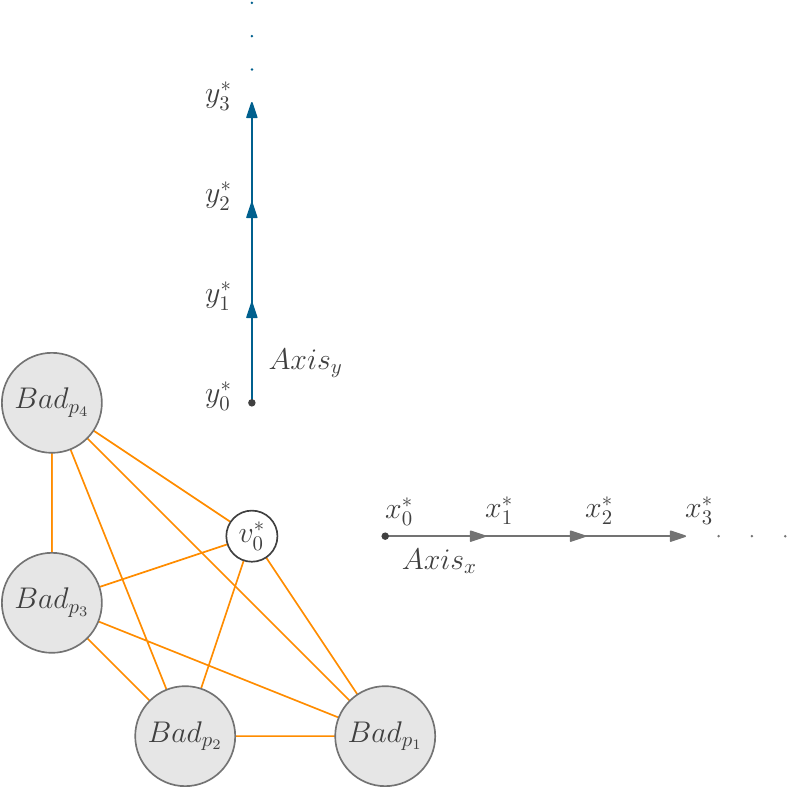}
    \caption{Depiction of the chase $\chase_{\Sigma}(\canondb(Q))$.}
    \label{fig:m-structure}
\end{minipage}
\end{figure}

\medskip

\noindent {\bf View construction.}

We will have a single UCQ view. Intuitively, there is a ``start disjunct'' that will match in the canonical database of the chase to produce many view facts.
Then there will be some ''unfaithful'' disjuncts, representing alternative ways of producing these facts. A choice of these disjuncts will represent  an attempt
to tile the quarter plane. Applying $Q$ at the final stage will correspond to checking correctness.

Define $Slot(x)$ as a conjunction of all atoms over our schema $\sig$ in one free variable $x$, with the exception of $Net(x,x)$.  
We will refer to the structure defined by this query as a {\em slot}. 
Then define $\vslot$ as $\bigwedge_{p \in \forb} Slot(v_p)$. In the end,  $\qfree$ will be a disjunct of our single UCQ view $V$.
Note that $\vslot$ shares free variable 
$v_p$ for $p \in \forb$ with $\qfree$: $v_p$ appears in a $\badp$ conjunct of $\qfree$. 
Query $\vslot$ is depicted in \cref{fig:v-slot}.
\begin{figure}[H]
    \centering
    \includegraphics[width=0.9\textwidth]{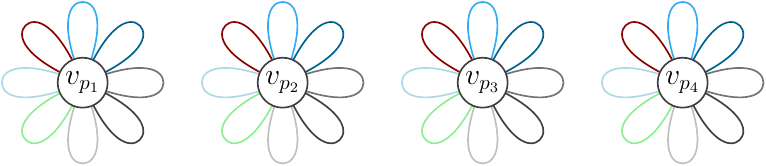}
    \caption{Depiction of $\vslot$. Note free variables shared with $\qfree$.}
    \label{fig:v-slot}
\end{figure}

We define $\vnet$ as $\qnet$. 

Then we define $\vchoice(v_0, x_0, x_1, y_0, y_1, \vps)$ as the  CQ: $$T_i(v_0) \wedge  \gx(v_0, x_0, x_1) \wedge  \gy(v_0, y_0, y_1) \wedge \vnet(v_0, x_0, x_1, y_0, y_1, \vps) \wedge \vslot(\vps).$$
$\vchoice$ is  depicted in \cref{fig:v-chose-i}. 

Again, keeping in mind that  $\qfree$ will eventually be a disjunct of our single UCQ view $V$ note that $\vchoice$ shares free variables $v_0, x_0, x_1, y_0$, and $y_1$ with $\qfree$. Those variables appear in the $\qstrt$ conjunct of $\qfree$. 

\begin{figure}[H]
\centering
\begin{minipage}{0.4\textwidth}
    \centering
    \includegraphics[width=0.9\textwidth]{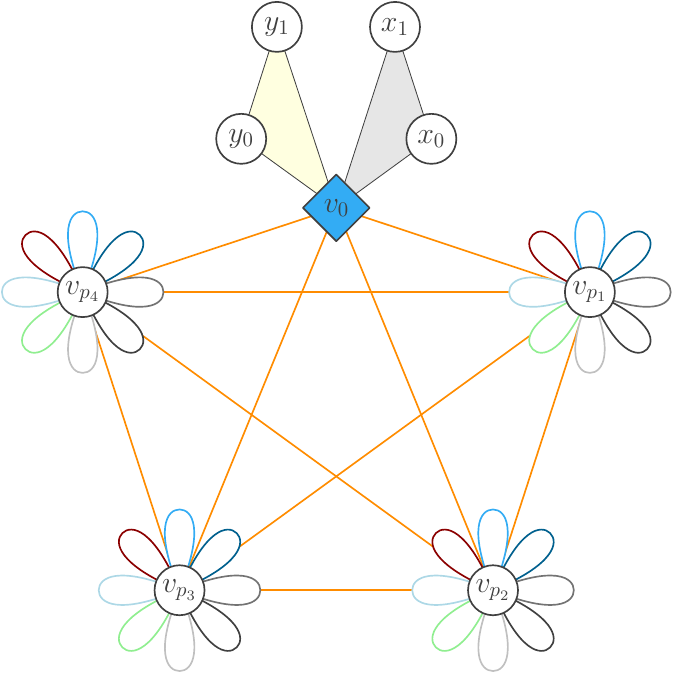}
    \caption{Depiction of $\vchoice$.}
    \label{fig:v-chose-i}
\end{minipage}
\hfill
\begin{minipage}{0.5\textwidth}
\vspace{8.35mm}
\centering
\begin{minipage}{0.6\textwidth}
$\mondet(Q,\views, \Sigma)$:
\begin{algorithmic}[1]
    \For{$Q_n$ approximation of  $Q$}
    \State $C_n \defeq \chase_{\Sigma}(\canondb(Q_n))$
     \State $\viewinst_n \defeq \views(F_n)$                       
        \For{$Q'_{m,n} \in \backview_{\views}(D_n)$}             
         \State $C'_{m,n} \defeq \chase_{\Sigma}(Q'_{m,n})$
        \If{$C'_{m,n} \not \models Q$}                      
            \State \textbf{return} \textbf{false}
        \EndIf
         \EndFor
      \EndFor
            \State \textbf{return} \textbf{true}
\end{algorithmic}
\end{minipage}
\vspace{8.25mm}
\caption{Process for checking monotonic determinacy.} \label{fig:alg-restated}
\end{minipage}
\end{figure}

We let $\vunf(v_0, x_0, x_1, y_0, y_1, \vps) = \bigvee_{i \in \cols} \vchoice(v_0, x_0, x_1, y_0, y_1, \vps)$. And finally we are able to define the unique view of our schema:
\[
V(x_0, x_1, y_0, y_1, \vps) = \exists{v_0} ~ \qfree(v_0, x_0, x_1, y_0, y_1, \vps) \vee \vunf(v_0, x_0, x_1, y_0, y_1, \vps)
\]

\bigskip
\noindent
\noindent
\myparagraph{Correctness of the reduction}
With $Q, V, \Sigma$ constructed, let us show that the transformation is a correct
reduction from tiling to monotonic determinacy.
We consider the process \cref{alg:query-mondet-constraints}, restated for convenience in Figure \ref{fig:alg-restated}. For the remainder of this section we will focus on proving the following lemma:

\begin{lemma}\label{lem:app-third-main}
There exists a valid tiling of $\pair{\cols, \forb}$ if and only if $\mondet(Q,\views, \Sigma)$ returns false.
\end{lemma} 

\medskip

We start by taking the view image of the chase of the canonical database.
We note that the ``unfaithful disjuncts'' will not contribute anything:

\begin{claim}\label{claim:v-unfaithful-does-not-map-to-q}
$\vunf$ does not map to $\chase_{\Sigma}(\canondb(Q))$.
\end{claim}
\begin{proof}
$\vunf$ contains $\gx$ and $\gy$ atoms, which do not appear in $\qwhole$.
\end{proof}

In contrast,  $\qfree$ will map to $\chase_{\Sigma}(\canondb(Q))$ and produce a number of view facts.

\begin{claim}\label{claim:q-free-mappings-to-m-h-identity}
Let $h$ be any homomorphism from $\qfree$ to $\chase_{\Sigma}(\canondb(Q))$.
Then $h(v_i) = \vs{i}$ for any variable subscript $i$. Thus different view facts generated by such an $h$ can only differ on the  ``axis-related variables'' $x_0, x_1, y_0$, and $y_1$. 
\end{claim}
\begin{proof}
Note that:
\begin{enumerate}
\item
For every $\badp$, its only automorphism is the identity. 
\item  Every $\badp$ can be mapped only to itself. 
\item  Variables of $\qnet$ need to be mapped bijectively onto themselves due to the $\net$-clique structure of $\qnet$. 
\end{enumerate}
The claim follows easily from these observations.
\end{proof}

Thus we will generate a collection of view facts for different matches of $\qfree$.
Now we characterize exactly what these view facts are.

\begin{claim}\label{claim:q-free-mappings-axis}
Let $H$ be the set of homomorphisms from $\qfree$ to $\chase_{\Sigma}(\canondb(Q))$. Then $\set{h(x_0, x_1, y_0, y_1) \mid h \in H} = \set{\pair{\xs{n}, \xs{n + 1}, \ys{m}, \ys{m + 1}} \mid n,m \in \mathbb{N}}$.
\end{claim}
\begin{proof}
Here we note that:
\begin{enumerate}
\item  The set of $\ax$ atoms contained in $\chase_{\Sigma}(\canondb(Q))$  is $\set{\ax(\xs{n}, \xs{n + 1}) \mid n \in \mathbb{N}}$. 
\item  The set of $\ay$ atoms $\chase_{\Sigma}(\canondb(Q))$ contains is $\set{\ay(\ys{m}, \ys{m + 1}) \mid m \in \mathbb{N}}$. 
\item  Atom $\ax{x_0, x_1}$ can be freely mapped to any $\ax$ atom of $\chase_{\Sigma}(\canondb(Q))$. 
\item  Atom $\ay{y_0, y_1}$ can be freely mapped to any $\ay$ atom of $\chase_{\Sigma}(\canondb(Q))$.
\end{enumerate}
The claim follows directly from these observations.
\end{proof}

Thus, we conclude that we have view facts corresponding  to axis-related variables as follows:

\begin{lemma} \label{lem:imageunfaithful}
When applying the view $V(x_0, x_1, y_0, y_1, \vps)$ to  $\chase_{\Sigma}(\canondb(Q))$, we get view facts for  exactly these tuples:
$$ \set{\pair{\xs{n}, \xs{n+1}, \ys{m}, \ys{m + 1}, \vsps} \mid n,m \in \mathbb{N}}.$$
\end{lemma}
\begin{proof}
\begin{align*}
V(\chase_{\Sigma}(\canondb(Q))) \quad=&\quad \qfree(\chase_{\Sigma}(\canondb(Q))) \tag{\cref{claim:v-unfaithful-does-not-map-to-q}}\\
\qfree(\chase_{\Sigma}(\canondb(Q))) \quad=&\quad  \set{\pair{\xs{n}, \xs{n+1}, \ys{m}, \ys{m + 1}, \vsps} \mid n,m \in \mathbb{N}} \tag{Claims~\ref{claim:q-free-mappings-to-m-h-identity} and \ref{claim:q-free-mappings-axis}}
\end{align*}
\end{proof}

Now let us inspect the structures considered in Line~\ref{alg:query-determinacy-constraints:view-acc-inverse} of $\mondet(Q,\views, \Sigma)$. 
For each view fact, we will choose a view definition, either as a $\qfree$ disjunct or an unfaithful disjunct.
Whenever a structure is created by choosing  a $\qfree$ disjunct of $V$,  $Q$ will be satisfied 
(the return in Line ~ \ref{alg:query-determinacy-constraints:return}  will not fire), 
and thus we continue with the process.

Let $\unfaithfultests$, the set of \emph{unfaithful tests},   be the set of all structures considered in the {\bf for} loop of \cref{alg:query-mondet-constraints}
that are created by taking a $\vunf$ disjunct of $V$.

Consider applying the ``inverse rule'' corresponding to $\,\vunf$, applied
 to an element $\pair{\xs{n}, \xs{n+1}, \ys{m}, \ys{m + 1}, \vsps}$ within $V(\chase_{\Sigma}(\canondb(Q)))$. As $\vunf$ is a UCQ this means applying one of its $\vchoicei{t}$ disjuncts. 
More precisely, we apply the inverse of $\exists{v_0}\; \vchoicei{t}(v_0, x_0, x_1, y_0, y_1, \vps)$ to $\pair{\xs{n}, \xs{n+1}, \ys{m}, \ys{m + 1}, \vsps}$.\footnote{Note, that the variable $v_0$ is existentially quantified in $V$.} 


From now on, for convenience, we rename the fresh element created from applying a step for an unfaithful test above as $\vs{i,j}$. Also, let:
$$\gridnode{n,m}{t} ::= \vchoicei{t}(\vs{i,j}, \xs{m}, \xs{n+1}, \ys{m}, \ys{m + 1}, \vsps).$$


\begin{definition}\label{def:third-grid}
Let $f\in  \cols^{\mathbb{N}\times \mathbb{N}}$ be a function encoding a tiling of a quarter plane. Then let 
$$\thirdgrid{f} ::= \bigcup_{\pair{n,m}\in \mathbb{N}\times\mathbb{N}} \gridnode{n,m}{f(n,m)}.$$
\end{definition}

From  Lemma \ref{lem:imageunfaithful} 
we conclude that the structures of interest to us in Line~\ref{alg:query-determinacy-constraints:view-acc-inverse}  -- that is,  those contained in $\unfaithfultests$ --  will have the 
shape of the structures above.

\begin{corollary}\label{cor:app-third-ut-encodes-tilling}
For every (not necessarily valid), tiling $f\in  \cols^{\mathbb{N}\times \mathbb{N}}$ of the quarter plane the structure $\thirdgrid{f}$ is contained (up to isomorphism) in $\unfaithfultests$.
Conversely, every member of $\unfaithfultests$ is isomorphic to such a structure.
\end{corollary}
\begin{proof}
Consider \cref{def:third-grid}. Then
the range of $i$ and $j$ comes from Lemma \ref{lem:imageunfaithful}, 
and the range of $f$ comes from disjuncts of $\vunf$.
\end{proof}

This means that each element of $\unfaithfultests$ is a grid-like structure, part of which is depicted on \cref{fig:n-structure}.
\begin{figure}[H]
\begin{minipage}{0.5\textwidth}
    \centering
    \includegraphics[width=0.9\textwidth]{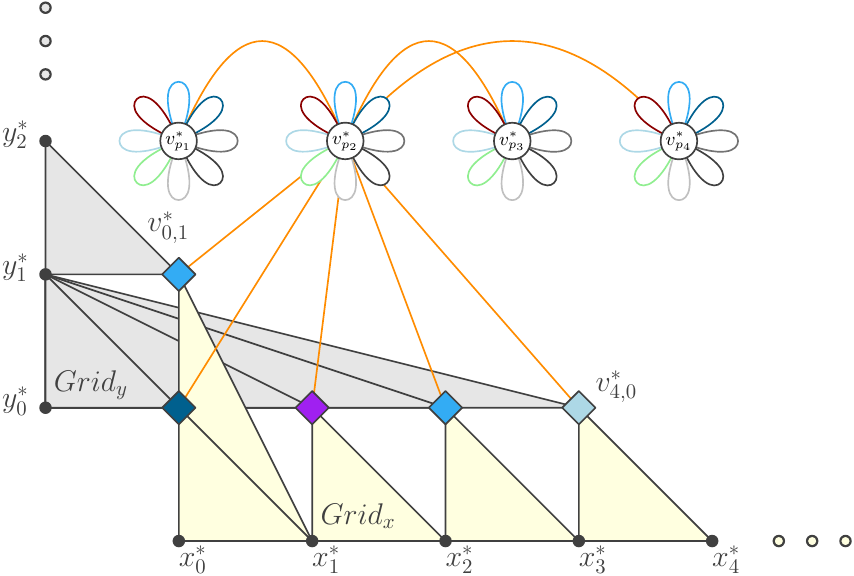}
    \caption{Partial depiction of a structure $\ut$.}
    \label{fig:n-structure}
\end{minipage}
\end{figure}

In the absence of any assumption of monotonic determinacy, the structures in $\unfaithfultests$, based on the $T_i(v_{m,n})$ atoms,  define all tilings of a quarter plane, including the invalid ones.

\begin{definition}
For $p \in \forb$ and natural $n,m$ define a mapping $\badmatch{p, n, m}$ from the variables of $\badp$ to the domain of $\ut$ as:
$$
\badmatch{p, n, m} = \begin{cases}
\pair{x, x,', x'', y, y', v_p, v_p'}\mapsto\pair{\xs{n}, \xs{n + 1}, \xs{n + 2}, \ys{m}, \ys{m + 1}, \vs{n,m}, \vs{n + 1, m}} & \text{when $p$ is a horizontal pattern,}\\
\pair{x, x,', y, y', y'', v_p, v_p'}\mapsto\pair{\xs{n}, \xs{n + 1}, \ys{m}, \ys{m + 1}, \ys{m + 2}, \vs{n, m}, \vs{n, m + 1}} & \text{when $p$ is a vertical pattern.}
\end{cases}
$$
\end{definition}

\begin{claim}\label{claim:third-qbad-mappings}
There exists $p \in \forb$ appearing at coordinates\footnote{Recall that this means tiles at $\pair{(n,m), (n + 1, m)}$ 
form $p$ if it is horizontal, or tiles at $\pair{(n,m), (n, m + 1)}$ form $p$ if it is vertical.} $(n,m)$ in the tiling defined by $\ut \in \unfaithfultests$ if and only if $\badmatch{p,n,m}$ is a homomorphism from $\badp$ to $\ut$.
\end{claim}
\begin{proof}
 Let $f \in\cols^{\mathbb{N}\times \mathbb{N}}$ define the tiling of $\ut$.

($\Rightarrow$) Assume that one of the horizontal  forbidden patterns $p = \pair{i,j,\horz} \in \forb$ appears in the tiling defined by $f$ at coordinates $\pair{n,m}$. That is, $\pair{f(n,m), f(n+1, m)} = \pair{i,j}$. Then $\badmatch{p,n,m}$ clearly defines a homomorphism from $\badp$ to $\ut$. The case for vertical forbidden pattern is analogous.

($\Leftarrow$) Assume that $\badmatch{p,n,m}$ is a homomorphism from $\badp$ to $\ut$. Assume $p = \pair{i,j,\horz}$, then $\pair{f(n,m), f(n+1, m)} = \pair{i,j}$ thus $p$ appears in the tiling defined by $\ut$.
\end{proof}

Thus we know that incorrectness of a tiling associated with a member $\ut$ of $\unfaithfultests$
corresponds to mapping one of the $\badp$ queries, which are a component of $\qwhole$, into $\ut$.
 We still need to show that an unsuccessful tiling corresponds to mapping the entirety of $\qwhole$, which
involves a mapping witnessing all $\badp$ queries at once. That is, we need to solve
the challenge mentioned early, due to the fact that $\qwhole$ is a CQ.

Informally, the idea is as follows. In one direction, suppose the tiling is invalid. This is witnessed
by
 the presence of some forbidden pattern
$p$.  Towards defining a homomorphism of
$\qwhole$, we  map $v_0$ to the constant $\vs{p}$,
 and the other variables $v_{p'}$ 
to the constants $\vs{p'}$ corresponding to themselves. For a given $p' \neq p$,  we
can  extend the mapping
to become a homomorphism of  $\badpb$ in a ``vacuous way'' -- by mapping all quantified 
variables of $\badpb$ into $\vs{p'}$. This will give us a homomorphism
because  $\vs{p'}$ satisfies
each unary predicate, other than $\net$ self-loops.
Turning to the variable $v_p$, we can extend to a homomorphism of
$\badp$,  because $p$ is a violation. 

In the other direction, suppose we have a homomorphism $h$ of $\qwhole$. All variables $v_p$ must map
either to constant $\vs{p'}$  for some $p'$,  or to some constant $\vs{n,m}$.
Because there are no $\net$ self-loops, these variables must map injectively. Thus there
is one variable $v_p$ which must map on to some constant $\vs{n,m}$. The fact that
$h$ is a homomorphism and the properties of $\vs{n,m}$ imply
that the forbidden pattern $p$ occurs in the tiling.

We now explain more formally.

\begin{lemma}\label{lem:app-third-ut-test}
The tiling defined by $\ut \in \unfaithfultests$ is invalid if and only if $\qwhole$ maps to $\ut$.
\end{lemma}
\begin{proof}
We will start with a number of observations. Recall that from \cref{cor:app-third-ut-encodes-tilling} the shape of $\ut$ is as in \cref{def:third-grid}.

Consider a $\qstrt$ conjunct. It consists of three atoms: $\ax(x_0, x_1)$, $\ax(y_0, y_1)$,  and $\gs(v_0)$. For $x_0$ and $x_1$,  the
only way to map them homomorphically is to use one of $\set{\pair{x_0, x_1}\mapsto\pair{\vs{p}, \vs{p}} \mid p \in \forb}$ mappings as $\ax$ (and $\ay$) hold only for elements of $\vsps$ in $\ut$. 
The same can be said about $\pair{y_0, y_1}$. From now on, we will disregard the
mapping of $x_0, x_1, y_0$, and $y_1$ variables, as they can be always properly mapped to $\ut$, and are not appearing in $\qnet$. 

The atom $\gs(v_0)$ can only be homomorphically mapped to a $Slot(\vs{p})$ substructure of $\ut$, since $\gs$ atoms appear only there. 

Consider a $\badp$ conjunct of $\qwhole$.
Note that it can be homomorphically mapped to any $Slot(\vs{p'})$, no matter which
$p'$ is considered, as all atoms except $\net$ hold for $v_{p'}$. Observe that there are no other mappings of $\badp$ to $\ut$ when elements of $\set{\badmatch{p,n,m} \mid n,m\in\mathbb{N}}$ are disregarded\footnote{Note that some $\badmatch{p,n,m}$ might be a homomorphism from $\badp$ to $\ut$.}.

Now let us consider $\qnet$. The variables in consideration are those in $\vps$ and $v_0$. As $\exists{x}\;\net(x,x)$ does not hold in $\ut$, all of those variables need
 to be mapped to different elements of $\ut$. Moreover due to the shape of $\ut$ the variables in $\vps$ and $v_0$ have to be mapped to one of the following sets of variables $\set{\vsps \cup \set{v_{n,m}} \mid n,m \in \mathbb{N}}$.

Let us gather together what we noted so far about homomorphisms from $\qwhole$ to $\ut$:
\begin{itemize}
\item[($\diamondsuit$)] $v_0$ is always mapped to a variable of $\vps$.
\item[($\spadesuit$)] All variables $\badp$ conjuncts can only be mapped to a single $\vs{p'}$ of $\ut$ except for mappings defined by elements of $\set{\badmatch{p,n,m} \mid n,m\in\mathbb{N}}$ that are homomorphisms.
\item [($\heartsuit$)] Variables in $\vps$ and $v_0$ have to be injectively mapped to one of the following sets of variables $\set{\vsps \cup \set{v_{n,m}} \mid n,m \in \mathbb{N}}$.
\item[($\clubsuit$)] Any mapping from $\mathtt{vars}(\qwhole) \setminus \set{x_0, x_1, y_0, y_1}$ to the domain of $\ut$ that is a homomorphism can always be extended to variables $x_0, x_1, y_0$, and  $y_1$.
\end{itemize}
From the above reasoning it should be clear, that a mapping of $\mathtt{vars}(\qwhole) \setminus \set{x_0, x_1, y_0, y_1}$ to the domain of $\ut$ is a homomorphism if and only if it satisfies conditions ($\diamondsuit$), ($\spadesuit$), and ($\heartsuit$).

\medskip
Let $f \in\cols^{\mathbb{N}\times \mathbb{N}}$ define the tiling of $\ut$.

($\Rightarrow$) Assume that one of the horizontal  forbidden patterns $p = \pair{i,j,\horz} \in \forb$ appears in the tiling defined by $f$ at coordinates $\pair{n,m}$. That is $\pair{f(n,m), f(n+1, m)} = \pair{i,j}$. 
We describe a mapping of $\mathtt{vars}(\qwhole) \setminus \set{x_0, x_1, y_0, y_1}$ to the domain of $\ut$:
\begin{itemize}
\item $v_0$ is mapped to $v_p$. \hfill ($\diamondsuit$)
\item Variables of $\badp$ are mapped to $\ut$ via $\badmatch{p,n,m}$. \hfill satisfies ($\spadesuit$) from \cref{claim:third-qbad-mappings}
\item For $p' \in \forb$ s.t.  $p' \neq p$ variables of $\badpb$ are mapped to $\vs{p'}$. \hfill trivially satisfies ($\spadesuit$)
\end{itemize}

Note that ($\heartsuit$) is also satisfied. From ($\clubsuit$) we conclude that there exists a homomorphism from $\qwhole$ to $\ut$.

($\Leftarrow$) Let $h$ be  a witness homomorphism from $\qwhole$ to $\ut$. We know that $h$ satisfies conditions ($\diamondsuit$), ($\spadesuit$), and ($\heartsuit$). Moreover there are only $|\forb|$  $Slot$ substructures in $\ut$, there
are  $|\forb|$  $\badp$ conjuncts in $\ut$,  and $h$ maps $v_0$ to one of the elements in $\vsps$. Thus we can conclude that one of the 
$\badp$ is mapped via $\badmatch{p,n,m}$ to $\ut$. Then from \cref{claim:third-qbad-mappings} we observe that  $\ut$ encodes an invalid tiling.

\end{proof}
\setlength{\parindent}{\tindent}

\subsection{Rewritability for entailment problems with full TGDs, via inverse rules and forward-backward}

In this appendix we will give an additional rewritability result for full TGDs: with CQ views and Datalog queries, monotonically
determined queries have Datalog rewritings. This is of independent interest, although it is a variation of a standard result
in the absence of background knowledge in the form of rules.
 We will use it in the proof of one of the undecidability results, Theorem \ref{thm:mdl-cq-fulltgd}, which claims not just undecidability
of monotonic determinacy, but undecidability of Datalog rewritability.

Recall the notion of $Q$-entailing from the body of a paper: a view instance entails $Q$ with respect
to $\views$, $\Sigma$ if for every instance satisfying $\Sigma$ whose view image contains the given view instance, $Q$ holds.

A \emph{$\views$, $\Sigma$ certain answer rewriting of $Q$} is a query $R$ on the view schema such that
$R$ holds on $\vinst$ exactly when $\vinst$ is $Q$-entailing with respect to $\Sigma$, $\views$.

Notice that a certain answer rewriting is required to simulate entailment on all instances of the view schema,
not just the ones that are view images of base instances satisfying the rules.

In particular, if  $Q$ is monotonically determined by $\views$ with respect to $\Sigma$, and $R$ is a $\views$, $\Sigma$ certain answer
rewriting of $Q$, then $R$ is also a rewriting of $Q$ over the views

\begin{theorem} \label{thm:inverse} If $Q$ is in Datalog, $\views$ are CQ views, and $\Sigma$ consists of full TGDs, then we can compute
a $\views$, $\Sigma$ certain  answer rewriting in Datalog. 
\end{theorem}

\begin{proof}
The approach is a variation of the standard ``inverse rules'' algorithm \cite{inverserules}.  $Q$-entailment corresponds
to entailment with $\Sigma$ and the ``backward view rules'' $v(\vec x) \rightarrow \phi_v(\vec x)$. In the case
of CQ views, the latter are linear source-to-target TGDs.  We can convert the existentials in the head to skolem functions,
and consider the view predicates as extensional predicates, the base predicates as intensional. Then the backward view
rules become Datalog rules with Skolem functions. Composing these with the TGDs and the rules of $Q$, we get a Datalog
program with Skolem functions. But since the  rules of $\Sigma$ and of $Q$ never introduce Skolems,
the depth of functions in Skolem terms never goes above $1$. Then the Skolem functions can be simulated with additional
predicates. See \cite{inverserules} or \cite{uspods20journal} for details on how to perform this simulation. 
\end{proof}

\begin{corollary} \label{cor:inverse} If $Q$ is in Datalog, $\views$ are CQ views, and $\Sigma$ consists of full TGDs, and $Q$ is monotonically
determined by $\views$ with respect to $\Sigma$, then there is a rewriting in Datalog.
\end{corollary}

\newcommand{\adrule}[2]{#2 &\datalogarrow \; #1}
\newcommand{\deprule}[2]{ #1 &\;\;\to\;\;{#2} } 

\subsection{Undecidability results for the case of full TGDs, proof of Theorem \ref{thm:mdl-cq-fulltgd}}

Using the forward-backward technique, we showed decidability of monotonic determinacy in cases where
the BTVIP holds and the rules, queries, and views are reasonably well-behaved -- e.g. FGDL views and $\fgtgd$ rules.
All of our undecidability results will rely on \emph{failure} of the BTVIP: we need the first steps
of Algorithm \ref{alg:query-mondet-constraints} -- unfolding the query, chasing, applying views -- to build a grid-like structure.

Recall that the chase terminates if for every finite instance
we can apply a finite number of chase steps to reach another finite instance where the rules hold.
In the body of the paper, we focused on rules where the chase is not guaranteed to terminate. The chase can fail to terminate
even for the  simplest rules used in our main undecidability results, which involve UIDs: see
Theorem~\ref{thm:undecuidcqucq}.

An obvious question is what happens if we only have terminating chase, but do not have the bounded treewidth property.
Obviously, if we are in one of the situations where we have undecidability without rules, we are still undecidable.

But can adding rules with terminating chase move us from decidability to undecidability?

The answer is yes. We illustrate this with the case where the query is in MDL and the views are CQs.
This case is decidable in the absence of rules, using the BTVIP:  see Figure  \ref{fig:nocondecide}.
Recall from the body that \emph{full TGDs} are TGDs with no existentials in the head. Clearly, applying chase
steps to a finite instance with full TGDs always leads to a finite instance where the rules hold -- that is, the
chase always terminates.

We now prove Theorem \ref{thm:mdl-cq-fulltgd}, which we recall:

\medskip

The problem of monotonic determinacy is undecidable when $Q$ ranges over MDL, $\views$ over atomic views, and $\Sigma$ over full TGDs.
Likewise the problem of Datalog rewritability, CQ rewritability, and the variants of all these problems over finite instances.

\medskip

Notice first that by Corollary \ref{cor:inverse}, monotonic determinacy is the same as Datalog rewritability for this class
of $Q, \views, \Sigma$. Notice also that since the chases of all instances involved will be finite, finite controllability of monotonic
determinacy (as well as Datalog or CQ rewritability) will be obvious.

Thus our strategy will consist of first showing a family of examples where monotonic determinacy is undecidable.
And this example will have some  additional properties that will suffice to show that \emph{whenever they are monotonically
determined, they are CQ rewritable}. From this last fact, it will follow from undecidability of  monotonic determinacy that
CQ rewritability and Datalog rewritability are undecidable.

The first property they will have is that \emph{all view predicates are unary atomic queries}. In particular,  when we chase 
the view images of approximations of the query $Q$, as in the early steps of Algorithm \ref{alg:query-mondet-constraints},
we will get a database with only unary atoms. If we turn each such  database into a Boolean query, we get an infinite disjunction
of CQs with only unary atoms, and clearly any such disjunction degenerates to a UCQ. But in fact, \emph{there will only be one
possible view image}, so the corresponding query is a CQ. Thus if the query is monotonically determined, the CQ formed from
the view image will be a rewriting.

We  now give the details.

\newcommand{\mstates}{\mathcal{S}}
\newcommand{\mstate}{s}
\newcommand{\mstatestrt}{\mstate_{\mathrm{start}}}
\newcommand{\mstateend}{\mstate_{\mathrm{end}}}

\newcommand{\malphabet}{\mathcal{A}}
\newcommand{\predicate}[1]{\mathtt{#1}}
\newcommand{\apred}{\predicate{A}}
\newcommand{\bpred}{\predicate{B}}
\newcommand{\cpred}{\predicate{C}}
\newcommand{\predstart}{\predicate{First}}
\newcommand{\predend}{\predicate{Last}}
\newcommand{\predleft}{\predicate{Left}}
\newcommand{\predright}{\predicate{Right}}
\newcommand{\predblank}{\predicate{Blank}}
\newcommand{\predreach}{\predicate{Reached}}
\newcommand{\prededge}{\predicate{Succ}}
\newcommand{\predpath}{\predicate{Succ^{+}}}
\newcommand{\predzero}{\predicate{0}}
\newcommand{\predone}{\predicate{1}}
\newcommand{\dtm}{\mathcal{M}}
\newcommand{\msignature}{\sigma}
\newcommand{\cencodings}{\mathcal{C}}
\newcommand{\mcell}{CellAt_{i}}

\newcommand{\deltaleft}{\delta_{\mathrm{left}}}
\newcommand{\deltamid}{\delta_{\mathrm{mid}}}
\newcommand{\deltaright}{\delta_{\mathrm{right}}}

\subsubsection{Deterministic Turing Machines}
Throughout this section, we will utilize a deterministic Turing Machine (TM) variant operating on a finite tape. Formally, we define a TM as a tuple $\pair{\malphabet, \predblank, \predleft, \predright, \mstates, \mstatestrt, \mstateend, \Delta}$, consisting of the tape alphabet $\malphabet$ with a distinguished symbol $\predblank$ and symbols $\predleft$ and $\predright$, signifying the left and right ends of the tape, respectively. The set of states $\mstates$ includes the distinguished starting and halting states $\mstatestrt$ and $\mstateend$, along with the transition function $\Delta$. We assume that a $\dtm$ is not allowed to traverse beyond the tape or change the tape cells marked with $\predleft$ and $\predright$, and that in a state $\mstateend$, it halts the execution.

\smallskip
\noindent
\textbf{Halting problem.}
Given a tape, we say that it is {\em initial} when its first cell is filled with $\predleft$, the last is filled with $\predright$ and the rest are filled with $\predblank$.
The problem of determining whether there exists $i$ such that a given TM $\dtm$ halts when starting from the state $\mstatestrt$ with its head over the leftmost tape cell of the initial tape of length $i$ blank-filled tape tape of length $i$ with its edges marked by $\predleft$ and $\predright$
is undecidable. We say that TM $\dtm$ {\em halts} if there exists such natural $i$.

\smallskip
\noindent
\textbf{Relational encoding.}
In a later construction, we will encode the above-defined machine in a relational structure. To this end, we define a set of binary 
predicates representing cell states of the tape over time. Let $\cencodings$ be the set $\malphabet \cup \mstates \times \malphabet$, where each of its elements is treated as a binary predicate symbol. 
Here, the elements of $\malphabet$ represent cells without the head of the TM, and $\mstates \times \malphabet$ is used to 
represent a cell content with a head over it. We intend to use the first position of each symbol 
in $\cencodings$ to represent a position on the tape (space), and the second to represent a step number of the TM's execution (time). 
Thus, $\apred(s, t)$, for $\apred \in \cencodings$, should be read as `'the $s^{th}$ cell at time $t$ contains $\apred$.'' We will consistently use variables $s$, and $t$ to denote ``space and time coordinates''.

\smallskip
\noindent
\textbf{Exploiting determinism.}
Given a TM $\dtm$ and an initial tape $T$ of length $i$, let $\mcell(s, t): \mathbb{N} \times \mathbb{N} \to \cencodings$ be a partial binary 
function returning the contents of the $s^{th}$ tape cell in the time $t$ of the run of $\dtm$ over $T$. Note that there is a functional dependence between:
\begin{align*}
\mcell(1, t),\; \mcell(2, t) \quad\text{and}& \quad \mcell(1, t + 1), \tag*{for any $t$}\\
\mcell(s - 1, t),\; \mcell(s, t),\; \mcell(s + 1, t) \quad\text{and}&\quad \mcell(s, t + 1), \tag*{for any $t$ and $1 < s < i$} \\
\mcell(i-1, t),\; \mcell(i, t) \quad\text{and}&\quad \mcell(s, t + 1). \tag*{for any $t$}
\end{align*} 
as TM is deterministic. Let functions 
$\deltaleft: \cencodings \times \cencodings \to \cencodings$,
$\deltamid: \cencodings \times \cencodings \times \cencodings \to \cencodings$, and 
$\deltaright: \cencodings \times \cencodings \to \cencodings$
represent those dependencies. Note that $\deltaleft$ and $\deltaright$ are not constant, as the head might appear over the edges of the tape.

\subsubsection{The reduction}
In this section we show the reduction from the non-halting problem to monotonic determinacy problem from \cref{thm:mdl-cq-fulltgd}.
From now on fix a DTM $\dtm = \pair{\malphabet, \predblank, \predleft, \predright, \mstates, \mstatestrt, \mstateend, \Delta}$.
We use the signature $\msignature$ consisting of a set of unary symbols $\set{\predstart, \predend}$ and a set of binary symbols $\set{\prededge, \predpath} \cup \cencodings$.

%
%
\noindent
\textbf{Query.}
Let $Q$ be the following Boolean MDL query:
\begin{align*}
	\adrule{\predreach(x)\wedge \;\predend(x)}{\predicate{Goal}}\\
	\adrule{\predstart(x)}{\predreach(x)}\\
	\adrule{\predreach(x)\wedge \;\prededge(x,y)}{\predreach(y)}
\end{align*}

%
%
\noindent
\textbf{Views.}
Let $V$ be the following set of unary atomic views:
\begin{align*}
	V_{\predstart}(x) &= \predstart(x)\\
	V_{\predend}(x) &= \predend(x)
\end{align*}

\newcommand{\sigmapath}{\Sigma_{\predpath}}
\newcommand{\sigmasetup}{\Sigma_{\mathrm{setup}}}
\newcommand{\sigmatm}{\Sigma_{\Delta}}
\newcommand{\sigmabad}{\Sigma_{\mathrm{bad}}}
%
%
\noindent
\textbf{Constraints.}
Let $\Sigma$ be the set of full TGDs formed by the union of four disjoint below-defined sets $\sigmapath, \sigmasetup, \sigmatm \sigmabad$. As noted before, we will use variables $s$, and $t$ to denote ``space and time coordinates''. We allow for a few comments in the definition to explain their intended meanings. Finally, the definition order reflects the execution order during the chase. 

$\sigmapath$ computes transitive closure of $\prededge$:
\begin{align*}
\deprule{\prededge(x, y)}{\predpath(x, y)}\\
\deprule{\predpath(x, y), \prededge(y,z)}{\predpath(x, z)}\\
\end{align*}

$\sigmasetup$ sets up the relational encoding of the initial tape (binary atoms in first two TGDs are necessary for \cref{prop:mdlcqfull-second-chase-void}):
\begin{align*}
\deprule{\predstart(s)\wedge \predstart(t) \wedge \prededge(t, t')}{\pair{\mstatestrt, \predleft}(s,t)}\tag{the first cell contains $\predleft$ and the head}\\
\deprule{\predend(s)\wedge \predstart(t)\wedge \prededge(t, t')}{\predright(s,t)}\tag{the last cell contains $\predright$}\\
\deprule{\predstart(s)\wedge \predpath(s,s')\wedge \predpath(s',s'')\wedge \predend(s'')\wedge \predstart(t)}{\predblank(s',t)}\tag{every cell between the first and the last contains $\predblank$}
\end{align*}

$\sigmatm$ simulates $\deltaleft, \deltamid,$ and $\deltaright$ and contains:
\begin{align*}
\deprule{\predstart(s) \wedge \prededge(s, s') \wedge \apred(s,t)\wedge \bpred(s', t)\wedge  \prededge(t, t')}{\deltaleft(\apred, \bpred)(s, t')}\tag{for every $(\apred, \bpred) \in dom(\deltaleft)$}\\
\deprule{\prededge(s, s')\wedge \prededge(s', s')\wedge \apred(s,t)\wedge \bpred(s', t)\wedge \cpred(s'', t)\wedge  \prededge(t, t')}{\deltamid(\apred, \bpred, \cpred)(s', t')}\tag{for every $(\apred, \bpred, \cpred) \in dom(\deltamid)$}\\
\deprule{\prededge(s,s')\wedge\predend(s')\wedge \apred(s,t)\wedge \bpred(s', t) \wedge \prededge(t, t')}{\deltaright(\apred, \bpred)(s', t')}\tag{for every $(\apred, \bpred) \in dom(\deltaright)$}\\
\end{align*} 

$\sigmabad$ contains, for every $\predicate{S} \in \cencodings$ that encodes a tape cell containing the head of $\dtm$ in a non-halting state:
\begin{align*}
\deprule{\predicate{S}(c, t)\wedge \predend(t)}{\predstart(t)}
\end{align*}

\noindent
\textbf{Reduction - proof of correctness.}\\
Note that  $Q$ is a very simple MDL query. We have:
\begin{proposition}
The set of CQ approximations of $Q$ consists of the following instances $Q_i$, for every $i \in \mathbb{N}$:
$$	\predstart(a_1), \prededge(a_1, a_2), \ldots, \prededge(a_{n-1}, a_n), \predend(a_n).$$
\end{proposition}

As the case for $Q_0$ is always trivial ($Q_0 = \set{\predstart(x_1), \predend(x_1)}$) - the set of views determines the query under the rules, from now on we shall consider only cases for $Q_i$ with $i \geq 1$.

Consider \cref{alg:query-mondet-constraints}, and let $Q'_i$ be the instance $\textsc{BackV}_{V}(V(\textsc{Chase}_{\Sigma}(Q_i)))$. Note $Q'_i$ is a single instance as $\textsc{BackV}_{V}$ is deterministic for CQ views. Now, the intermediate goal is to
connect satisfaction of $Q$ in $Q'_i$ with 
a non-halting state appearing in a precise time in the chase of $Q_i$.

\begin{proposition}\label{prop:mdlcqfull-qp-only-unary}
Instance $Q'_i$ consists only of unary atoms.
\end{proposition}
\begin{proof}
Note that $V$ consists of unary and atomic views.
\end{proof}

\newcommand{\chaseqi}{\textsc{Chase}_{\Sigma}(Q_i)}
\newcommand{\chaseqpi}{\textsc{Chase}_{\Sigma}(Q'_i)}

\begin{proposition}\label{prop:mdlcqfull-second-chase-void}
Instances $Q'_i$ and $\chaseqpi$ are equal.
\end{proposition}
\begin{proof}
Note that every body of TGDs in $\Sigma$ contains at least one binary atom, then the proposition follows from \cref{prop:mdlcqfull-qp-only-unary}.
\end{proof}

The following is an immediate conclusion from the above. Note that the MDL query $Q$ can be satisfied over an instance $\inst$ containing only unary atoms if and only if $\inst \models \exists{x}\;\predstart(x) \wedge \predend(x)$.

\begin{corollary}\label{prop:mdlcqfulltgds-1}
$\chaseqi \models\exists{x}\;\predstart(x) \wedge \predend(x) \iff \chaseqpi \models Q.$ 
\end{corollary}

\begin{proposition}\label{prop:mdlcqfulltgds-2}
$\chaseqi \models\exists{x}\;\predstart(x) \wedge \predend(x) \iff \chaseqi \models \exists{s,t}\;\predicate{S}(s, t)\wedge \predend(t)$ for some $\predicate{S} \in \cencodings$ that encodes a tape cell containing the head of $\dtm$ in a non-halting state.
\end{proposition}
\begin{proof}
First, note that $Q_i \not\models\exists{x}\;\predstart(x) \wedge \predend(x)$. From this, it is enough to note that only TGDs from $\sigmabad$ can derive $\predstart$ atoms, and that no $\predend$ atom can be derived during the chase.
\end{proof}

We are now ready to connect satisfaction of $Q$  in $Q'_i$ with the Turing Machine:

\begin{corollary}\label{cor:mdlcqfull-qi-qpi-q-conn}
$\chaseqpi \models Q \iff \chaseqi \models \exists{s,t}\;\predicate{S}(s, t)\wedge \predend(t)$ for some $\predicate{S} \in \cencodings$ that encodes a tape cell containing the head of $\dtm$ in a non-halting state.
\end{corollary}
\begin{proof}
Follows directly from \cref{prop:mdlcqfulltgds-1} and \cref{prop:mdlcqfulltgds-2}
\end{proof}

Now we want to argue the chase over $Q_i$ simulates the first $i$ steps of the run of TM $\dtm$ over initial tape of length $i$.

\begin{proposition}\label{prop:mdlcqfulltgds-3}
For every natural number $i$ and every pair of natural numbers $s,t \leq i$, and every $\apred \in cencodings$ we have:
$$\mcell(s,t) = \apred \iff \apred(x_{s}, x_{t})\in \chaseqi.$$
\end{proposition}
\begin{proof}
Consider $\sigmasetup$, it should be clear that the lemma holds for time $t = 1$, that is for the initial initial tape of length $i$.

Finally, note how each application of TGDs from $\sigmatm$ simulate $\deltaleft$, $\deltamid$, and $\deltaright$. The rest of the lemma follows by a simple inductive argument.
\end{proof}

We can conclude the following  from combining \cref{cor:mdlcqfull-qi-qpi-q-conn,prop:mdlcqfulltgds-3}.
\begin{corollary}
The following are equivalent:
\begin{itemize}
	\item $\chaseqpi \models Q$, for every $i$.
	\item TM $\dtm$ does not halt. 
\end{itemize}
\end{corollary}

From this, we conclude that the monotonic decidability problem for MDL queries, unary CQ views, and full TGDs is undecidable. 

For the undecidability of CQ rewritability and Datalog rewritability, we observe that we have fulfilled the requirements 
promised before we began the construction.
The construction guarantees that whenever monotonic determinacy holds, the view image is always isomorphic to the instance $\set{V_\predstart(x), V_\predend(x), V_\predstart(y)}$. Thus monotonic determinacy is equivalent to the
existence of a CQ rewriting: the CQ obtained by turning this view image into a CQ by existentially quantifying
the variables.

This concludes the proof of \cref{thm:mdl-cq-fulltgd}.


\end{document}